\documentclass[11pt,a4paper]{article}

\usepackage[margin=1in]{geometry}
\usepackage{multirow}
\usepackage{amsmath}
\usepackage{mdframed}
\usepackage{amsthm}
\usepackage{longtable}
\usepackage{comment}
\allowdisplaybreaks
\usepackage{amsfonts}
\usepackage[normalem]{ulem}
\usepackage{amssymb}
\usepackage{bbm}
\usepackage{extarrows}
\usepackage{bm}
\usepackage{array}
\usepackage[backref,linktoc=all]{hyperref}
\hypersetup{colorlinks,linkcolor=red,citecolor=blue,urlcolor=gray}
\usepackage[capitalise]{cleveref}
\usepackage{enumitem}
\usepackage[dvipsnames]{xcolor}
\usepackage[caption=false]{subfig}

\title{Tight Revenue Gaps among Multi-Unit Mechanisms\footnote{A preliminary version of this work~\cite{JJLZ21} appears in Proceedings of the 22nd ACM Conference on Economics and Computation (EC'21). The proof for the ``Anonymous Reserve vs.\ Anonymous Pricing'' problem is omitted there yet is included in the current full version.

\phantom{${}^\ast$} {\bf Funding:} The first author is supported by NSF IIS-1838154, NSF CCF-1703925, NSF CCF-1814873 and NSF CCF-1563155. The second/fourth authors are supported by NSF CAREER award CCF-1844887. The third author is supported by Science and Technology Innovation 2030 – ``New Generation of Artificial Intelligence'' Major Project No.(2018AAA0100903), NSFC grant 61922052 and 61932002, Innovation Program of Shanghai Municipal Education Commission, Program for Innovative Research Team of Shanghai University of Finance and Economics (IRTSHUFE) and the Fundamental Research Funds for the Central Universities.}}

\author{
Yaonan Jin\thanks{Department of Computer Science, Columbia University, New York, NY 10027 USA (yj2552@columbia.edu).}
\and 
Shunhua Jiang\thanks{Department of Computer Science, Columbia University, New York, NY 10027 USA (sj3005@columbia.edu).}
\and Pinyan Lu\thanks{Institute for Theoretical Computer Science, Shanghai University of Finance and Economics, Shanghai, People's Republic of China
(lu.pinyan@mail.shufe.edu.cn).}
\and 
Hengjie Zhang\thanks{Department of Computer Science, Columbia University, New York, NY 10027 USA (hz2613@columbia.edu).}
}

\date{}

\theoremstyle{plain}
\newtheorem{theorem}{Theorem}
\newtheorem*{theorem*}{Theorem}
\newtheorem{lemma}{Lemma}
\newtheorem*{lemma*}{Lemma}
\newtheorem*{observation*}{Observation}
\newtheorem{corollary}{Corollary}
\newtheorem{fact}{Fact}
\newtheorem{claim}{Claim}

\theoremstyle{definition}

\newtheorem{example}{Example}

\theoremstyle{remark}
\newtheorem{remark}{Remark}

\crefname{theorem}{Theorem}{Theorems}
\crefname{lemma}{Lemma}{Lemmas}
\crefname{corollary}{Corollary}{Corollaries}
\crefname{fact}{Fact}{Facts}
\crefname{claim}{Claim}{Claims}
\crefname{definition}{Definition}{Definitions}
\crefname{example}{Example}{Examples}
\crefname{remark}{Remark}{Remarks}

\newcommand{\ignore}[1]{}
\usepackage[dvipsnames]{xcolor}
\usepackage{subfig}
\usepackage{graphicx}
\usepackage{tikz}
\usetikzlibrary{arrows, decorations.markings}
\usetikzlibrary{decorations.pathreplacing}

\tikzstyle{vecArrow} = [
thick,
decoration = {
	markings,
	mark = at position 1 with {
		\arrow[semithick]{open triangle 90}
	}
},
double distance = 3pt,
shorten >= 4.5pt,
preaction = {decorate},
postaction = {
	draw,
	line width = 3pt,
	white,
	shorten >= 4.5pt
}
]
\tikzstyle{innerWhite} = [
semithick,
white,
line width = 3pt,
shorten >= 3pt
]

\newcommand{\wh}{\widehat}

\newcommand{\ol}{\overline}

\renewcommand{\hat}{\wh}
\renewcommand{\bar}{\ol}

\DeclareMathOperator*{\E}{\mathbf{E}}

\renewcommand{\d}{\mathrm{d}}   
\newcommand{\lhs}{\mathrm{LHS}} 
\newcommand{\rhs}{\mathrm{RHS}} 

\newcommand{\eqdef}{\stackrel{\textrm{\tiny def}}{=}}
\newcommand{\eps}{\varepsilon}

\renewcommand{\vec}[1]{\ensuremath{\mathbf{#1}}}
\newcommand{\R}{\mathbb{R}}
\newcommand{\RP}{\R_{\geq 0}}
\newcommand{\N}{\mathbb{N}}
\newcommand{\NP}{\N_{\geq 1}}
\newcommand{\Z}{\mathbb{Z}}

\newcommand{\ar}{{\sf AR}}
\newcommand{\ap}{{\sf AP}}
\newcommand{\spm}{{\sf SPM}}
\newcommand{\opt}{{\sf OPT}}
\newcommand{\ear}{{\sf EAR}}

\newcommand{\AnonymousPricing}{\sf Anonymous Pricing}
\newcommand{\AnonymousReserve}{\sf Anonymous Reserve}
\newcommand{\MyersonAuction}{\sf Myerson Auction}
\newcommand{\SequentialPostedPricing}{\sf Sequential Posted Pricing}
\newcommand{\ExAnteRelaxation}{\sf Ex-Ante Relaxation}
\newcommand{\VCGAuction}{{\sf VCG Auction}}

\newcommand{\reg}{\textsc{Reg}}
\newcommand{\tri}{\textsc{Tri}}

\newcommand{\V}{\mathcal{V}}
\newcommand{\Q}{\mathcal{Q}}

\begin{document}
\maketitle


\begin{abstract}
This paper considers Bayesian revenue maximization in the $k$-unit setting, where a monopolist seller has $k$ copies of an indivisible item and faces $n$ unit-demand buyers (whose value distributions can be non-identical). Four basic mechanisms among others have been widely employed in practice and widely studied in the literature: {\sf Myerson Auction}, {\sf Sequential Posted-Pricing}, {\sf $(k + 1)$-th Price Auction with Anonymous Reserve}, and {\sf Anonymous Pricing}. Regarding a pair of mechanisms, we investigate the largest possible ratio between the two revenues (a.k.a.\ the revenue gap), over all possible value distributions of the buyers.

Divide these four mechanisms into two groups: (i)~the discriminating mechanism group, {\sf Myerson Auction} and {\sf Sequential Posted-Pricing}, and (ii)~the anonymous mechanism group, {\sf Anonymous Reserve} and {\sf Anonymous Pricing}. Within one group, the involved two mechanisms have an asymptotically tight revenue gap of $1 + \Theta(1 / \sqrt{k})$. In contrast, any two mechanisms from the different groups have an asymptotically tight revenue gap of $\Theta(\log k)$.
\end{abstract}






\thispagestyle{empty}
\newpage

\setcounter{page}{1}
\newpage
\section{Introduction}
\label{sec:intro}


``Simple vs.\ optimal'' is one of the central themes in Bayesian mechanism design. The revenue-optimal mechanisms are more of theoretical significance, but usually are complicated and are hard to implement. On the other hand, most daily-life mechanisms are much simpler, although sacrificing a (small) amount of revenue. This trade-off motivates the study on how well simple mechanisms can approximate the optimal mechanisms.

For example, consider selling a number of identical copies of a product on Amazon. (This scenario is characterized by the {\em multi-unit} model; see \cref{sec:background} for details.) The seller ideally would like to extract optimal revenues by using the remarkable Myerson's auction \cite{M81}, but often abandons it out of the following and other practical concerns:
\begin{itemize}
    \item Myerson' auction as a centralized {\em auction scheme} requires coordination between the seller and all potential buyers, which is inconvenient or even impossible in most real business.
    
    \item It individually charges the winning buyers {\em different payments} -- such price discrimination incurs fairness issues to the buyers and usually is illegal.
    
    \item It requires a complete profile of all potential buyers' value distributions. In practice, this means a full access to historical transaction records and other personal information, thus incurring privacy concerns to the buyers.
\end{itemize}
Instead, the seller simply posts a price. Each buyer on arrival makes a take-it-or-leave-it decision by himself, depending on whether the price is acceptable (and whether the product has not been sold out). This simple and prevalent mechanism, called {\em anonymous pricing}, clearly settles or at least mitigates the above issues. To justify the legitimacy of anonymous pricing, the remaining consideration is its revenue guarantees against Myerson's auction.



The last two decades have seen extensive progress on the ``simple versus optimal'' trade-off \cite[and the reference therein]{BK96, GHW01, BHW02, GHKKKM05, HR09, A14, CGL14, CGL15, FILS15, DFK16, CFHOV17, AHNPY19, JLTX19, JLQTX19a, JLQ19}. By now we can say that it constitutes a subfield within mechanism design. In this work, we will study this trade-off in the multi-unit model.

\subsection{Background}\label{sec:background}

Let us first review the previous results. In the most basic {\em single-item} model, four fundamental mechanisms among others are widely studied. Denote by $\vec{F} = \{F_{j}\}_{j \in [n]}$ the independent value distributions of buyers $j \in [n]$. These four mechanisms work as follows (see \Cref{subsec:prelim:mech} for the formal definitions).

\begin{itemize}
    \item {\AnonymousPricing} (\ap):
    This mechanism treats all buyers equally by posting a price $p$. 
    On arrival, a buyer will pay this price $p$ and take the item, when his value $b_{j} \sim F_{j}$ is higher than $p$ (and the item is still available). If the seller knows the value distributions $\{F_{j}\}_{j \in [n]}$, she would select a particular price $p$ to maximize her expected revenue among all {\AnonymousPricing} mechanisms.
    
    \item {\SequentialPostedPricing} (\spm):
    This mechanism selects an array of prices $\{p_{j}\}_{j \in [n]}$ and an ordering $\sigma: [n] \mapsto [n]$. The buyers join in the mechanism {\em sequentially} $\sigma(1), \cdots, \sigma(n)$, and each index-$\sigma(j)$ buyer must pay the {\em order-specific} price $p_{j}$ if winning. This discrimination can give better revenue than {\AnonymousPricing}.
    
    \item {\AnonymousReserve} (\ar):
    This is a variant of the {\sf Second-Price Auction}. The seller ignores the buyers whose bids $b_{j}$ are below an anonymous reserve $r$. The winner (which exists only if the highest bid $b_{(1)}$ is above the reserve $r$) is the highest of the remaining buyers, and his payment is the bigger one between the second highest bid $b_{(2)}$ and the reserve $r$.
    
    \item {\MyersonAuction} (\opt):
    A generic auction $\mathcal{A}: \{b_{j}\}_{j \in [n]} \mapsto (\vec{x}, \bm{\pi})$ is a mapping from the bids/values to the allocations $\vec{x} = (x_{j})_{j \in [n]}$ and the payments $\bm{\pi} = (\pi_{j})_{j \in [n]}$. In the single-item case, {\MyersonAuction} is the optimal one among those mappings \cite{M81}. (When the distributions $\{F_{j}\}_{j \in [n]}$ are identical, {\MyersonAuction} degenerates to {\AnonymousReserve}.)
\end{itemize}

\begin{figure}[h]
    \centering
    {\small
    \begin{tikzpicture}[scale = 2]
        \node(n1) at (2.25, 0.8) [draw, thick] {{\opt}: discriminating auction};
        \node(n2) at (0, 0) [draw, thick] {{\spm}: discriminating pricing};
        \node(n3) at (4.5, 0) [draw, thick] {{\ar}: anonymous auction};
        \node(n4) at (2.25, -0.8) [draw, thick] {{\ap}: anonymous pricing};
        \draw[thick, ->, >=triangle 45] (n1.west) to node[anchor = -15] {$[1.34, ~ 1.49]$} (n2.north);
        \draw[thick, ->, >=triangle 45] (n1.east) to node[anchor = -165] {$[2.15, ~ \mathcal{C}^*]$} (n3.north);
        \draw[thick, ->, >=triangle 45] (n1.south) to node[right] {$\mathcal{C}^* \approx 2.62$} (n4.north);
        \draw[thick, ->, >=triangle 45] (n2.south) to node[anchor = 15] {$\mathcal{C}^* \approx 2.62$} (n4.west);
        \draw[thick, ->, >=triangle 45] (n3.south) to node[anchor = 165] {$\pi^{2} / 6 \approx 1.64$} (n4.east);
    \end{tikzpicture}
    }
\caption{Demonstration for the {\bf previous results} in the {\bf single-item} setting with asymmetric regular buyers, where an interval indicates the best known lower/upper bounds, and a number indicates a tight bound. For the references of these results and further discussions, one can refer to \cite[Section~6]{JLQTX19b} and \cite[Chapter~4]{H13}.}
\label{fig:single_item}
\end{figure}
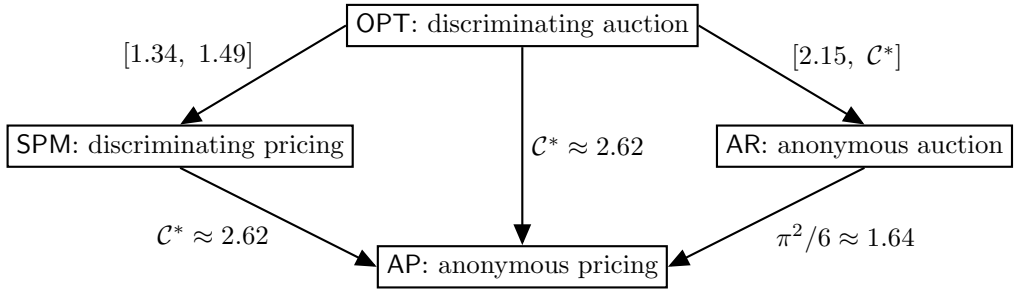

These four mechanisms together form the hierarchy in \Cref{fig:single_item}, where each arrow goes from a more complicated mechanism with higher revenue to a simpler mechanism with lower revenue. There are two notable distinctions among the four mechanisms. 
\begin{itemize}
    \item Anonymity ({\ap} and {\ar}) vs.\ Discrimination ({\spm} and {\opt}). We say a mechanism is {\em discriminating} if, when different buyers become the winner, the required payments can be different. Otherwise we say the mechanism is {\em anonymous}. Intuitively, discrimination gives a mechanism more power to extract revenue.
    
    \item Pricing ({\ap} and {\spm}) vs.\ Auction ({\ar} and {\opt}). In a pricing scheme, the buyers simply make take-it-or-leave-it decisions based on the given prices. In contrast, an auction is an arbitrary mapping from the bids to the allocations and the payments. Auctions can gain higher revenues than pricing schemes by further leveraging the competition among buyers.
\end{itemize}
Because {\spm} is a discriminating pricing scheme and {\ar} is an anonymous auction, they have different powers and are incomparable. Accordingly, there are five comparable mechanism pairs (i.e., the five arrows in \Cref{fig:single_item}).








To understand the relative powers of those mechanisms, the very first question is how large the {\em revenue gap} between any two mechanisms can be. We characterize the revenue gap as the {\em approximation ratio}\footnote{The earlier ``mechanism design for digit goods'' literature \cite{GHW01, BHW02, GHKKKM05, CGL14, CGL15}, due to technical reasons, often uses the term ``competitive ratio'' rather than ``approximation ratio''.} between the two revenues. Formally, for a more complicated mechanism $\mathcal{M}_{1}$ and a simpler mechanism $\mathcal{M}_{2}$, their approximation ratio is given by
\[
    \Re_{\mathcal{M}_1 / \mathcal{M}_2} ~ \eqdef ~ \sup \left\{\frac{\textsc{Rev}_{\mathcal{M}_{1}}(\vec{F})}{\textsc{Rev}_{\mathcal{M}_{2}}(\vec{F})} ~ \bigg| ~ \vec{F} \in \mathcal{F} \right\},
\]
where $\textsc{Rev}_{\mathcal{M}}(\vec{F})$ denotes the revenue from a mechanism $\mathcal{M}$ on an input instance $\vec{F} = \{F_{j}\}_{j \in [n]}$, and the supremum is taken over a certain family of distributions $\vec{F} \in \mathcal{F}$.



For the single-item model, the known results are shown in \Cref{fig:single_item}. Notice that all these revenue gaps are universal constants, and most of them have matching lower and upper bounds.

\vspace{.1in}
\noindent
{\bf From Single Unit to Multiple Units.}
In this work, we focus on the {\em $k$-unit setting}, where the seller has $k \geq 1$ identical copies of an item, and aims to sell them to $n$ unit-demand buyers. This setting is much more realistic and common in real business. Further, it is of intermediate complexity in comparison with the (more restricted) single-item setting and the (more general) multi-item setting.\footnote{In the $k$-unit setting, the $k$ copies are {\em identical}. But in the multi-item setting, the items can be {\em heterogeneous}.} Nonetheless, the ``simple vs.\ optimal'' trade-offs are much less understood in this setting than in the single-item setting.

Since the $k$-unit setting is still a {\em single-parameter} setting, {\MyersonAuction} remains revenue-optimal \cite{M81}. In addition, both of {\AnonymousPricing} and {\SequentialPostedPricing} can be naturally extended to this setting. For {\AnonymousReserve}, the counterpart auction is no longer ``second-price-type'', but is the {\sf $(k+1)$-th Price Auction with Anonymous Reserve}.

\subsection{An overview of our results}
\label{sec:our_results}

In the $k$-unit setting, previously only the revenue gap $\Re_{\opt/\spm}$ between {\opt} and {\spm} is well understood~\cite{Y11,A14}, whereas the other four gaps are widely open. By exploring the relative power of those mechanisms systematically, in this work we establish the (asymptotically) tight ratios of all previously unknown revenue gaps. We formalize our new results as the next two theorems and demonstrate them in \Cref{fig:result}. Here, the regularity assumption is very standard in the mechanism design literature~\cite{M81}.\footnote{Roughly speaking, this assumption means the seller extracts a higher expected revenue from a buyer via a moderate price, compared with an over-high price or an over-low price; see \cref{subsec:prelim:distr} for its definition.}

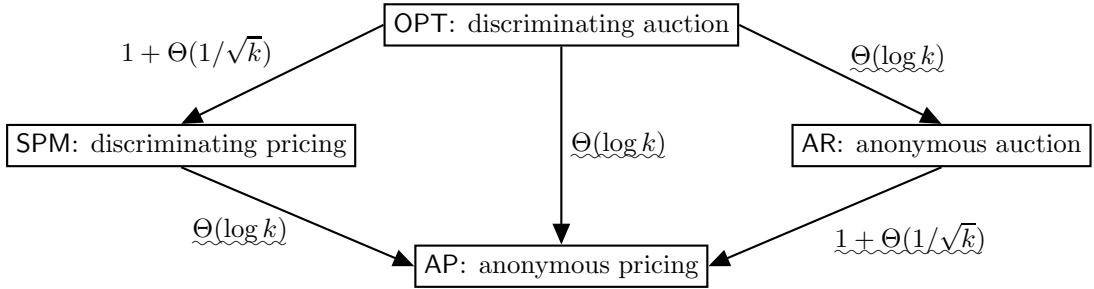
\begin{figure}[h]
\centering
    {\small
    \begin{tikzpicture}[scale = 2]
        \node(n1) at (2.5, 0.8) [draw, thick] {$\opt$: discriminating auction};
        \node(n2) at (0, 0) [draw, thick] {$\spm$: discriminating pricing};
        \node(n3) at (5, 0) [draw, thick] {$\ar$: anonymous auction};
        \node(n4) at (2.5, -0.8) [draw, thick] {$\ap$: anonymous pricing};
        \draw[thick, ->, >=triangle 45] (n1.west) to node[anchor = -15] {$1 + \Theta(1 / \sqrt{k})$} (n2.north);
        \draw[thick, ->, >=triangle 45] (n1.east) to node[anchor = -165] {\uwave{$\Theta(\log{k})$}} (n3.north);
        \draw[thick, ->, >=triangle 45] (n1.south) to node[right] {\uwave{$\Theta(\log{k})$}} (n4.north);
        \draw[thick, ->, >=triangle 45] (n2.south) to node[anchor = 15] {\uwave{$\Theta(\log{k})$}} (n4.west);
        \draw[thick, ->, >=triangle 45] (n3.south) to node[anchor = 165] {\uwave{$1 + \Theta(1 / \sqrt{k})$}} (n4.east);
    \end{tikzpicture}
    }
\caption{Demonstration for the revenue gaps among basic mechanisms in the {\bf $k$-unit} setting, given that the value distributions are regular. Our new results are \uwave{underwaved}. The $1 + \Theta(1 / \sqrt{k})$ approximation result between {\ar} and {\ap} is given in \cref{thm:intro_ar_ap}, and the other three results are given in \cref{thm:intro_ear_ap}.}
\label{fig:result}
\vspace{-.1in}
\end{figure}

\begin{theorem}[Anonymous Reserve vs.\ Anonymous Pricing]
\label{thm:intro_ar_ap}
For the unit-demand buyers $j \in [n]$, in each of the following three settings,\footnote{In the i.i.d.\ regular setting, the tight bound $1 / (1 - k^{k} / (e^{k} k!)) \approx 1 / (1 - 1 / \sqrt{2 \pi k})$ follows from a combination of (upper bound) \cite[Section~4.2]{Y11} and (lower bound) \cite[Section~4.3]{DFK16}. This ``i.i.d.\ regular'' bound and our bound $\Re_{\ar/\ap}(k)$ are different for each $k \in \NP$, but asymptotically are of the same order $= 1 + \Theta(1 / \sqrt{k})$.} the revenue gap $\Re_{\ar/\ap}(k)$ between {\AnonymousReserve} and {\AnonymousPricing} is $\Re_{\ar/\ap}(k) = 1 + k \cdot \int_{0}^{\infty} \frac{T_k(x) \cdot (1 - T_{k + 1}(x))}{(k - \sum_{i \in [k]} T_{i}(x))^{2}} \cdot \d x$, where for each $i \in [k + 1]$ the function $T_{i}(x) \eqdef e^{-x} \cdot \sum_{t \in [0: i - 1]} \frac{1}{t!} \cdot x^{t}$.
\begin{enumerate}[font = {\em \bfseries}]
    \item The asymmetric general setting, where the buyers have independent but not necessarily identical value distributions.
    
    \item The i.i.d.\ general setting, where the value distributions are identical.
    
    \item The asymmetric regular setting, where the value distributions are regular but not necessarily identical.
\end{enumerate}
Asymptotically, this bound is of order $\Re_{\ar/\ap}(k) = 1 + \Theta(1 / \sqrt{k})$.
\end{theorem}

\begin{theorem}[Discriminating Mechanisms vs.\ Anonymous Mechanisms]
\label{thm:intro_ear_ap}
When the unit-demand buyers $j \in [n]$ have independent and regular value distributions, each of the next three revenue gaps is of order $\Theta(\log k)$:
\begin{enumerate}[font = {\em \bfseries}]
    \item The revenue gap $\Re_{\opt/\ap}(k)$ between {\MyersonAuction} and {\AnonymousPricing}.
    
    \item The revenue gap $\Re_{\spm/\ap}(k)$ between {\SequentialPostedPricing} and {\AnonymousPricing}.
    
    \item The revenue gap $\Re_{\opt/\ar}(k)$ between {\MyersonAuction} and {\AnonymousReserve}.
\end{enumerate}
\end{theorem}

Similar to the {\ar} vs.\ {\ap} revenue gap, the prior works \cite{Y11,A14} show that the {\opt} vs.\ {\spm} revenue gap is also of order $1 + \Theta(1 / \sqrt{k})$. Consequently, regarding the discriminating mechanism group ({\opt} and {\spm}) and the anonymous mechanism group ({\ar} and {\ap}), each revenue gap across these two groups is $\Theta(\log{k})$, but the revenue gap between the two mechanisms in one group tends to vanish (at the rate of $1 / \sqrt{k}$) when the number of copies $k \in \mathbb{N}_{\geq 1}$ becomes large. These messages can be easily inferred from \Cref{fig:result}.

As mentioned, the revenue gaps identify the power and the limit of ``discrimination vs.\ anonymity'' and ``auction vs.\ pricing'' in revenue maximization. Different from the single-item setting, where all the revenue gaps are universal constants (see \Cref{fig:single_item}), our new results in the $k$-unit setting are more informative. When the number of copies $k \in \mathbb{N}_{\geq 1}$ is large:
\begin{itemize}
    \item Auctions are not much more helpful than pricing schemes in extracting the revenue (i.e., just an $1 + \Theta(1/\sqrt{k})$ improvement), no matter whether discrimination is allowed or not.
    
    \item Discrimination is always very useful, and can even give an unbounded improvement (up to a $\Theta(\log k)$ factor) on the revenue.
\end{itemize}
These propositions meet what we observe in real business: auctions are rarely used in practice, whereas different kinds of price discrimination are rather common.

\subsubsection{First Result: Anonymous Reserve vs.\ Anonymous Pricing}
\label{sec:intro_ar_ap}


In this section, we sketch the proof of our $1 + \Theta(1 / \sqrt{k})$ approximation result for the {\ar} vs. {\ap} revenue gap (\cref{thm:intro_ar_ap}). In fact, we can represent the exact ratio $\Re_{\ar / \ap}$ as an explicit integration formula, (although this formula in general does not admit an elementary expression). We acquire this formula by solving a mathematical programming generalized from \cite[Program~(4)]{JLTX19}, which resolves the same problem for the single-item case $k = 1$.

However, many crucial properties of the single-item case do not preserve in the general case $k \geq 1$. In the single-item case, {\AnonymousReserve} relies on the first/second order statistics $b_{(1)}$ and $b_{(2)}$ (i.e., the biggest and second biggest sampled bids/values), and {\AnonymousPricing} relies on the $b_{(1)}$. Therefore, we only need to reason about these two random variables, $b_{(1)}$ and $b_{(2)}$, together with the correlation between them. In the $k$-unit case, however, up to $(k + 1)$ random variables $b_{(1)}, \cdots, b_{(k + 1)}$ must be taken into account, and the correlation among them becomes much more complicated.

For the above reasons, we cannot modify and re-adopt the approach of the work \cite{JLTX19} in a naive way. Instead, with the purpose of handling the highly correlated order statistics $b_{(i)}$'s, we will develop a new structural lemma about the {\em Poisson binomial distributions} (PBDs). This new lemma mainly relies on the {\em log-concavity} of the PBDs.

\begin{lemma*}[Bernoulli Sum Lemma]
Given two arrays of Bernoulli random variables: $\{X_{j}\}_{j \in [n]}$ are i.i.d., while $\{Y_{j}\}_{j \in [n]}$ are independent yet not necessarily identically distributed. For the random sums $X = \sum_{j \in [n]} X_{j}$ and $Y = \sum_{j \in [n]} Y_{j}$, there exists some threshold $s \in \R$ such that:
\begin{enumerate}[font = {\em \bfseries}]
    \item $\Pr[X \leq t] \geq \Pr[Y \leq t]$ for any $t < s$.
    
    \item $\Pr[X \leq t] \leq \Pr[Y \leq t]$ for any $t \geq s$.
\end{enumerate}
\end{lemma*}

With the help of this lemma, we can characterize the worst-case instance of the mentioned mathematical programming, for $k \geq 1$ and $n \geq 1$. To this end, let us formulate the {\ar} and {\ap} revenues. Denote by $F_{j}$ the cumulative distribution function (CDF) of buyer $j$'s value, and $D_{i}$ the CDF of the $i$-th order statistic $b_{(i)}$. The {\AnonymousReserve} revenue (\cref{fact:ar_revenue}) is given by
\begin{align*}
    & \ar(r) ~ = ~ \ap(r) ~ + ~ k \cdot \int_{r}^{\infty} (1 - D_{k + 1}(x)) \cdot \d x,
    && \forall r \geq 0.
\end{align*}
where $\ap(r)$ is the revenue by posting the price $p = r$ in {\AnonymousPricing}. Further, the {\ap} revenue (\cref{fact:ap_revenue}) depends on the top-$k$ CDF's $\{D_{i}(r)\}_{i \in [k]}$ at this reserve $r \geq 0$.

Now consider a Bernoulli sum $Y = \sum_{j \in [n]} Y_{j}$, for which the individual failure probabilities are $\Pr[Y_{j} = 0] = F_{j}(r)$. This choice of the failure probabilities ensures $\Pr[Y \leq i - 1] = D_{i}(r)$ for every $i \geq 1$. Further, we can find another array of i.i.d.\ Bernoulli random variables $\{X_{j}\}_{j \in [n]}$ so that the sum $X = \sum_{j \in [n]} X_{j}$ satisfies
\[
    \Pr[X \leq k] ~ = ~ \Pr[Y \leq k] ~ = ~ D_{k + 1}(r).
\]
(The existence of such $\{X_{j}\}_{j \in [n]}$ is obvious.) Then our Bernoulli Sum Lemma shows that
\[
    \Pr[X \leq i - 1] ~ \geq ~ \Pr[Y \leq i - 1] ~ = ~ D_{i}(r)
\]
for each $i \in [k]$, where the equality holds when the $\{Y_{j}\}_{j \in [n]}$ are also i.i.d.

Informally speaking, the above inequalities and the equality condition imply that, the ratio $\ar(r) / \ap(r)$ is maximized when the value CDF's are equal $F_{1}(r) = \cdots = F_{n}(r)$ at this reserve. Following this argument and with extra efforts, we have the next observation.

\begin{observation*}
For each $k \geq 1$ and $n \geq 1$, the worst case for the $\Re_{\ar/\ap}$ revenue gap happens when the value distributions are identical, i.e., $\vec{F}^* = \{F^*\}^{n}$, (although this worst-case common distribution $F^*$ is given by an implicit equation and does not admit an elementary expression).
\end{observation*}

Furthermore, it is noteworthy that the above approach enables a unified constructive proof for the upper-bound/lower-bound parts of the general case $k \geq 1$. In contrast, the former work \cite{JLTX19} establishes these two parts of the single-item case separately, and their upper-bound proof is non-constructive.



Our Bernoulli Sum Lemma can find its applications in related directions. As mentioned, we leverage it mainly to handle the order statistics. Apart from the ``simple vs.\ optimal mechanism design'' paradigm, on other topics such as ``learning simple mechanisms from samples'' \cite{CGL15,MM16,MR16,CD17,JLX19}, the order statistics are also of fundamental interests. Conceivably, our new lemma would be helpful for those topics, in a similar manner as this paper.



\subsubsection{Second Result: Discriminating Mechanisms vs.\ Anonymous Mechanisms}
\label{sec:intro_ear_ap}

In this section we sketch the proof of \cref{thm:intro_ear_ap}, which claims that the revenue gaps $\Re_{\opt / \ap}$, $\Re_{\spm / \ap}$ and $\Re_{\opt / \ar}$ are all of order $\Theta(\log k)$. In fact, any one bound implies the other two. This is because the revenue gaps within the discriminating/anonymous groups ($\opt$ vs.\ $\spm$, and $\ar$ vs.\ $\ap$) are both {\em constants} $1 + \Theta(1 / \sqrt{k}) = \Theta(1)$, and these constants are dominated by the $\Theta(\log k)$ bound.

For these reasons, it suffices to only prove the {\opt} vs.\ {\ap} revenue gap $\Re_{\opt / \ap} = \Theta(\log k)$. Actually, an $\Omega(\log k)$ lower bound for this revenue gap is already shown in \cite[Example~5.4]{HR09}, so we only need to prove the $O(\log k)$ upper bound.




We actually prove the $O(\log k)$ upper bound between {\AnonymousPricing} and a benchmark called {\ExAnteRelaxation} ({\ear} in short). It is known that this benchmark always exceeds the {\MyersonAuction} revenue~\cite{CHMS10}. To acquire the $O(\log k)$ upper bound, we will start with a mathematical programming generalized from \cite[Equations~(1) and (2)]{AHNPY19}.

However, the general-case mathematical programming has a very different structure as it is in the single-item case. When $k = 1$, the worst-case instance (i.e., the optimal solution, see \cite[Section~4.3]{AHNPY19}) turns out to be a {\em continuum of ``small'' buyers} -- any single buyer has an infinitesimal contribution to the {\ear} benchmark, but there are infinitely many buyers $n \to \infty$ (in a sense of large markets \cite{MSVV07,AGN14}). Accordingly, it is better to think about the ``density'' of different types of buyers, instead of the number of buyers.

But in the general case, the $\Omega(\log k)$ lower-bound instance \cite[Example~5.4]{HR09} essentially is constituted by ``big'' buyers -- a certain amount of buyers contribute at least $1 / k$ unit to the {\ear} benchmark each, while every other buyer contributes strictly $0$ unit and can be omitted. More importantly (see \cref{rem:small_distributions}), if we insist on a continuum of ``small'' buyers in the general case $k \geq 1$, then the {\ear} vs.\ {\ap} revenue gap turns out to be (at most) a universal constant for whatever $k \geq 1$.

For these reasons, the current approach must be very different from the single-item case. At a high level, to handle the general case $k \geq 1$, we will classify the buyers $j \in [n]$ into groups, and then bound the individual contributions from these groups to the {\ear} benchmark.

In more details, we can employ the technique developed in \cite[Lemma~4.1]{AHNPY19}, and thus transform the mentioned mathematical programming into the following one.

\begin{mdframed}
{\bf Variables}:
\begin{itemize}
    \item $\{v_{j}\}_{j \in [n]} \in \R_{\geq 0}^{n}$, where $v_{j} = \mathop{\arg\max} \{p \cdot (1 - F_{j}(p)): p \geq 0\}$ for each $j \in [n]$, are the {\em monopoly prices} of the distributions $\vec{F} = \{F_{j}\}_{j \in [n]}$.
    
    \item $\{q_{j}\}_{j \in [n]} \in [0, 1]^{n}$, where $q_{j} = 1 - F_{j}(v_{j})$ for each $j \in [n]$, are the {\em monopoly quantiles}.
    
    \item The resulting $\{v_{j} q_{j}\}_{j \in [n]} \in \R_{\geq 0}^{n}$ are the {\em monopoly revenues}.
\end{itemize}
{\bf Constraints}:
\begin{itemize}
    \item The capacity constraint, $\sum_{j \in [n]} q_{j} \leq k$.
    
    \item The feasibility constraint, $\ap(p, \vec{F}) \leq 1$ for all $p \in \RP$.
\end{itemize}
{\bf Objective}:
Maximize the {\ExAnteRelaxation} benchmark $\ear(\vec{F}) = \sum_{j \in [n]} v_{j} q_{j}$.
\end{mdframed}

Regarding the {\ear} benchmark, the monopoly revenues $\{v_{j} q_{j}\}_{j \in [n]}$ are precisely the individual contributions from the distributions $\{F_{j}\}_{j \in [n]}$. Given the capacity constraint (in a sense of the Knapsack Problem), the monopoly quantiles $\{q_{j}\}_{j \in [n]}$ can be viewed as the individual capacities. Therefore, the monopoly prices $\{v_{j}\}_{j \in [n]}$ can be viewed as the {\em bang-per-buck ratios} (i.e., the contribution to the {\ear} benchmark per unit of the capacity).

To find the optimal solution, clearly we prefer those distributions with higher bang-per-buck ratios $\{v_{j}\}_{j \in [n]}$, but also need to take the capacities $\{q_{j}\}_{j \in [n]}$ into account. Informally, we will classify the buyers into three groups $[n] = L \cup H_{S} \cup H_{B}$:
\begin{itemize}
    \item $L = \{j \in [n]: v_{j} < 1 / k\}$. Because these group-$L$ distributions have {\em lower} bang-per-buck ratios $v_{j} < 1 / k$, conceivably the total contribution by this group to the {\ear} benchmark shall be small. Indeed, we will prove a {\em constant} upper bound $\sum_{j \in L} v_{j} q_{j} = O(1)$.
    
    \item $H_{S} = \{j \in [n]: v_{j} \geq 1 / k ~ \mbox{and} ~ v_{j} q_{j} < 1 / (2k)\}$. In other words, the group-$H_{S}$ distributions have {\em high} enough bang-per-buck ratios $v_{j} \geq 1 / k$ but {\em small} capacities, i.e., $v_{j} q_{j} < 1 / (2k)$. It turns out that the total contribution by this group is also small, and we also will prove a {\em constant} upper bound $\sum_{j \in H_{S}} v_{j} q_{j} = O(1)$.
    
    \item $H_{B} = \{j \in [n]: v_{j} \geq 1 / k ~ \mbox{and} ~ v_{j} q_{j} \geq 1 / (2k)\}$. That is, these group-$H_{B}$ distributions have {\em high} enough bang-per-buck ratios and {\em big} enough capacities. Therefore, this group should contribute the most to the {\ear} benchmark. Taking into account the feasibility constraint, $\ap(p, \vec{F}) \leq 1$ for all $p \in \RP$, we will show $\sum_{j \in H_{B}} v_{j} q_{j} = O(\log k)$.
\end{itemize}
The actual grouping criteria in our proof are more complicated than the above ones, in order to handle other technical issues.

Finally, we notice that our grouping criteria borrow ideas from the ``budget-feasible mechanism'' literature \cite{S10,CGL11,GJLZ19}, where the target is to design approximately optimal mechanisms for the Knapsack Problem under the incentive concerns. We hope that these ideas can find more applications to the ``simple vs.\ optimal mechanism design'' research topic.


\subsection{Further related works}

The revenue gaps among the mentioned mechanisms, {\MyersonAuction}, {\SequentialPostedPricing}, {\AnonymousReserve}, and {\AnonymousPricing}, are extensively studied in the literature. Below we provide an overview of the previous results (mainly in the single-item setting and in the $k$-unit settings). As a supplement, the reader can refer to the surveys~\cite{L17,CFHOV18,JLQTX19b} and the textbook~\cite{H13}.

\vspace{.1in}
\noindent
{\bf AR vs.\ AP.}
This revenue gap studies the relative power between the auction schemes and the pricing schemes, when the price discrimination is not allowed. The previously known results in the single-item case are shown in the next table.
\begin{longtable}{|c|c|c|}
    \hline
    i.i.d.\ regular & $e / (e - 1) \approx 1.58$ & \cite[Thm~6]{CHMS10} \& \cite[Thm~4.13]{H13} \\
    \hline
    i.i.d.\ general & \multirow{3}{*}{$\pi^{2} / 6 \approx 1.64$} & \multirow{3}{*}{\cite[Thm~2]{JLTX19}} \\
    \cline{1-1}
    asymmetric regular & & \\
    \cline{1-1}
    asymmetric general & & \\
    \hline
\end{longtable}
In the $k$-unit case, the tight bound $1 / (1 - k^{k} / (e^{k} k!)) \approx 1 / (1 - 1 / \sqrt{2 \pi k})$ for i.i.d.\ regular buyers is shown in \cite[Section~4.2]{Y11} and \cite[Section~4.3]{DFK16}. Our new results settle the remaining pieces of the puzzle -- even if the i.i.d.\ assumption and/or the regularity assumption are removed, this revenue gap is still of order $1 + \Theta(1 / \sqrt{k})$.

\vspace{.1in}
\noindent
{\bf SPM vs.\ AP.}
This revenue gap investigates the power of price discrimination in the pricing schemes. Below we summarize the known results and our new results, in both the single-item case and the $k$-unit case.
\begin{longtable}{|c|c|c|c|}
    \hline
    \multicolumn{2}{|c|}{single-item case} & {\spm} vs.\ {\ap} & {\opt} vs.\ {\ap} \\
    \hline
    \hline
    i.i.d.\ regular & $e / (e - 1) \approx 1.58$ & \multicolumn{2}{c|}{\cite[Thm~6]{CHMS10} \& \cite[Thm~4.13]{H13}} \\
    \hline
    i.i.d.\ general & $2 - 1 / n$ & \cite[Thm~3]{DFK16} & \cite[Thm~4.9]{H13} \\
    \hline
    asymmetric regular & constant $\mathcal{C}^* \approx 2.62$ & \cite[Thm~1]{JLTX19} & \cite[Thm~1]{JLQTX19a} \\
    \hline
    asymmetric general & $n$ & \multicolumn{2}{c|}{\cite[Prop~6.1]{AHNPY19}} \\
    \hline
    \multicolumn{4}{c}{} \\
    \hline
    \multicolumn{2}{|c|}{$k$-unit case} & {\spm} vs.\ {\ap} & {\opt} vs.\ {\ap} \\
    \hline
    \hline
    i.i.d.\ regular & $1 / (1 - k^{k} / (e^{k} k!))$ & \cite[Thm~1]{DFK16} & \cite[Sec~4.2]{Y11} \\
    \hline
    i.i.d.\ general & $2 - k / n$ & \cite[Thm~3]{DFK16} & \cite[Sec~4.5]{H13} \\
    \hline
    asymmetric regular & $\Theta(\log k)$ & \multicolumn{2}{c|}{this work} \\
    \hline
    asymmetric general & $n$ & \multicolumn{2}{c|}{\cite[Prop~6.1]{AHNPY19}} \\
    \hline
\end{longtable}

\noindent
{\bf OPT vs.\ AP.}
This revenue gap is to illustrate that even the simplest mechanism, {\AnonymousPricing}, can approximate the optimal revenue in quite general settings. Actually, in each of the single-item/$k$-unit, i.i.d./asymmetric, regular/general settings, this ratio ``coincedentally'' is equal to the {\spm} vs.\ {\ap} revenue gap, namely $\Re_{\opt / \ap} = \Re_{\spm / \ap}$.\footnote{The reader may wonder why the revenue gaps $\Re_{\opt / \ap}$ and $\Re_{\spm / \ap}$ are equal, in each of the single-item/$k$-unit, i.i.d./asymmetric, regular/general settings. This is because, in each of these settings, the worst-case instance $\{F_{j}^*\}_{j \in [n]}$ of the {\opt} vs.\ {\ap} problem has a nice property: for each $F_{j}^*$, the corresponding {\em virtual-value distribution} is supported on the non-positive semiaxis $(-\infty, 0]$ plus a {\em single positive number} $v_{j}^* > 0$. When an instance satisfies this property, we can adopt the arguments in \cite[Lemma~1]{JLTX19} to show that {\opt} and {\spm} extract the same amount of revenue, which implies $\Re_{\opt / \ap} = \Re_{\spm / \ap}$.} (But the results respectively for $\Re_{\opt / \ap}$ and $\Re_{\spm / \ap}$ are credited to different works.) For brevity, we summarize the results on the both revenue gaps together in the above tables.

Instead of the regularity assumption, the stronger {\em monotone-hazard-rate} (MHR) distributional assumption is also very standard in the mechanism design literature. The previous works \cite{GZ18,JLQ19} study the {\opt} vs.\ {\ap} revenue gap in the single-item i.i.d.\ MHR setting.

\vspace{.1in}
\noindent
{\bf OPT vs.\ AR.}
This ratio studies the power of price discrimination in the auction schemes. When the value distributions are i.i.d.\ and regular, {\MyersonAuction} and {\AnonymousReserve} turn out to be identical~\cite{M81}. The results beyond the i.i.d.\ regular case are given below.
\begin{longtable}{|c|c|c|}
    \hline
    \multicolumn{3}{|c|}{single-item case} \\
    \hline
    \hline
    i.i.d.\ general & $2 - 1 / n$ & \cite[Thm~4.9]{H13} \\
    \hline
    \multirow{2}{*}{asymmetric regular} & $\mathrm{LB} \approx 2.15$ & \cite[Sec~5]{HR09} \& \cite[Thm~3]{JLTX19} \\
    \cline{2-3}
    & $\mathrm{UB} = \mathcal{C}^* \approx 2.62$ & \cite[Sec~5]{HR09} \& \cite[Thm~1]{JLQTX19a} \\
    \hline
    asymmetric general & $n$ & \cite[Prop~6.1]{AHNPY19} \\
    \hline
    \multicolumn{3}{c}{} \\
    \hline
    \multicolumn{3}{|c|}{$k$-unit case} \\
    \hline
    \hline
    i.i.d.\ general & $2 - k / n$ & \cite[Sec~4.5]{H13} \\
    \hline
    asymmetric regular & $\Theta(\log k)$ & this work \\
    \hline
    asymmetric general & $n$ & \cite[Prop~6.1]{AHNPY19} \\
    \hline
\end{longtable}

Notably, the tight ratio in the single-item asymmetric regular setting is still unknown. Hartline and Roughgarden first prove that this ratio is between $2$ and $4$~\cite[Section~5]{HR09}. Afterwards, the lower bound is improved to $\approx 2.15$ \cite[Theorem~3]{JLTX19}. But the best known upper bound just follows from the tight {\opt} vs.\ {\ap} revenue gap $\mathcal{C}^* \approx 2.62$ by implication. We highly believe this factor-$\mathcal{C}^*$ barrier can be broken, for which new techniques tailored for {\AnonymousReserve} rather than {\AnonymousPricing} are required.



Beyond the {\AnonymousReserve} mechanism, other simple auctions with the more powerful {\em personalized reserves} are also extensively studied~\cite{HR09,BGLPS18,MS19}.

\vspace{.1in}
\noindent
{\bf OPT vs.\ SPM.}
This revenue gap investigates the relative power between the auction schemes and the pricing schemes, when the price discrimination is allowed. Indeed, the previous works \cite{HKS07,CFPV19} show that this problem is identical to the {\em ordered prophet inequality} problem in stopping theory. In each of the single-item/$k$-unit i.i.d./asymmetric settings, the tight revenue gaps under/without the regularity assumption turn out to be the same (see, e.g., \cite[Section~3.1]{Y11}). The previous results in the single-item/$k$-unit cases are summarized below.
\begin{longtable}{|c|c|c|}
    \hline
    \multicolumn{3}{|c|}{single-item case} \\
    \hline
    \hline
    i.i.d.\ & constant $\beta \approx 1.34$ & \cite[Thm~1.3]{CFHOV17} \\
    \hline
    \multirow{2}{*}{asymmetric} & $\mathrm{LB} = \beta \approx 1.34$ & \cite[Thm~1.3]{CFHOV17} \\
    \cline{2-3}
    & $\mathrm{UB} = 1 / (1 - 1 / e + 1 / 27) \approx 1.49$ & \cite[Thm~1.1]{CSZ19} \\
    \hline
    \multicolumn{3}{c}{} \\
    \hline
    \multicolumn{3}{|c|}{$k$-unit case} \\
    \hline
    \hline
    \multirow{2}{*}{i.i.d./asymmetric} & $\mathrm{LB} = 1 + \Omega(1 / \sqrt{k})$ & \cite[Thm~7]{HKS07} \\
    \cline{2-3}
    & $\mathrm{UB} = 1 / (1 - k^{k} / (e^{k} k!)) \approx 1 / (1 - 1 / \sqrt{2 \pi k})$ & \cite[Sec~4.2]{Y11} \\
    \hline
\end{longtable}
Noticeably, the tight ratio in the single-item asymmetric setting is still unknown. The best known lower bound just follows from the tight ``i.i.d.''\ revenue gap $\beta \approx 1.34$ by implication. Recently, there is an outburst of activity on the upper bound~\cite{ACK18,BGLPS18,CSZ19}, and the best known result is $1 / (1 - 1 / e + 1 / 27) \approx 1.49$ \cite[Theorem~1.1]{CSZ19}. It remains an interesting open question to further refine the upper bound.




Beyond the $k$-unit setting, the {\opt} vs.\ {\spm} revenue gap is also studied in the more general {\em matroid} setting. For this, the work \cite[Theorem~5]{CHMS10} first shows an upper bound of $2$, and then \cite[Section~4.1]{Y11} improves it to $e / (e - 1) \approx 1.58$.

The {\SequentialPostedPricing} mechanism crucially leverages the order in which the buyers participate in the mechanism. Instead, the {\em order-oblivious} counterpart mechanisms are extensively studied as well~\cite{CHMS10,A14,AWT18,ACK18,BGLPS18,EHKS18,CSZ19}.

\vspace{.1in}
\noindent
{\bf Organization.}
In \cref{sec:prelim} we introduce the notation and the requisite knowledge about the considered mechanisms. The {\AnonymousReserve} vs.\ {\AnonymousPricing} problem is investigated in \cref{sec:ar} (with some technical details deferred to \cref{sec:asymptotic_formula}). The {\ExAnteRelaxation} vs.\ {\AnonymousPricing} problem is investigated in \cref{sec:ear_ap}.

\section{Notation and Preliminaries}
\label{sec:prelim}

This section includes the notation to be adopted in this paper, and the basic knowledge about probability (e.g.\ the regular/triangle distributions) and the concerning mechanisms.

\vspace{.1in}
\noindent
{\bf Notation.}
Denote by $\RP$ (resp.\ $\NP$) the set of all non-negative real numbers (resp.\ positive integers). For any pair of integers $b \geq a \geq 0$, define the sets $[a] \eqdef \{1, 2, \cdots, a\}$ and $[a: b] \eqdef \{a, a + 1, \cdots, b\}$. Denote by $\mathbbm{1}\{\cdot\}$ the indicator function. The function $|\cdot|_+$ maps a real number $z \in \R$ to $\max\{0, z\}$.

\subsection{Probability}
\label{subsec:prelim:distr}

We use the bold letter $\vec{F} = \{F_{j}\}_{j \in [n]}$ to denote an instance (namely an $n$-dimensional {\em product distribution}), where $F_{j}$ is the bid distribution of the buyer $j \in [n]$. For ease of notation, $F_{j}$ also represents the corresponding cumulative density function (CDF).

We assume the CDF's $\{F_{j}\}_{j \in [n]}$ to be {\em left-continuous}, in the sense that when the $j$-th buyer has a random bid $b_{j} \sim F_{j}$ for a price-$p$ item, his willing-to-pay probability is $\Pr[b_{j} \geq p]$ rather than $\Pr[b_{j} > p]$. We also define the inverse CDF $F_{j}^{-1}(y) \eqdef \inf \{x \in \RP: F_{j}(x) \geq y\}$ for any $y \in [0, 1]$; notice that possibly $F_{j}^{-1}(1) = \infty$. We say a distribution $F_{j}$ {\em stochastically dominates} another $\bar{F}_{j}$, when $F_{j}(x) \leq \bar{F}_{j}(x)$ for all $x \in \RP$. Further, an instance $\vec{F} = \{F_{j}\}_{j \in [n]}$ dominates another instance $\bar{\vec{F}} = \{\bar{F}_{j}\}_{j \in [n]}$, when $F_{j}$ dominates $\bar{F}_{j}$ for each $j \in [n]$.

For a CDF $F_{j}$, we are also interested in two associated parameters $(v_{j}, q_{j})$. The {\em monopoly quantile} $q_{j} \in [0, 1]$ and the {\em monopoly price} $v_{j} \in \RP$ are respectively given by
\begin{align*}
    & q_{j} ~ \eqdef ~ \mathop{\arg\max}_{q \in [0, 1]} \{F_{j}^{-1}(1 - q) \cdot q\}
    && \mbox{and}
    && v_{j} ~ \eqdef ~ F_{j}^{-1}(1 - q_{j}).
\end{align*}
If there are multiple maximizers $q_{j}$, we would choose the smallest $q_{j}$ among the alternatives; notice that possibly $q_{j} = 0$ and $v_{j} = \infty$.

Sampling a bid profile from the instance $\vec{b} = (b_{j})_{j \in [n]} \sim \vec{F}$, the $i$-th highest bids (for $i \in [n]$) $b_{(1)} \geq \cdots \geq b_{(i)} \geq \cdots \geq b_{(n)}$ will be of particular interest. We denote by $D_{i}$ the corresponding distributions/CDF's, namely $D_{i}(x) = \Pr[b_{(i)} < x]$ for all $x \in \RP$. Again, we assume $\{D_{i}\}_{i \in [n]}$ to be left-continuous. The formulas for the $i$-th highest CDF's are given below.

\begin{fact}[Order Statistics]
\label{fac:order_CDF}
For each $i \in [n + 1]$, the $i$-th highest CDF is given by
\begin{align*}
    D_{i}(x)
    & ~ = ~ \sum_{t \in [0: i - 1]} \sum_{|W| = t} \Big(\prod_{j \notin W} \Pr[b_{j} < x]\Big) \cdot \Big(\prod_{j \in W} \Pr[b_{j} \geq x]\Big) \\
    & ~ = ~ \sum_{t \in [0: i - 1]} \sum_{|W| = t} \Big(\prod_{j \notin W} F_{j}(x)\Big) \cdot \Big(\prod_{j \in W} (1 - F_{j}(x))\Big),
    && \forall x \geq 0.
\end{align*}
\end{fact}

\noindent
{\bf Regular distribution.}
Denote by $\reg$ this distribution family. According to~\cite{M81}, a distribution is regular $F_{j} \in \reg$ if and only if the {\em virtual value function} $\varphi_j(x) \eqdef x - \frac{1 - F_{j}(x)}{f_j(x)}$ is non-decreasing on the support of $F_{j}$, where $f_{j}$ is the probability density function (PDF). Such a regular CDF $F_{j}$ is illustrated in \Cref{subfig:prelim:regular}.


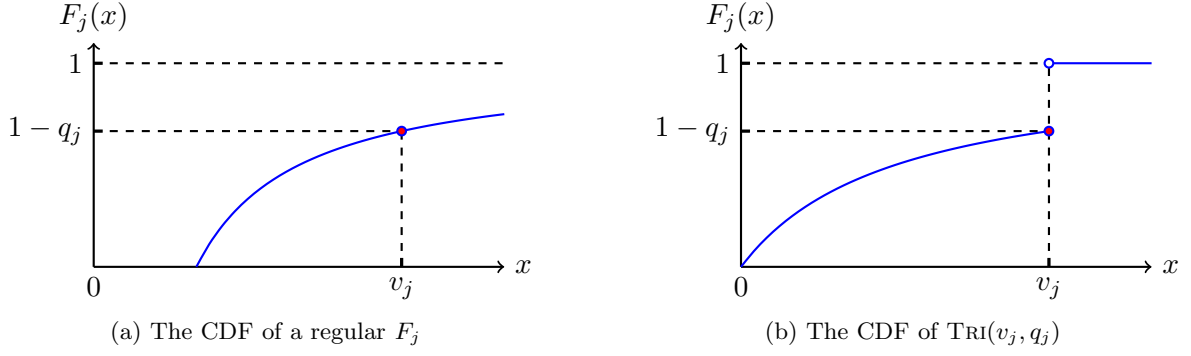
\begin{figure}
\centering
\subfloat[The CDF of a regular $F_{j}$
\label{subfig:prelim:regular}]{
\begin{tikzpicture}[thick, smooth, scale = 2.7]
\draw[->] (0, 0) -- (2, 0);
\draw[->] (0, 0) -- (0, 1.1);
\node[above] at (0, 1.1) {$F_{j}(x)$};
\node[right] at (2, 0) {$x$};
\node[below] at (1.5, 0) {$v_j$};
\node[left] at (0, 1) {$1$};
\node[below] at (0, 0) {$0$};
\node[left] at (0, 0.6667) {$1 - q_{j}$};
\draw[style = dashed] (0, 1) -- (2, 1);
\draw[style = dashed] (1.5, 0) -- (1.5, 0.6667);

\draw[color = blue, domain = 0.5: 2] plot (\x, {1 - 0.5 / \x});
\draw[style = dashed] (0, 0.6667) -- (1.5, 0.6667);
\draw[blue, fill = red] (1.5, 0.6667) circle (0.6pt);

\draw[very thick] (0pt, 1) -- (1.25pt, 1);
\draw[very thick] (0pt, 0.6667) -- (1.25pt, 0.6667);
\draw[very thick] (1.5, 0pt) -- (1.5, 1.25pt);
\end{tikzpicture}
}
\hfill
\subfloat[The CDF of $\tri(v_{j}, q_{j})$
\label{subfig:prelim:triangle}]{
\begin{tikzpicture}[thick, smooth, domain = 0: 1.5, scale = 2.7]
\draw[->] (0, 0) -- (2, 0);
\draw[->] (0, 0) -- (0, 1.1);
\node[above] at (0, 1.1) {$F_{j}(x)$};
\node[right] at (2, 0) {$x$};
\node[below] at (1.5, 0) {$v_j$};
\node[left] at (0, 1) {$1$};
\node[below] at (0, 0) {$0$};
\node[left] at (0, 0.6667) {$1 - q_{j}$};
\draw[style = dashed] (0, 1) -- (1.5, 1);
\draw[style = dashed] (1.5, 0) -- (1.5, 1);

\draw[color = blue] plot (\x, {\x / (\x + 0.75)});
\draw[style = dashed] (0, 0.6667) -- (1.5, 0.6667);
\draw[blue, fill = red] (1.5, 0.6667) circle (0.6pt);
\draw[color = blue] (1.5, 1) -- (2, 1);
\draw [color = blue, fill = white] (1.5, 1) circle(0.6pt);

\draw[very thick] (0pt, 1) -- (1.25pt, 1);
\draw[very thick] (0pt, 0.6667) -- (1.25pt, 0.6667);
\draw[very thick] (1.5, 0pt) -- (1.5, 1.25pt);
\end{tikzpicture}
}
\caption{Demonstration for the regular distribution and the triangle distribution.}
\label{fig:prelim:regular_triangle}
\end{figure}

\vspace{.1in}
\noindent
{\bf Triangle distribution.}
This distribution family, denoted by $\tri$, is introduced in \cite{AHNPY19} and is a subset of the regular distribution family $\reg$. Such a distribution $\tri(v_{j}, q_{j})$ is determined by the monopoly price $v_{j} \in \RP$ and the monopoly quantile $q_{j} \in [0,1]$. In precise, the corresponding CDF is given below and is illustrated in \Cref{subfig:prelim:triangle}.
\begin{align*}
    F_{j}(x) ~ \eqdef ~
    \begin{cases}
        \frac{(1 - q_{j}) \cdot x}{(1 - q_{j}) \cdot x + v_{j} q_{j}},
        & \forall x \in [0, v_{j}] \\
        1, & \forall x \in (v_{j}, \infty)
    \end{cases}.
\end{align*}

\subsection{Mechanisms}
\label{subsec:prelim:mech}

We focus on such a revenue maximization scenario: the seller has $k \in \NP$ {\em homogeneous} items and faces $n \geq k$ {\em unit-demand} buyers, and the buyers draw their bids $\vec{b} = \{b_{j}\}_{j \in [n]} \sim \vec{F}$ independently from a publicly known product distribution $\vec{F} = \{F_{j}\}_{j \in [n]}$. For convenience, we interchange buyer/bidder.

In the bulk of the work, we will concern three mechanisms: {\AnonymousPricing}, {\AnonymousReserve}, and {\ExAnteRelaxation}. Below we briefly introduce these mechanisms; for more details, the reader can refer to \cite[Chapter~4]{H13}.

\vspace{.1in}
\noindent
{\bf Anonymous Pricing.}
In such a mechanism, the seller posts an {\em a priori} price $p \in \RP$ to any single item; then in an arbitrary coming order, each of the first $k$ coming buyers that are willing to pay the price $p \in \RP$, will get an item by paying this price. Given any bid profile $\vec{b} \sim \vec{F}$, let $b_{(n + 1)} \eqdef 0$ and reorder the bids such that $b_{(1)} \geq \cdots \geq b_{(i)} \geq \cdots \geq b_{(n + 1)}$.

Depending on how many bids exceed the posted price, the mechanism gives a revenue of
\begin{align*}
    \textsc{Rev}(\ap)
    & ~ = ~ \sum_{i \in [k]} i \cdot p \cdot \mathbbm{1}\{b_{(i)} \geq p > b_{(i + 1)}\}
    ~ + ~ k \cdot p \cdot \mathbbm{1}\{b_{(k + 1)} \geq p\} \\
    & ~ = ~ \sum_{i \in [k]} p \cdot \mathbbm{1}\{b_{(i)} \geq p\}.
\end{align*}
Taking the randomness over $\vec{b} \sim \vec{F}$ into account results in the expected revenue.

\begin{fact}[Revenue Formula for {\AnonymousPricing}]
\label{fact:ap_revenue}
Under any posted price $p \in \RP$, the {\AnonymousPricing} mechanism extracts an expected revenue of
\[
    \ap(p, \vec{F})
    ~ \eqdef ~ p \cdot \sum_{i \in [k]} (1 - D_{i}(p)).
\]
Let $\ap(\vec{F}) \eqdef \max_{p \in \RP} \{\ap(p, \vec{F})\}$ denote the optimal {\AnonymousPricing} revenue.
\end{fact}

\noindent
{\bf Anonymous Reserve.}
In such a mechanism, the seller sets an {\em a priori} reserve $r \in \RP$ on any single item. When at most $k$ bidders are willing to pay the reserve $r \in \RP$, {\AnonymousReserve} has the same allocation/payment rule as {\AnonymousPricing}, thus the same revenue. But when at least $(k + 1)$ bidders are willing to pay this reserve, each of the top-$k$ bidders (with an arbitrary tie-breaking rule) wins an item by paying the $(k + 1)$-th highest bid $b_{(k + 1)} \geq r$.

Running on a specific bid profile $\vec{b} \sim \vec{F}$, the mechanism generates a revenue of
\begin{align*}
    \textsc{Rev}(\ar)
    & ~ = ~ \sum_{i \in [k]} i \cdot r \cdot \mathbbm{1}\{b_{(i)} \geq r > b_{(i + 1)}\}
    ~ + ~ k \cdot b_{(k + 1)} \cdot \mathbbm{1}\{b_{(k + 1)} \geq r\} \\
    & ~ = ~ \sum_{i \in [k]} r \cdot \mathbbm{1}\{b_{(i)} \geq r\}
    ~ + ~ k \cdot |b_{(k + 1)} - r|_{+}.
\end{align*}
Taking the randomness over $\vec{b} \sim \vec{F}$ into account gives the expected revenue. (Note that \cite[Fact~1]{CGM15} get the revenue formula below in the single-item case $k = 1$.)

\begin{fact}[Revenue Formula for {\AnonymousReserve} {\cite[Fact~1]{CGM15}}]
\label{fact:ar_revenue}
Under any reserve $r \in \RP$, the {\AnonymousReserve} mechanism extracts an expected revenue of
\[
    \ar(r, \vec{F})
    ~ \eqdef ~ r \cdot \sum_{i \in [k]} (1 - D_{i}(r))
    ~ + ~ k \cdot \int_{r}^{\infty} (1 - D_{k + 1}(x)) \cdot \d x.
\]
Let $\ar(\vec{F}) \eqdef \max_{r \in \RP} \{\ar(r, \vec{F})\}$ denote the optimal {\AnonymousReserve} revenue.
\end{fact}

\noindent
{\bf Sequential Posted Pricing.}
In such a mechanism, the seller sets an ordering $\sigma: [n] \mapsto [n]$ and {\em a priori} prices $\{p_{j}\}_{j \in [n]}$. The buyers come {\em sequentially} $\sigma(1), \cdots, \sigma(n)$, and each of the first $k$ coming buyers that are willing to pay individual prices $b_{\sigma(j)} \geq p_{j}$ gets an item and pays $p_{j}$.

\vspace{.1in}
\noindent
{\bf Myerson Auction.}
This mechanism ranks the buyers $\varphi_{(1)}(b_{(1)}) \geq \varphi_{(2)}(b_{(2)}) \geq \cdots \varphi_{(n)}(b_{(n)})$ in decreasing order of virtual values and allocates the items to the top-$k$ buyers that have {\em nonnegative} virtual values $\varphi_{(j)}(b_{(j)}) \geq 0$.

\vspace{.1in}
\noindent
{\bf Ex-Ante Relaxation.}
This notion is introduced by \cite{CHMS10}. Although just being a ``fake'' mechanism,\footnote{Namely, in the concerning Bayesian mechanism design setting, {\ExAnteRelaxation} is unimplementable.} {\ExAnteRelaxation} is useful to upper bound the revenue from the optimal truthful mechanism, {\MyersonAuction}.

For a regular instance, an {\ExAnteRelaxation} mechanism is specified by an allocation rule $\vec{q}' = \{q'_{j}\}_{j \in [n]} \in [0, 1]^{n}$. Here, each $q_{j} \in [0, 1]$ represents the probability that the buyer $j \in [n]$ wins an item. This allocation rule is feasible iff $\sum_{j \in [n]} q'_{j} \leq k$, because we only have $k$ items. The following fact characterizes the resulting ``revenue''.

\begin{fact}[Revenue Formula for {\ExAnteRelaxation} {\cite[Lemma~2]{CHMS10}}]
\label{fact:ear_revenue}
Given a regular instance $\vec{F} = \{F_{j}\}_{j \in [n]}$, under any feasible allocation rule $\vec{q}' = \{q'_{j}\}_{j \in [n]} \in [0, 1]^{n}$ that $\sum_{j \in [n]} q'_{j} \leq k$, the {\ExAnteRelaxation} mechanism extracts an expected revenue of
\[
    \ear(\vec{q}', \vec{F}) ~ \eqdef ~ \sum_{j \in [n]} F_{j}^{-1}(1 - q'_{j}) \cdot q'_{j}
\]
\end{fact}

\begin{remark}
We will study the {\ExAnteRelaxation} mechanism just for the regular instances. The revenue formulas for the irregular instances are more complicated, for which the reader can refer to \cite[Lemma~2]{CHMS10}.
\end{remark}

\noindent
{\bf Revenue monotonicity.}
Based on the revenue formulas given in \Cref{fact:ar_revenue,fact:ap_revenue,fact:ear_revenue}, one can easily check the following fact (a.k.a.\ the revenue monotonicity in the literature).

\begin{fact}[Revenue Monotonicity]
\label{fact:revenue_monotonicity}
Given that an instance $\vec{F} = \{F_{j}\}_{j \in [n]}$ stochastically dominates another instance $\bar{\vec{F}} = \{\bar{F}_{j}\}_{j \in [n]}$, the following hold:
\begin{enumerate}[font = {\em \bfseries}]
    \item $\ap(p, \vec{F}) \geq \ap(p, \bar{\vec{F}})$ for any posted price $p \in \RP$, and thus $\ap(\vec{F}) \geq \ap(\bar{\vec{F}})$.
    
    \item $\ar(r, \vec{F}) \geq \ar(r, \bar{\vec{F}})$ for any reserve $r \in \RP$, and thus $\ar(\vec{F}) \geq \ar(\bar{\vec{F}})$.
    
    \item $\ear(\vec{q}', \vec{F}) \geq \ear(\vec{q}', \bar{\vec{F}})$ for any allocation $\vec{q}' = \{q'_{j}\}_{j \in [n]} \in [0, 1]^{n}$ with $\sum_{j \in [n]} q'_{j} \leq k$.
\end{enumerate}
\end{fact}
\section{Anonymous Reserve vs. Anonymous Pricing}
\label{sec:ar}

In this section, we investigate the {\AnonymousReserve} vs. {\AnonymousPricing} problem. Based on the revenue formulas (see \Cref{subsec:prelim:mech}), the revenue gap between both mechanisms is characterized by the following mathematical program.
\begin{align}
    & \sup
    && \ar(r, \vec{F}) ~ = ~ r \cdot \sum_{i \in [k]} (1 - D_{i}(r)) + k \cdot \int_{r}^{\infty} (1 - D_{k + 1}(x)) \cdot \d x,
    && \forall r \in \RP,
    \label{prog:ar}\tag{P1} \\
    & \mbox{\,\,s.t.}
    && \ap(p, \vec{F}) ~ = ~ p \cdot \sum_{i \in [k]} (1 - D_{i}(p)) ~ \leq ~ 1,
    && \forall p \in \RP,
    \label{cstr:ar}\tag{C1} \\
    & && \vec{F} = \{F_{j}\}_{j \in [n]},
    && \forall n \in \NP.
    \notag
\end{align}
By finding the optimal solution to Program~\eqref{prog:ar}, we will prove the next theorem.

\begin{theorem}[{\ar} vs.\ {\ap}]
\label{thm:ar_ap:uniform}
Given that the seller has $k \in \NP$ homogeneous items and faces $n \geq k$ independent unit-demand buyers, the revenue gap $\Re_{\ar/\ap}(k, n)$ between {\AnonymousReserve} and {\AnonymousPricing} satisfies the following:
\begin{enumerate}[font = {\em \bfseries}]
    \item The revenue gap $\Re_{\ar/\ap}(k, n)$ is maximized when all the buyers have the same bid distribution $\{F^*\}^{n}$, and their common CDF $F^*$ is an implicit function given by $F^*(x) = 0$ for all $x \in [0, \frac{1}{k}]$ and $\ap(x, \{F^*\}^{n}) = 1$ for all $x \in (\frac{1}{k}, \infty)$.
    
    \item Over all $n \geq k$, the supremum revenue gap $\Re_{\ar/\ap}(k) \eqdef \sup_{n \geq k} \Re_{\ar/\ap}(k, n)$ is achieved by
    \[
        \Re_{\ar / \ap}(k, \infty) ~ = ~ 1 + k \cdot \int_{0}^{\infty} \frac{T_k(x) \cdot (1 - T_{k + 1}(x))}{(k - \sum_{i \in [k]} T_{i}(x))^{2}} \cdot \d x,
    \]
    where the functions $T_{i}(x) \eqdef e^{-x} \cdot \sum_{t \in [0: i - 1]} \frac{1}{t!} \cdot x^{t}$ for all $i \in [k + 1]$.
    
    \item For each $k \in \NP$, the supremum revenue gap is bounded between
    \[
        1 + \frac{0.1}{\sqrt{k}}
        ~ \leq ~ \Re_{\ar / \ap}(k)
        ~ \leq ~ 1 + \frac{2}{\sqrt{k}}.
    \]
    
    \item For each $k \in \NP$, the ratio $\Re_{\ar/\ap}(k)$ is tight not only in the asymmetric general setting, but (more restrictedly) also in the asymmetric regular setting and/or the i.i.d.\ general setting.
\end{enumerate}
\end{theorem}

We first outline our approach towards \cref{thm:ar_ap:uniform}. Central to the upper-bound analysis is a basic result about the sum of independent Bernoulli random variables, which is formalized below as \Cref{thm:bernoulli} and can be of independent interest (see \Cref{fig:bernoulli} for a demonstration). We will prove this theorem in \Cref{subsec:ar_ap:bernoulli}.

\begin{theorem}[Bernoulli Sum Lemma]
\label{thm:bernoulli}
For two arrays of Bernoulli random variables: $\{X_{j}\}_{j \in [n]}$ are independent and identically distributed (i.i.d.\ in short), while $\{Y_{j}\}_{j \in [n]}$ independent yet not necessarily identically distributed. For the random sums $X = \sum_{j \in [n]} X_{j}$ and $Y = \sum_{j \in [n]} Y_{j}$, there exists some threshold $s \in \R$ such that:
\begin{enumerate}[font = {\em \bfseries}]
    \item $\Pr[X \leq t] \geq \Pr[Y \leq t]$ for any $t < s$.
    
    \item $\Pr[X \leq t] \leq \Pr[Y \leq t]$ for any $t \geq s$.
\end{enumerate}
\end{theorem}

\begin{figure}
\centering
\begin{tikzpicture}[thick, smooth]
\draw[->] (0, 0) -- (10.5, 0);
\draw[->] (0, 0) -- (0, 6.5);
\node[above] at (0, 6.5) {CDF};
\node[right] at (10.5, 0) {$t$};
\node[left] at (0, 6) {$1$};
\draw[loosely dashed] (0, 6) -- (10.5, 6);

\draw[dotted, OliveGreen] (5, 0) -- (5, 6);
\node[right, OliveGreen] at (5, 1) {intersection point $s = 5$};

\node[below] at (0, 0) {$0$};
\node[below] at (1, 0) {$1$};
\node[below] at (2, 0) {$2$};
\node[below] at (3, 0) {$3$};
\node[below] at (4, 0) {$4$};
\node[below] at (5, 0) {$5$};
\node[below] at (6, 0) {$6$};
\node[below] at (7, 0) {$7$};
\node[below] at (8, 0) {$8$};
\node[below] at (9, 0) {$9$};
\node[below] at (10, 0) {$10$};

\draw[blue, fill = blue] (0, 0.000127729334370 * 6) circle (1.5pt);
\draw[blue, densely dashed] (0, 0.000127729334370 * 6) -- (1, 0.000127729334370 * 6);
\draw[blue, fill = blue] (1, 0.002780709190392 * 6) circle (1.5pt);
\draw[blue, densely dashed] (1, 0.002780709190392 * 6) -- (2, 0.002780709190392 * 6);
\draw[blue, fill = blue] (2, 0.024818622730211 * 6) circle (1.5pt);
\draw[blue, densely dashed] (2, 0.024818622730211 * 6) -- (3, 0.024818622730211 * 6);
\draw[blue, fill = blue] (3, 0.119082008961172 * 6) circle (1.5pt);
\draw[blue, densely dashed] (3, 0.119082008961172 * 6) -- (4, 0.119082008961172 * 6);
\draw[blue, fill = blue] (4, 0.343004672892540 * 6) circle (1.5pt);
\draw[blue, densely dashed] (4, 0.343004672892540 * 6) -- (5, 0.343004672892540 * 6);
\draw[blue, fill = blue] (5, 0.644307828475970 * 6) circle (1.5pt);
\draw[blue, densely dashed] (5, 0.644307828475970 * 6) -- (6, 0.644307828475970 * 6);
\draw[blue, fill = blue] (6, 0.875218065880751 * 6) circle (1.5pt);
\draw[blue, densely dashed] (6, 0.875218065880751 * 6) -- (7, 0.875218065880751 * 6);
\draw[blue, fill = blue] (7, 0.974510715353482 * 6) circle (1.5pt);
\draw[blue, densely dashed] (7, 0.974510715353482 * 6) -- (8, 0.974510715353482 * 6);
\draw[blue, fill = blue] (8, 0.997367868872672 * 6) circle (1.5pt);
\draw[blue, densely dashed] (8, 0.997367868872672 * 6) -- (9, 0.997367868872672 * 6);
\draw[blue, fill = blue] (9, 0.999895270819812 * 6) circle (1.5pt);
\draw[blue, densely dashed] (9, 0.999895270819812 * 6) -- (10, 0.999895270819812 * 6);
\draw[blue, fill = blue] (10, 1 * 6) circle (1.5pt);

\draw[red, fill = red] (0, 0.000976562500000 * 6) circle (1.5pt);
\draw[red] (0, 0.000976562500000 * 6) -- (1, 0.000976562500000 * 6);
\draw[red, fill = red] (1, 0.010742187500000 * 6) circle (1.5pt);
\draw[red] (1, 0.010742187500000 * 6) -- (2, 0.010742187500000 * 6);
\draw[red, fill = red] (2, 0.054687500000000 * 6) circle (1.5pt);
\draw[red] (2, 0.054687500000000 * 6) -- (3, 0.054687500000000 * 6);
\draw[red, fill = red] (3, 0.171875000000000 * 6) circle (1.5pt);
\draw[red] (3, 0.171875000000000 * 6) -- (4, 0.171875000000000 * 6);
\draw[red, fill = red] (4, 0.376953125000000 * 6) circle (1.5pt);
\draw[red] (4, 0.376953125000000 * 6) -- (5, 0.376953125000000 * 6);
\draw[red, fill = red] (5, 0.623046875000000 * 6) circle (1.5pt);
\draw[red] (5, 0.623046875000000 * 6) -- (6, 0.623046875000000 * 6);
\draw[red, fill = red] (6, 0.828125000000000 * 6) circle (1.5pt);
\draw[red] (6, 0.828125000000000 * 6) -- (7, 0.828125000000000 * 6);
\draw[red, fill = red] (7, 0.945312500000000 * 6) circle (1.5pt);
\draw[red] (7, 0.945312500000000 * 6) -- (8, 0.945312500000000 * 6);
\draw[red, fill = red] (8, 0.989257812500000 * 6) circle (1.5pt);
\draw[red] (8, 0.989257812500000 * 6) -- (9, 0.989257812500000 * 6);
\draw[red, fill = red] (9, 0.999023437500000 * 6) circle (1.5pt);
\draw[red] (9, 0.999023437500000 * 6) -- (10, 0.999023437500000 * 6);
\draw[red, fill = red] (10, 1 * 6) circle (1.5pt);
\end{tikzpicture}
\caption{Demonstration for \Cref{thm:bernoulli} when $n = 10$. Note that ``full line $\geq$ dashed line'' when $t < 5$ and ``full line $\leq$ dashed line'' when $t \geq 5$, where the full line refers to the sum of i.i.d.\ Bernoulli variables, and the dashed line refers to the sum of independent yet not necessarily identical Bernoulli variables.}
\label{fig:bernoulli}
\end{figure}
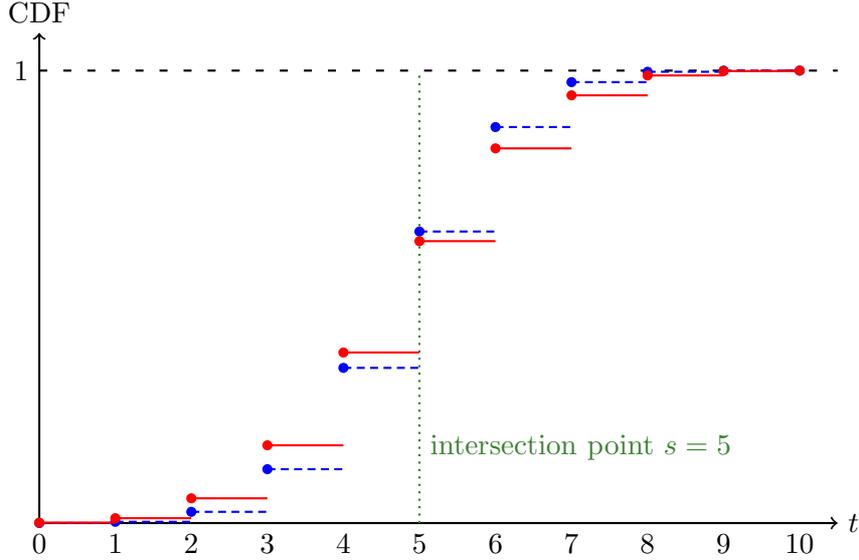

Based on \Cref{thm:bernoulli} together with further optimization arguments, we acquire Part~1 and Part~2 of \Cref{thm:ar_ap:uniform} respectively in \Cref{subsec:ar_ap:worst_case,subsec:ar_ap:gap}. Part~3 of \Cref{thm:ar_ap:uniform} requires some advanced tools from real analysis; its proof is technically involved and is deferred to \Cref{sec:asymptotic_formula}. (Notice that Parts~1 to 3 require no distributional assumption.) Eventually, we construct two matching lower-bound examples in \Cref{subsec:ar_ap:lower}, one in the i.i.d.\ general setting and one in the asymmetric regular setting, hence Part~4 of \Cref{thm:ar_ap:uniform}.

All the above results concern the {\AnonymousReserve} vs.\ {\AnonymousPricing} problem under a cardinality constraint, namely up to $k \in \NP$ buyers can simultaneously win. In the literature on mechanism design, many works also consider the more general constraint that the winning buyers satisfy a {\em matroid} constraint (e.g.\ see \cite{HR09}). For this setting, we will show in \Cref{subsec:ar_ap:matroid} an lower bound $\Omega(\log k)$ for the counterpart revenue gap.

\subsection{Bernoulli sum lemma}
\label{subsec:ar_ap:bernoulli}

The following well known result, that any independent Bernoulli sum is a {\em log-concave} random variable (e.g.\ see \cite{JKM13}), is crucial to our proof of \Cref{thm:bernoulli}.

\begin{fact}[Log-Concavity of Bernoulli Sum]
\label{fact:bernoulli}
Let $Z = \sum_{j \in [n]} Z_{j}$ be the sum of $n \in \NP$ independent yet not necessarily identical Bernoulli random variables, then for any integer $t \in \Z$,
\[
    \Pr[Z = t]^2 ~ \geq ~ \Pr[Z = t - 1] \cdot \Pr[Z = t + 1].
\]
\end{fact}

For ease of presentation, we will only justify Part~1 of \Cref{thm:bernoulli}, and Part~2 follows from similar arguments. Further, it suffices to consider the case that $\Pr[X \leq s] = \Pr[Y \leq s]$ (i.e.\ the other case that $\Pr[X \leq s] > \Pr[Y \leq s]$ can be accommodated by properly scaling the failure probability of the i.i.d.\ random variables $\{X_{j}\}_{j \in [n]}$). Further, since we concern about Bernoulli random variables and their sums, we safely assume $t \leq s$ to be integers between $[0: n]$. We obtain Part~1 of \Cref{thm:bernoulli} in two steps:
\begin{itemize}
    \item Local transformation. In \Cref{lem:bernoulli,cor:bernoulli}, we will pick a pair of non-identically distributed variables $Y_{j_{1}}$ and $Y_{j_{2}}$ (for some $j_{1} \neq j_{2} \in [n]$), and replace them by another pair of i.i.d.\ variables $\ol{Y}_{j_{1}}$ and $\ol{Y}_{j_{2}}$. The new pair is carefully constructed, so as to ensure certain properties.
    
    \item Global transformation. We conduct the local transformation on the variables $\{Y_{j}\}_{j \in [n]}$ round by round (in a nontrivial way), which preserves the mentioned properties by induction. Together with extra arguments from real analysis (see \Cref{cla:bernoulli_part1}), these properties will lead to Part~1 of \Cref{thm:bernoulli}.
\end{itemize}

\begin{lemma}[Averaging Two Variables.]
\label{lem:bernoulli}
Assume w.l.o.g.\ that two variables $Y_{j_{1}}$ and $Y_{j_{2}}$ given in \Cref{thm:bernoulli} (for some $j_{1} \neq j_{2} \in [n]$) are not identically distributed, then there exists another pair of i.i.d.\ Bernoulli random variables $\ol{Y}_{j_{1}}$ and $\ol{Y}_{j_{2}}$ such that:
\begin{enumerate}[font = {\em \bfseries}]
    \item $\Pr[\ol{Y}_{j_{1}} + \ol{Y}_{j_{2}} + \sum_{j \notin \{j_{1}, j_{2}\}} Y_{j} \leq s] = \Pr[Y \leq s]$.
    
    \item $\Pr[\ol{Y}_{j_{1}} + \ol{Y}_{j_{2}} + \sum_{j \notin \{j_{1}, j_{2}\}} Y_{j} \leq t] \geq \Pr[Y \leq t]$ for any $t \in [0: s - 1]$.
\end{enumerate}
\end{lemma}

\begin{proof}[Proof of \Cref{lem:bernoulli}]
For simplicity, we reindex the variables $\{Y_{j}\}_{j \in [n]}$ such that $j_{1}  = 1$ and $j_{2} = 2$. We adopt the following notations:
\begin{itemize}
    \item Let $q_j \eqdef \Pr[Y_{j} = 0] \in [0, 1]$ for all $j \in [n]$ and $\ol{q} \eqdef \Pr[\ol{Y}_{1} = 0] = \Pr[\ol{Y}_{2} = 0] \in [0, 1]$. W.l.o.g.\ we have $q_{1} < q_{2}$, given that $Y_{1}$ and $Y_{2}$ are not identically distributed.
    
    \item Let $Y' \eqdef \sum_{j \in [3: n]} Y_{j}$ and $a_{i} \eqdef \Pr[Y' = i]$ for all $i \in \Z$. Because $Y'$ is the sum of $(n - 2)$ Bernoulli random variables, $a_{i} \neq 0$ only if $i \in [0: n - 2]$.
\end{itemize}
It follows from \Cref{fact:bernoulli} that $a_{t}^2 \geq a_{t - 1} \cdot a_{t + 1}$ for all $t \in \Z$. By induction, one can easily see that
\begin{align}
    a_{t} \cdot a_{s - 1}
    ~ \geq ~ a_{t - 1} \cdot a_{s}
    \label{eq:log_concave}
\end{align}
for all integers $t \leq s \in \Z$, which is more convenient for our later use.

For any integer $i \in \Z$, by considering all of the four possibilities $(Y_{1}, Y_{2}) \in \{0, 1\}^{2}$, we can reformulate the probability $\Pr[Y = i] \in [0, 1]$ as follows:
\begin{align*}
    \Pr[Y = i]
    & ~ = ~ \Pr[Y_{1} + Y_{2} + Y' = i] \\
    & ~ = ~ \Pr[Y' = i] \cdot \Pr[Y_{1} = Y_{2} = 0]
    ~ + ~ \Pr[Y' = i - 2] \cdot \Pr[Y_{1} = Y_{2} = 1] \\
    & \phantom{~ = ~} ~ + ~ \Pr[Y' = i - 1] \cdot \Pr[Y_{1} = 0, Y_{2} = 1]
    ~ + ~ \Pr[Y' = i - 1] \cdot \Pr[Y_{1} = 1, = Y_{2} = 0] \\
    & ~ = ~ a_{i} \cdot q_{1} \cdot q_{2}
    ~ + ~ a_{i - 2} \cdot (1 - q_{1}) \cdot (1 - q_{2}) \\
    & \phantom{~ = ~} ~ + ~ a_{i - 1} \cdot q_{1} \cdot (1 - q_{2})
    ~ + ~ a_{i - 1} \cdot (1 - q_{1}) \cdot q_{2} \\
    & ~ = ~ ((a_{i} - a_{i - 1}) - (a_{i - 1} - a_{i - 2})) \cdot q_{1} \cdot q_{2}
    ~ + ~ (a_{i - 1} - a_{i - 2}) \cdot (q_{1} + q_{2})
    ~ + ~ a_{i - 2}.
\end{align*}
Thus, we can rewrite the telescoping sum $\Pr[Y \leq t] = \sum_{i \in [0: t]} \Pr[Y = i]$ as follows:
\begin{align}
    \Pr[Y \leq t]
    & ~ = ~ ((a_{t} - a_{t - 1}) - (a_{-1} - a_{-2})) \cdot q_{1} \cdot q_{2}
    ~ + ~ (a_{t - 1} - a_{-2}) \cdot (q_{1} + q_{2})
    ~ + ~ \sum_{i \in [0: t]} a_{i - 2}
    \notag \\
    & ~ = ~ (a_{t} - a_{t - 1}) \cdot q_{1} \cdot q_{2} ~ + ~ a_{t - 1} \cdot (q_{1} + q_{2}) ~ + ~ \mathrm{const},
    \label{eq:lem:bernoulli:1}
\end{align}
where the last step follows because $a_{-1} = a_{-2} = 0$ (recall that $a_{i} \neq 0$ only if $i \in [0: n - 2]$) and we denote $\mathrm{const} \eqdef \sum_{i \in [0: t - 2]} a_{i} = \Pr[Y' \leq t - 2]$.

We emphasize that \Cref{eq:lem:bernoulli:1} is a {\em multilinear} function of $\{q_j\}_{j \in [n]} \in [0, 1]^n$, and the last summand $\mathrm{const}$ is irrelevant to both $q_{1} = \Pr[Y_{1} = 0]$ and $q_{2} = \Pr[Y_{1} = 0]$. That is, suppose that $q_{2}$ and $\{q_j\}_{j = 3}^n$ are held constant, we can regard \Cref{eq:lem:bernoulli:1} as a linear function of $q_{1} \in [0, 1]$. Further, the corresponding slope
\begin{align*}
    (a_{t} - a_{t - 1}) \cdot q_{2} + a_{t - 1}
    ~ = ~ a_{t} \cdot q_{2} + a_{t - 1} \cdot (1 - q_{2})
\end{align*}
must be {\em non-negative}, because the probabilities $q_{2}, a_{t}, a_{t - 1} \in [0, 1]$. Similarly, when we regard \Cref{eq:lem:bernoulli:1} as a univariate of $q_{2} \in [0, 1]$, this is also a non-decreasing linear function.

Following the above arguments but considering $\ol{q} \in [0, 1]$ in place of $q_{1}$ and $q_{2}$, we also have
\begin{align}
    \Pr[\ol{Y}_{1} + \ol{Y}_{2} + Y' \leq t]
    & ~ = ~ (a_{t} - a_{t - 1}) \cdot \ol{q}^{2} ~ + ~ a_{t - 1} \cdot 2\ol{q} ~ + ~ \mathrm{const},
    \label{eq:lem:bernoulli:1.5}
\end{align}
Again, \Cref{eq:lem:bernoulli:1.5} is a non-decreasing function in $\ol{q} \in [0, 1]$. Given the monotonicity of \Cref{eq:lem:bernoulli:1,eq:lem:bernoulli:1.5} and since $q_{1} < q_{2}$, we can easily check that
\begin{align*}
    & \textnormal{\Cref{eq:lem:bernoulli:1.5}}|_{\ol{q} = q_{1}}
    ~ \leq ~ \textnormal{\Cref{eq:lem:bernoulli:1}}, \\
    & \textnormal{\Cref{eq:lem:bernoulli:1.5}}|_{\ol{q} = q_{2}}
    ~ \geq ~ \textnormal{\Cref{eq:lem:bernoulli:1}}.
\end{align*}
In the case that $t = s$, the equality $\Pr[\ol{Y}_{1} + \ol{Y}_{2} + Y' \leq s] = \Pr[Y \leq s]$ holds for at least one $\ol{q} \in [q_{1}, q_{2}] \subseteq [0, 1]$, due to the intermediate value theorem. That is,
\begin{align}
    (a_{s} - a_{s - 1}) \cdot \ol{q}^2 ~ + ~ a_{s - 1} \cdot 2\ol{q}
    ~ = ~ (a_{s} - a_{s - 1}) \cdot q_{1} \cdot q_{2} ~ + ~ a_{s - 1} \cdot (q_{1} + q_{2}).
\label{eq:lem:bernoulli:2}
\end{align}

This accomplishes Part~1 of \Cref{lem:bernoulli}. We next show that the above particular $\ol{q} \in [q_{1}, q_{2}]$ (for which $\Pr[\ol{Y}_{1} + \ol{Y}_{2} + Y' \leq s] = \Pr[Y \leq s]$) guarantees Part~2:
\begin{align*}
    \Pr[\ol{Y}_{1} + \ol{Y}_{2} + Y' \leq t]
    ~ \geq ~ \Pr[Y \leq t],
\end{align*}
for all integers $t \in [0: s - 1]$. The proof is based on case analysis.

\vspace{.1in}
\noindent
{\bf Case~I ($a_{s} = a_{s - 1}$).}
Based on \Cref{eq:log_concave}, i.e., $a_{t} \cdot a_{s - 1} \geq a_{t - 1} \cdot a_{s}$, we easily infer $a_{t} \geq a_{t - 1}$ (note that $a_{t - 1}, a_{t}, a_{s - 1}, a_{s} \in [0, 1]$ are probabilities). Moreover, \Cref{eq:lem:bernoulli:2} degenerates into $a_{s - 1} \cdot 2\ol{q} = a_{s - 1} \cdot (q_{1} + q_{2})$, by which we can safely choose $\ol{q} = \frac{1}{2} \cdot (q_{1} + q_{2})$. (Particularly, when $a_{s - 1} = 0$, the probabilities $\Pr[\ol{Y}_{1} + \ol{Y}_{2} + Y' \leq s] = \Pr[Y \leq s] = \mathrm{const}$ depend not on $\ol{q} \in [0, 1]$. That is, $\ol{q} \in [0, 1]$ can be arbitrary, and we just choose $\ol{q} = \frac{1}{2} \cdot (q_{1} + q_{2})$.) As a consequence,
\begin{align*}
    \Pr[\ol{Y}_{1} + \ol{Y}_{2} + Y' \leq t] - \Pr[Y \leq t]
    & ~ = ~ (a_{t} - a_{t - 1}) \cdot (\ol{q}^2 - q_{1} \cdot q_{2}) \\
    & ~ = ~ \frac{1}{4} \cdot (a_{t} - a_{t - 1}) \cdot (q_{1} - q_{2})^2 \\
    & ~ \geq ~ 0,
\end{align*}
where the last step follows because $a_{t} \geq a_{t - 1}$.

\vspace{.1in}
\noindent
{\bf Case~II ($a_{s} \neq a_{s - 1}$).}
In this case, we can reformulate \Cref{eq:lem:bernoulli:1} as follows:
\begin{align*}
    \Pr[Y \leq t]
    & ~ = ~ \frac{a_{t} - a_{t - 1}}{a_{s} - a_{s - 1}} \cdot ((a_{s} - a_{s - 1}) \cdot q_{1} \cdot q_{2} + a_{s - 1} \cdot (q_{1} + q_{2})) \\
    & \phantom{~ = ~} ~ + ~ \Big(a_{t - 1} - \frac{a_{t} - a_{t - 1}}{a_{s} - a_{s - 1}} \cdot a_{s - 1}\Big) \cdot (q_{1} + q_{2}) ~ + ~ \mathrm{const} \\
    & ~ = ~ \frac{a_{t} - a_{t - 1}}{a_{s} - a_{s - 1}} \cdot ((a_{s} - a_{s - 1}) \cdot \ol{q}^2 + a_{s - 1} \cdot 2\ol{q}) \\
    & \phantom{~ = ~} ~ + ~ \Big(a_{t - 1} - \frac{a_{t} - a_{t - 1}}{a_{s} - a_{s - 1}} \cdot a_{s - 1}\Big) \cdot (q_{1} + q_{2}) ~ + ~ \mathrm{const} \\
    & ~ = ~ \Pr[\ol{Y}_{1} + \ol{Y}_{2} + Y' \leq t]
    - \Big(a_{t - 1} - \frac{a_{t} - a_{t - 1}}{a_{s} - a_{s - 1}} \cdot a_{s - 1}\Big) \cdot 2\ol{q} \\
    & \phantom{~ = ~} ~ + ~ \Big(a_{t - 1} - \frac{a_{t} - a_{t - 1}}{a_{s} - a_{s - 1}} \cdot a_{s - 1}\Big) \cdot (q_{1} + q_{2}) \\
    & ~ = ~ \Pr[\ol{Y}_{1} + \ol{Y}_{2} + Y' \leq t]
    - \underbrace{\big(a_{t - 1} - \frac{a_{t} - a_{t - 1}}{a_{s} - a_{s - 1}} \cdot a_{s - 1}\big) \cdot (2\ol{q} - q_{1} - q_{2})}_{\heartsuit},
\end{align*}
where the first step rearranges \Cref{eq:lem:bernoulli:1}; the second step applies \Cref{eq:lem:bernoulli:2} to the first summand; the third step applies \Cref{eq:lem:bernoulli:1.5}; and the last step is by elementary calculation.

To accomplish the lemma for a certain $t \in [0: s - 1]$, it remains to show $\heartsuit \geq 0$.

\vspace{.1in}
\noindent
{\bf Case~II.A ($a_{s} > a_{s - 1}$).}
To see $\heartsuit \geq 0$ in this case, it suffices to show the following:
\begin{align}
    a_{t - 1} & ~ \leq ~ \frac{a_{t} - a_{t - 1}}{a_{s} - a_{s - 1}} \cdot a_{s - 1},
    \label{eq:lem:bernoulli:3} \\
    \ol{q} & ~ \leq ~ \frac{1}{2} \cdot (q_{1} + q_{2}).
    \label{eq:lem:bernoulli:4}
\end{align}
Because $a_{s} > a_{s - 1}$, \Cref{eq:lem:bernoulli:3} is equivalent to $a_{t - 1} \cdot a_{s} \leq a_{t} \cdot a_{s - 1}$, i.e.\ what we have shown in \Cref{eq:log_concave}. And for \Cref{eq:lem:bernoulli:4}, assume on the opposite that $\ol{q} > \frac{1}{2} \cdot (q_{1} + q_{2})$, then
\begin{align*}
    \lhs \mbox{ of } \eqref{eq:lem:bernoulli:2}
    & ~ = ~ (a_{s} - a_{s - 1}) \cdot \ol{q}^2 ~ + ~ a_{s - 1} \cdot 2\ol{q} \\
    & ~ \geq ~ (a_{s} - a_{s - 1}) \cdot \ol{q}^2 ~ + ~ a_{s - 1} \cdot (q_{1} + q_{2}) \\
    & ~ > ~ (a_{s} - a_{s - 1}) \cdot q_{1} \cdot q_{2} ~ + ~ a_{s - 1} \cdot (q_{1} + q_{2}) \\
    & ~ = ~ \rhs \mbox{ of } \eqref{eq:lem:bernoulli:2},
\end{align*}
where the third step is {\em strict} because $a_{s} > a_{s - 1}$ and $\ol{q} > \frac{1}{2} \cdot (q_{1} + q_{2}) \geq \sqrt{q_{1} \cdot q_{2}}$. This gives a contradiction. By refuting our assumption, we confirm \Cref{eq:lem:bernoulli:4} and thus $\heartsuit \geq 0$.

\vspace{.1in}
\noindent
{\bf Case~II.B ($a_{s} < a_{s - 1}$).}
Via similar arguments as in Case~II.A, we have $a_{t - 1} \geq \frac{a_{t} - a_{t - 1}}{a_{s} - a_{s - 1}} \cdot a_{s - 1}$ and $\ol{q} \geq \frac{1}{2} \cdot (q_{1} + q_{2})$. Given these, we also have $\heartsuit \geq 0$.

Putting all the cases together accomplishes \Cref{lem:bernoulli}.
\end{proof}

In the remainder of \Cref{subsec:ar_ap:bernoulli}, we continue to adopt the notations given in the proof of \Cref{lem:bernoulli}, and introduce two more notations:
\begin{itemize}
    \item Define $\Xi_{j_{1}, j_{2}}$ as the operator specified by \Cref{lem:bernoulli}, i.e., replacing a pair of non-identical failure probabilities $q_{j_{1}} \neq q_{j_{2}}$ by another pair of identical probabilities $\{\ol{q}\}^{2}$.
    
    \item Define the operation composition $\Xi(\vec{q}) \eqdef \Xi_{1, n} \circ \cdots \circ \Xi_{1, 3} \circ \Xi_{1, 2}(\vec{q})$, i.e.\ modifying $(q_{1}, q_{2})$ first, then the new $(\ol{q}_{1}, \ol{q}_{3})$ second, and so on. Note that both $\Xi_{j_{1}, j_{2}}$ and $\Xi$ are continuous mappings from $[0, 1]^{n}$ to $[0, 1]^{n}$.
\end{itemize}
Accessing the proof of \Cref{lem:bernoulli}, we can easily conclude the following corollary.

\begin{corollary}[Averaging Two Variables.]
\label{cor:bernoulli}
Given that $q_{j_{1}} \neq q_{j_{2}}$, the operation $\Xi_{j_{1}, j_{2}}$ specified by \Cref{lem:bernoulli} guarantees the strict inequalities $\min\{q_{j_{1}}, q_{j_{2}}\} < \ol{q} < \max\{q_{j_{1}}, q_{j_{2}}\}$.
\end{corollary}

We are ready to prove Part~1 of \Cref{thm:bernoulli}, namely the existence of a desired array of i.i.d.\ Bernoulli random variables $\{X_{j}\}_{j \in [n]}$.

\begin{claim}[Part~1 of \Cref{thm:bernoulli}]
\label{cla:bernoulli_part1}
$\Pr[X \leq t] \geq \Pr[Y \leq t]$ for any $t < s$.
\end{claim}

\begin{proof}[Proof of \Cref{cla:bernoulli_part1}]
Indeed, when not all the failure probabilities $\vec{q} = \{q_j\}_{j \in [n]} \in [0, 1]^n$ of the given variables $\{Y_{j}\}_{j \in [n]}$ are the same, we can infer from \Cref{lem:bernoulli,cor:bernoulli} an iterative algorithm, that computes the common failure probability $q^* = \Pr[X_{j} = 0] \in [0, 1]$ of the identically distributed variables $\{X_{j}\}_{j \in [n]}$. This algorithm is shown in \Cref{alg:Bernoulli}.

\begin{figure}
\begin{mdframed}
{\bf Input:} failure probabilities $\vec{q} = \{q_j\}_{j \in [n]} \in [0, 1]^n$ of $\{Y_{j}\}_{j \in [n]}$.
\begin{enumerate}[leftmargin = 2em]
\item Initialize $\vec{q}^{(0)} \gets \vec{q}$.

\item {\bf for} $\tau \in \NP$ {\bf do}

\item \qquad Reindex $\vec{q}^{(\tau - 1)}$ so that $q_{1}^{(\tau - 1)} \leq q_{2}^{(\tau - 1)} \leq \cdots \leq q_{n}^{(\tau - 1)}$.

\item \qquad Update $\vec{q}^{(\tau)} \gets \Xi(\vec{q}^{(\tau - 1)})$.

\item {\bf end for}
\end{enumerate}
\end{mdframed}
\caption{An algorithm for Part~1 of \Cref{thm:bernoulli}.}
\label{alg:Bernoulli}
\end{figure}

In a specific round $\tau \in \NP$, because the interim probabilities $\vec{q}^{(\tau - 1)} = \{q_j^{(\tau - 1)}\}_{j \in [n]}$ are reindexed in increasing order, we can infer from \Cref{cor:bernoulli} (together with the definition of the operation composition $\Xi: [0, 1]^{n} \mapsto [0, 1]^{n}$) that
\begin{align*}
    & \min_{j \in [n]} q_{j}^{(\tau)} ~ \geq ~ \min_{j \in [n]} q_{j}^{(\tau - 1)}
    && \mbox{and}
    && \max_{j \in [n]} q_{j}^{(\tau)} ~ \leq ~ \max_{j \in [n]} q_{j}^{(\tau - 1)}.
\end{align*}
In particular, the second inequality above is strict, as long as not all the interim probabilities $\vec{q}^{(\tau - 1)} = \{q_j^{(\tau - 1)}\}_{j \in [n]}$ are identical.

We consider the distance $\ell(\vec{q}) \eqdef \max_{j_{1}, j_{2} \in [n]} (q_{j_{2}} - q_{j_{1}})\geq 0$; notice that this is a continuous function from $[0, 1]^{n}$ to $[0, 1]$. The above arguments ensure that in each round $\tau \in \NP$,
\begin{align}
    \ell(\vec{q}^{(\tau)})
    ~ = ~ \ell(\Xi(\vec{q}^{(\tau - 1)}))
    ~ \leq ~ \ell(\vec{q}^{(\tau - 1)}),
    \label{eq:bernoulli_part1}
\end{align}
where the inequality is strictly as long as not all $\{q_j^{(\tau - 1)}\}_{j \in [n]}$ are identical. Due to the squeeze theorem, the sequence $\{\vec{q}^{(\tau)}\}_{\tau=1}^{\infty}$ converges to some limit $\vec{q}^* = \lim_{\tau \to \infty} \vec{q}^{(\tau)} \in [0, 1]^{n}$. Further, since both $\ell$ and $\Xi$ are continuous functions, we deduce that 
\[
    \ell(\vec{q}^*)
    ~ = ~ \lim_{\tau \to \infty} \ell(\vec{q}^{(\tau)})
    ~ = ~ \lim_{\tau \to \infty} \ell(\Xi(\vec{q}^{(\tau - 1)}))
    ~ = ~ \ell(\Xi(\lim_{\tau \to \infty} \vec{q}^{(\tau - 1)}))
    ~ = ~ \ell(\Xi(\vec{q}^*)).
\]
As a result, it follows from \Cref{eq:bernoulli_part1} that all coordinates of $\vec{q}^*$ must be the same, namely $\vec{q}^* = \{q^*\}^{n}$ for some common failure probability $q^* \in [0, 1]$ of the i.i.d.\ $\{X_{j}\}_{j \in [n]}$.

We conclude with the existence of the desired i.i.d.\ Bernoulli random variables $\{X_{j}\}_{j \in [n]}$ and the sum $X = \sum_{j \in [n]} X_{j}$. In particular, applying \Cref{lem:bernoulli} over all rounds $\tau \in \NP$ gives
\begin{align*}
    & \Pr[X \leq t] ~ \geq ~ \Pr[Y \leq t],
    && \forall t \in [0: s].
\end{align*}

This completes the proof of \Cref{cla:bernoulli_part1}.
\end{proof}

\subsection{Worst-case instance}
\label{subsec:ar_ap:worst_case}

This section is to certify Part~1 of \Cref{thm:ar_ap:uniform}. Concretely, for any given $n \geq k \geq 1$, the worst-case instance $\vec{F} = \{F_{j}\}_{j \in [n]}$ of Program~\eqref{prog:ar} is achieved when the distributions $\{F_{j}\}_{j \in [n]}$ are identical and, for any posted price $p \geq 0$, make the most of the {\AnonymousPricing} revenue. We formalize this statement as the following claim.

\begin{claim}[Part~1 of \Cref{thm:ar_ap:uniform}]
\label{cla:ar_ap:uniform:1}
The revenue gap $\Re_{\ar/\ap}(k, n)$ is maximized when all the buyers have the same bid distribution $\{F^*\}^{n}$, and their common CDF $F^*$ is an implicit function given by $F^*(x) = 0$ for all $x \in [0, \frac{1}{k}]$ and $\ap(x, \{F^*\}^{n}) = 1$ for all $x \in (\frac{1}{k}, \infty)$.
\end{claim}

\begin{proof}[Proof of \Cref{cla:ar_ap:uniform:1}]
Recall Program~\eqref{prog:ar}, an implicit constraint is that the input must be {\em CDF's}, i.e.\ each $F_{j}: \RP \mapsto [0, 1]$ is a non-decreasing mapping with $F_{j}(0) = 0$ and $F_{j}(\infty) = 1$. We relax this constraint and consider all mappings from the domain $\RP$ to the codomain $[0, 1]$. In fact, the lemma holds even under this relaxation.

For simplicity, we still denote by $\vec{F} = \{F_{j}\}_{j \in [n]}$ the given mappings. Even though $\{F_{j}\}_{j \in [n]}$ may not be CDF's, we can still write down the corresponding ``$i$-th highest CDF'' $D_{i}(x)$ and the ``revenue'' formulas $\ap(p, \vec{F})$ and $\ar(r, \vec{F})$:
\begin{align}
    D_{i}(x) & ~ = ~ \sum_{t \in [0: i - 1]} \sum_{|W| = t} \Big(\prod_{j \notin W} F_{j}(x)\Big) \cdot \Big(\prod_{j \in W} (1 - F_{j}(x))\Big),
    && \forall x \geq 0, ~ i \in [n + 1];
    \notag \\
    \ap(p, \vec{F}) & ~ = ~ p \cdot \sum_{i \in [k]} (1 - D_{i}(p)),
    && \forall p \geq 0;
    \label{eq:ar_ap:formula_ap} \\
    \ar(r, \vec{F}) & ~ = ~ \ap(r, \vec{F}) ~ + ~ k \cdot \int_{r}^{\infty} (1 - D_{k + 1}(x)) \cdot \d x,
    && \forall r \geq 0.
    \notag
\end{align}
Notice that both ``revenue'' formulas satisfy the monotonicity given in \Cref{fact:revenue_monotonicity}. For a bunch of mappings $\{F_{j}\}_{j \in [n]}$ that are feasible to Program~\eqref{prog:ar}, we consider a two-step reduction:
\begin{enumerate}[label = (\roman*)]
    \item Pointwise convert $\vec{F} = \{F_{j}\}_{j \in [n]}$, according to \Cref{thm:bernoulli}, into a bunch of identical mappings $\bar{\vec{F}} = \{\bar{F}\}^{n}$.
    
    \item Pointwise scale $\bar{\vec{F}} = \{\bar{F}\}^{n}$ {\em } into another bunch of identical mappings $\vec{F}^* = \{F^*\}^{n}$, for which $\ap(p, \vec{F}^*) = 1$ for any $p \geq \frac{1}{k}$ and $\ap(p, \vec{F}^*) = k \cdot p$ for any $p < \frac{1}{k}$.
\end{enumerate}
Clearly, constraint~\eqref{cstr:ar} holds for $\vec{F}^*$. Below we show that for any reserve $r \in \RP$, the {\AnonymousReserve} revenue increases, namely $\ar(r, \vec{F}^*) \geq \ar(r, \vec{F})$.

Given any $x \geq 0$, let us consider the independent Bernoulli random variables $\{Y_{j}^{(x)}\}_{j \in [n]}$ with the failure probabilities $\Pr[Y_{j}^{(x)} = 0] \eqdef F_{j}(x)$. We denote their sum $Y^{(x)} \eqdef \sum_{j \in [n]} Y_{j}^{(x)}$. According to Part~1 of \Cref{thm:bernoulli}, there exists a particular bunch of i.i.d.\ variables $\{X_{j}^{(x)}\}_{j \in [n]}$, for which the sum $X^{(x)} \eqdef \sum_{j \in [n]} X_{j}^{(x)}$ satisfies
\begin{align*}
    \Pr[X^{(x)} \leq k]
    & ~ = ~ \Pr[Y^{(x)} \leq k]; \\
    \Pr[X^{(x)} \leq i]
    & ~ \geq ~ \Pr[Y^{(x)} \leq i],
    && \forall i \in [0: k - 1].
\end{align*}

For each $i \in [n + 1]$, one can easily see that $\Pr[Y^{(x)} \leq i - 1] = D_{i}(x)$, and we further denote $\ol{F}(x) \eqdef \Pr[X^{(x)}_{j} = 0]$ and $\ol{D}_{i}(x) \eqdef \Pr[X^{(x)} \leq i - 1]$. Take all $x \geq 0$ into account, it follows that
\begin{align*}
    \ol{D}_{k + 1}(x)
    & ~ = ~ D_{k + 1}(x),
    && \forall x \in \RP; \\
    \ol{D}_{i}(x)
    & ~ \geq ~ D_{i}(x),
    && \forall x \in \RP, i \in [k].
\end{align*}
In view of \Cref{eq:ar_ap:formula_ap}, for any price $p \in \RP$ we have
\begin{align}
    \ap(p, \bar{\vec{F}})
    ~ \leq ~ \ap(p, \vec{F})
    ~ \leq ~ \ap(p, \vec{F}^*),
    \label{eq:ar_ap:formula_ar:2}
\end{align}
where the last inequality holds by the construction of $\vec{F}^* = \{F^*\}^{n}$ given in Step~(ii).

Given \Cref{eq:ar_ap:formula_ar:2} and since either $\vec{F}^* = \{F^*\}^{n}$ or $\bar{\vec{F}} = \{\bar{F}\}^{n}$ involves {\em identical} mappings, we can infer from \Cref{fact:revenue_monotonicity} that the scaled common mapping $F^*$ pointwise dominates $\bar{F}$. In terms of the ``$(k + 1)$-th highest CDF'', we have $D_{k + 1}^*(x) \leq \bar{D}_{k + 1}(x) = D_{k + 1}(x)$ for all $x \geq 0$. We thus deduce that for any reserve $r \in \RP$,
\begin{align}
    k \cdot \int_{r}^{\infty} (1 - D_{k + 1}^*(x)) \cdot \d x
    ~ \geq ~ k \cdot \int_{r}^{\infty} (1 - D_{k + 1}(x)) \cdot \d x.
    \label{eq:ar_ap:formula_ar:3}
\end{align}

Combining \Cref{eq:ar_ap:formula_ar:2,eq:ar_ap:formula_ar:3} together, we conclude that $\vec{F}^*$ gives a better {\AnonymousReserve} revenue than $\vec{F}$: for any reserve $r \in \RP$,
\begin{align*}
    \ar(p, \vec{F}^*)
    & ~ = ~ \ap(p, \vec{F}^*) + k \cdot \int_{r}^{\infty} \big(1 - D_{k + 1}^*(x)\big) \cdot \d x \\
    & ~ \geq ~ \ap(p, \vec{F}) + k \cdot \int_{r}^{\infty} \big(1 - D_{k + 1}(x)\big) \cdot \d x \\
    & ~ = ~ \ar(p, \vec{F}).
\end{align*}

To complete the proof, it remains to show that the mapping $F^*$ is indeed a CDF, namely that $F^*$ is non-decreasing, $F^*(0) = 0$ and $F^*(\infty) = 1$. Under the construction given in Step~(ii), we know from \Cref{eq:ar_ap:formula_ap} that $\sum_{i \in [k]} D_{i}^*(x) = |k - \frac{1}{x}|_{+}$ is an increasing function. Particularly, $D_{i}^*(x) = 0$ for any $x \leq \frac{1}{k}$ and each $i \in [k]$, and $\sum_{i \in [k]} D_{i}(\infty) = k$.

Indeed, suppose we regard $q \eqdef F^*(x)$ as a single variable, then each summand
\[
    D_{i}^*(q) ~ := ~ \sum_{t = 0}^{i - 1} \binom{n}{t} \cdot q^{n - t} \cdot (1 - q)^t
\]
is an increasing function on $q \in [0, 1]$, with the minimum $D_{i}^*(q)|_{q = 0} = 0$ and the maximum $D_{i}^*(q)|_{q = 1} = 1$. To meet all the promised properties of $\sum_{i \in [k]} D_{i}^*$ (as a function of $x \geq 0$), the given $F^*$ must be a CDF, namely an increasing function supported on $x \in (\frac{1}{k}, \infty)$ so that $F^*(\frac{1}{k}) = 0$ and $F^*(\infty) = 1$.

This completes the proof of \Cref{cla:ar_ap:uniform:1}.
\end{proof}

\subsection{Supremum revenue gap}
\label{subsec:ar_ap:gap}

In \cref{subsec:ar_ap:worst_case}, we characterize the worst-case instance for any given population $n \geq k$. To avoid ambiguity, below we denote that instance by $\vec{F}_{(n)}^*$. In the next claim, we study the worst-case population and the resulting supremum revenue gap $\Re_{\ar/\ap}(k) = \sup_{n \geq k} \Re_{\ar/\ap}(k, n)$.

\begin{claim}[Part~2 of \Cref{thm:ar_ap:uniform}]
\label{cla:ar_ap:uniform:2}
Over all $n \geq k$, the supremum revenue gap $\Re_{\ar/\ap}(k) = \sup_{n \geq k} \Re_{\ar/\ap}(k, n)$ is achieved by
\begin{align*}
    \Re_{\ar / \ap}(k, \infty)
    ~ = ~ \lim_{n \to \infty} \ar(\vec{F}_{(n)}^*)
    ~ = ~ 1 + k \cdot \int_{0}^{\infty} \frac{T_k(x) \cdot (1 - T_{k + 1}(x))}{(k - \sum_{i \in [k]} T_{i}(x))^{2}} \cdot \d x,
\end{align*}
where the functions $T_{i}(x) \eqdef e^{-x} \cdot \sum_{t \in [0: i - 1]} \frac{1}{t!} \cdot x^t$ for all $i \in [k + 1]$.
\end{claim}

\begin{proof}[Proof of \Cref{cla:ar_ap:uniform:2}]
We first show that $\{\Re_{\ar/\ap}(k, n)\}_{n \geq k}$ is an {\em increasing} sequence, which by induction guarantees that $\Re_{\ar/\ap}(k) = \Re_{\ar / \ap}(k, \infty)$.

Indeed, the worst-case $n$-buyer instance $\vec{F}_{(n)}^* = \{F_{(n)}^*\}^{n}$ with the common CDF $F_{(n)}^*$ (specified by \Cref{cla:ar_ap:uniform:1}) can be regarded as such a $(n + 1)$-buyer instance: the index-$(n + 1)$ buyer has a deterministic bid of zero, while every other buyer $i \in [n]$ still has the bid CDF $F_{(n)}^*$. This $(n + 1)$-buyer instance is feasible to Program~\eqref{prog:ar}, and gives a less {\AnonymousReserve} revenue than the worst-case $(n + 1)$-buyer instance (due to \Cref{cla:ar_ap:uniform:1}). That is, the $n$-buyer and $(n + 1)$-buyer revenue gaps satisfy that $\Re_{\ar/\ap}(k, n) \leq \Re_{\ar/\ap}(k, n + 1)$, as desired.

It remains to prove the promised revenue formula for the limit instance $\vec{F}_{(\infty)}^*$. To this end, we first show the optimal {\AnonymousReserve} revenue from a specific $n$-buyer instance:
\begin{align}
    \ar(\vec{F}_{(n)}^*)
    ~ = ~ 1 ~ + ~ k \cdot \int_{1 / k}^{\infty} (1 - D_{k + 1}(x)) \cdot \d x.
    \label{eq:ar_ap:uniform:2:1}
\end{align}

Indeed, this optimal revenue can be achieved by any reserve $r \in [0, \frac{1}{k}]$:
\begin{align*}
    \ar(r, \vec{F}_{(n)}^*)
    & ~ = ~ r \cdot \sum_{i \in [k]} (1 - D_{i}(r)) ~ + ~ k \cdot \int_{r}^{\infty} (1 - D_{k + 1}(x)) \cdot \d x \\
    & ~ = ~ k \cdot r ~ + ~ k \cdot \int_{r}^{1 / k} \d x ~ + ~ k \cdot \int_{1 / k}^{\infty} (1 - D_{k + 1}(x)) \cdot \d x \\
    & ~ = ~ \rhs \mbox{ of } \eqref{eq:ar_ap:uniform:2:1},
\end{align*}
where the second step follows because every $i$-th highest CDF $D_{i}(x) = 0$ for all $x \in [0, \frac{1}{k}]$, due to \Cref{cla:ar_ap:uniform:1} that the common CDF $F_{(n)}^*$ is supported on $x \in (\frac{1}{k}, \infty)$.

Moreover, any reserve $r \in (\frac{1}{k}, \infty)$ cannot generate a higher {\AnonymousReserve} revenue:
\begin{align*}
    \ar(r, \vec{F}_{(n)}^*)
    & ~ = ~ \ap(r, \vec{F}_{(n)}^*) ~ + ~ k \cdot \int_{r}^{\infty} (1 - D_{k + 1}(x)) \cdot \d x \\
    & ~ = ~ 1 + k \cdot \int_{r}^{\infty} (1 - D_{k + 1}(x)) \cdot \d x \\
    & ~ \leq ~ \rhs \mbox{ of } \eqref{eq:ar_ap:uniform:2:1},
\end{align*}
where the second step holds since $\ap(r, \vec{F}_{(n)}^*) = 1$ (see \Cref{cla:ar_ap:uniform:1}); and the last step holds since the $(k + 1)$-th highest CDF $D_{k + 1}$ is pointwise bounded within $[0, 1]$.

Given \Cref{eq:ar_ap:uniform:2:1}, it remains to reason about the $(k + 1)$-th highest CDF $D_{k + 1}$. Below, we consider a specific bid $x \in (\frac{1}{k}, \infty)$ and, for each $n \geq k$ and all $i \in [k + 1]$, use the shorthand $F_{(n)}^* = F_{(n)}^*(x)$ and $D_{i} = D_{i}(x)$. In addition, we denote by $\hat{D}_{i} \eqdef \lim_{n \to \infty} D_{i}$ the $i$-th highest CDF resulted from the limit instance $\vec{F}_{(\infty)}^* = \lim_{n \to \infty} \vec{F}_{(n)}^*$.

It turns out that $F_{(\infty)}^* = 1$. Otherwise, any individual buyer is willing to pay with a {\em constant} probability $(1 - F_{(\infty)}^*) > 0$. This means the limit $i$-th highest CDF is $\hat{D}_{i} = \lim_{n \to \infty} D_{i} = 0$ for all $i \in [k + 1]$, since there are infinite buyers $n \to \infty$. This incurs a contradiction to constraint~\eqref{cstr:ear_ap}, namely that the {\AnonymousPricing} revenue exceeds one (note that $x > \frac{1}{k}$ is given):
\begin{align*}
    \ap(x, \vec{F}_{(\infty)}^*)
    ~ = ~ x \cdot \sum_{i \in [k]} (1 - \hat{D}_{i})
    ~ = ~ x \cdot k
    ~ > ~ 1.
\end{align*}

Given that $F_{(\infty)}^* = \lim_{n \to \infty} F_{(n)}^* = 1$, for a sufficiently large $n \geq k$ we have
\begin{align}
    n \cdot \Big(\frac{1}{F_{(n)}^*} - 1\Big)
    & ~ = ~ (1 + o_{n}(1)) \cdot n \cdot \ln\Big(1 + \frac{1}{F_{(n)}^*} - 1\Big)
    \notag \\
    & ~ = ~ -(1 + o_{n}(1)) \cdot n \cdot \ln F_{(n)}^*
    \notag \\
    & ~ = ~ -(1 + o_{n}(1)) \cdot \ln D_{1},
    \label{eq:ar_ap:uniform:2:2}
\end{align}
where the first step uses the Maclaurin series of $\ln(1 + w)$ in the neighborhood of $w = 0$; and the last step follows because the highest CDF $D_{1} = (F_{(n)}^*)^{n}$.

Based on \Cref{eq:ar_ap:uniform:2:2}, for any given $i \in [k + 1]$ and a sufficiently large $n \geq k$, we can reformulate the $i$-th highest CDF $D_{i}$ as follows:
\begin{align*}
    D_{i}
    & ~ = ~ \sum_{t \in [0: i - 1]} \binom{n}{t} \cdot (F_{(n)}^*)^{n - t} \cdot (1 - F_{(n)}^*)^{t} \\
    & ~ = ~ \sum_{t \in [0: i - 1]} \binom{n}{t} \cdot D_{1} \cdot \Big(\frac{1}{F_{(n)}^*} - 1\Big)^{t} \\
    & ~ = ~ (1 + o_{n}(1)) \cdot \sum_{t \in [0: i - 1]} \binom{n}{t} \cdot \frac{1}{n^{t}} \cdot D_{1} \cdot (-\ln D_{1})^{t} \\
    & ~ = ~ (1 + o_{n}(1)) \cdot \sum_{t \in [0: i - 1]} \frac{1}{t!} \cdot D_{1} \cdot (-\ln D_{1})^{t},
\end{align*}
where the first step applies \Cref{fac:order_CDF} (note that $\{F_{(n)}^*\}^{n}$ are i.i.d.); the second step follows since the highest CDF $D_{1} = (F_{(n)}^*)^{n}$; the third step applies \Cref{eq:ar_ap:uniform:2:2}; and the last step uses the fact that $\binom{n}{t} \cdot \frac{1}{n^{t}} = (1 + o_{n}(1)) \cdot \frac{1}{t!}$.

Following the above equation, the limit $i$-th highest CDF $\hat{D}_{i} = \lim_{n \to \infty} D_{i}$ satisfies that
\begin{align}
    \hat{D}_{i}
    & ~ = ~ \sum_{t \in [0: i - 1]} \frac{1}{t!} \cdot (\lim_{n \to \infty} D_{1}) \cdot (-\ln (\lim_{n \to \infty} D_{1}))^{t}
    \notag \\
    & ~ = ~ \sum_{t \in [0: i - 1]} \frac{1}{t!} \cdot \hat{D}_{1} \cdot (-\ln \hat{D}_{1})^{t}.
    \label{eq:ar_ap:uniform:2:3}
\end{align}
Note that this is an identity in the range $x \in (\frac{1}{k}, \infty)$. By taking the derivative, we also have
\begin{align}
    \frac{\d \hat{D}_{i}}{\d \hat{D}_{1}}
    & ~ = ~ \sum_{t \in [0: i - 1]} \frac{1}{t!} \cdot (-\ln \hat{D}_{1})^{t}
    ~ - ~ \sum_{t \in [1: i - 1]} \frac{1}{(t - 1)!} \cdot (-\ln \hat{D}_{1})^{t - 1}
    \notag \\
    & ~ = ~ \frac{1}{(i - 1)!} \cdot (-\ln \hat{D}_{1})^{i - 1}.
    \label{eq:ar_ap:uniform:2:4}
\end{align}

We actually have one more identity $1 = \ap(x, \vec{F}_{(\infty)}^*) = x \cdot (k - \sum_{i \in [k]} \hat{D}_{i})$ for $x \in (\frac{1}{k}, \infty)$, due to \Cref{cla:ar_ap:uniform:1} (in the case that $n \to \infty$). Rearrange this identity and take the derivative:
\begin{align}
    \frac{\d x}{\d \hat{D}_{1}}
    & ~ = ~ \frac{\d}{\d \hat{D}_{1}} \Big(\frac{1}{k - \sum_{i \in [k]} \hat{D}_{i}}\Big)
    \notag \\
    & ~ = ~ \frac{1}{(k - \sum_{i \in [k]} \hat{D}_{i})^{2}} \cdot \sum_{i \in [k]} \frac{\d \hat{D}_{i}}{\d \hat{D}_{1}}
    \notag \\
    & ~ = ~ \frac{1}{(k - \sum_{i \in [k]} \hat{D}_{i})^{2}} \cdot \sum_{i \in [0: k - 1]} \frac{1}{i!} \cdot (-\ln \hat{D}_{1})^{i}
    \notag \\
    & ~ = ~ \frac{1}{(k - \sum_{i \in [k]} \hat{D}_{i})^{2}} \cdot \frac{\hat{D}_{k}}{\hat{D}_{1}},
    \label{eq:ar_ap:uniform:2:5}
\end{align}
where the third step applies \Cref{eq:ar_ap:uniform:2:4}; and the last step applies \Cref{eq:ar_ap:uniform:2:3}.

Combining everything together, we deduce that
\begin{align}
    \label{eq:ar_ap:uniform:2:6}
    \ar(\vec{F}_{(\infty)}^*)
    & ~ = ~ 1 ~ + ~ k \cdot \int_{1 / k}^{\infty} (1 - \hat{D}_{k + 1}(x)) \cdot \d x
    \notag \\
    & ~ = ~ 1 ~ + ~ k \cdot \int_{1 / k}^{\infty} \frac{1 - \hat{D}_{k + 1}(x)}{(k - \sum_{i \in [k]} \hat{D}_{i}(x))^{2}} \cdot \frac{\hat{D}_{k}(x)}{\hat{D}_{1}(x)} \cdot \d \hat{D}_{1}(x),
\end{align}
where the first step applies \Cref{eq:ar_ap:uniform:2:1} for the limit instance $\vec{F}_{(\infty)}^* = \lim_{n \to \infty} \vec{F}_{(n)}^*$; and the last step follows from \Cref{eq:ar_ap:uniform:2:5}.

For the above revenue formula $\ar(\vec{F}_{(\infty)}^*)$, note that when the bid $x$ ranges from $\frac{1}{k}$ to $\infty$, the highest CDF $\hat{D}_{1}(x)$ ranges from $0$ to $1$. Moreover, \Cref{eq:ar_ap:uniform:2:3} characterizes, as a formula of $\hat{D}_{1}(x)$, the $i$-th highest CDF $\hat{D}_{i}(x)$. Thus, if we instead regard $\hat{D}_{1} \in (0, 1)$ as the variable,
\begin{align*}
    \ar(\vec{F}_{(\infty)}^*)
    & ~ = ~ 1 ~ + ~ k \cdot \int_{0}^{1} \frac{1 - \hat{D}_{k + 1}}{(k - \sum_{i \in [k]} \hat{D}_{i})^{2}} \cdot \frac{\hat{D}_{k}}{\hat{D}_{1}} \cdot \d \hat{D}_{1},
\end{align*}
where $\{\hat{D}_{i}\}_{i \in [2: k + 1]}$ once again are given by \Cref{eq:ar_ap:uniform:2:3}.

Under the substitution $z \eqdef -\ln \hat{D}_{1} \in (0, \infty)$, we can check via elementary calculation that
\begin{align*}
    \ar(\vec{F}_{(\infty)}^*)
    ~ = ~ 1 + k \cdot \int_{0}^{\infty} T_k(z) \cdot (1 - T_{k + 1}(z)) \cdot (k - \sum_{i \in [k]} T_{i}(z))^{-2} \cdot \d z,
\end{align*}
for the functions $\{T_{i}\}_{i \in [k + 1]}$ defined in the statement of the claim.

This completes the proof of \Cref{cla:ar_ap:uniform:2}.
\end{proof}

\begin{remark}
In the single-item case $k = 1$, we can deduce from \Cref{cla:ar_ap:uniform:2} that
\[
    \Re_{\ar/\ap}(1)
    ~ = ~ 1 + \int_{0}^{\infty} \frac{e^{x} - (1 + x)}{(e^{x} - 1)^2} \cdot \d x
    ~ = ~ \frac{\pi^2}{6}
    ~ \approx ~ 1.6449,
\]
which recovers the known result \cite[Theorem~2]{JLTX19}. In the multi-unit case $k \geq 2$, however, the supremum revenue gap $\Re_{\ar/\ap}(k)$ does not have an elementary expression. We will show in \Cref{sec:asymptotic_formula} that $\Re_{\ar/\ap}(k) = 1 + \Theta(1 / \sqrt{k})$. Associated with numeric calculation, it turns out that the worst case $\arg\max \{\Re_{\ar/\ap}(k): k \in \NP\}$ happens when $k = 1$.
\end{remark}

\subsection{Lower bound}
\label{subsec:ar_ap:lower}

We emphasize that all upper-bound results given in \Cref{subsec:ar_ap:worst_case,subsec:ar_ap:gap} just require the input distributions $\{F_{j}\}_{j \in [n]}$ to be independent. In this section, we construct matching lower-bound instances respectively in the i.i.d.\ general setting and the asymmetric regular setting. For convenience, we reuse the notations introduced before.

\vspace{.1in}
\noindent
{\bf I.I.D.\ general setting.}
Let us revisit the i.i.d.\ instance $\vec{F}_{(n)}^* = \{F_{(n)}^*\}^n$ specified by \Cref{cla:ar_ap:uniform:1}. As mentioned, the revenue gap $\{\Re_{\ar/\ap}(k, n)\}_{n \geq k}$ is an increasing sequence in the population $n \geq k$, and the limit/supremum revenue gap
\[
    \Re_{\ar/\ap}(k) ~ = ~ \lim_{n \to \infty} \Re_{\ar/\ap}(k, n)
\]
is {\em finite} for any $k \in \NP$ (see \Cref{sec:asymptotic_formula}). Accordingly, for a given $\eps > 0$, there is a threshold population $N_{1}(\eps) \geq k$ so that $\Re_{\ar/\ap}(k, n) \geq \Re_{\ar/\ap}(k) - \eps$, for any $n \geq N_{1}(\eps)$. Clearly, such instances $\vec{F}_{(n)}^* = \{F_{(n)}^*\}^n$ give the matching lower bound.

The common CDF $F_{(n)}^*$ specified in \Cref{cla:ar_ap:uniform:1} turns out to be the {\em equal-revenue} distribution (i.e.\ a ``boundary-case'' regular distribution) when $k = n = 1$, but is an irregular distribution otherwise. For example, when $k = 1$ and $n \geq 2$, we have
\[
    F_{(n)}^*(x)
    ~ = ~ \sqrt[n]{\Big| 1 - \frac{1}{x} \Big|_{+}},
\]
and the irregularity is shown in \cite[Lemma~12]{JLTX19}. In the other cases $n \geq k \geq 2$, the irregularity can be seen via similar but more technical arguments. For ease of presentation, here we omit the formal proof.

\vspace{.1in}
\noindent
{\bf Asymmetric regular setting.}
We next use the triangle distributions to construct an instance whose revenue gap matches the bound $\Re_{\ar/\ap}(k)$ given in \Cref{cla:ar_ap:uniform:2}. (Recall \Cref{subsec:prelim:distr} that a triangle distribution must be regular.) To this end, we would reuse the notations introduced in the proof of \Cref{cla:ar_ap:uniform:2}. The following claim is useful.

\begin{claim}[Threshold for Lower Bound]
\label{cla:ar_ap:uniform:lower}
Consider the limit instance $\vec{F}_{(\infty)}^*$ as well as its $i$-th highest CDF's $\{\hat{D}_{i}\}_{i \in [k + 1]}$ given in \Cref{cla:ar_ap:uniform:2}. For any $\eps > 0$, there exists a large enough $N_{2}(\eps) \in \NP$ so that
\begin{align*}
    k \cdot \int_{a}^{b} (1 - \hat{D}_{k + 1}(x)) \cdot \d x
    ~ \geq ~ \Re_{\ar/\ap}(k) - 1 - \eps,
\end{align*}
where $a \eqdef \frac{1}{k} + \frac{1}{N_{2}(\eps)}$ and $b \eqdef \frac{1}{k} + N_{2}(\eps)$; note that $a \leq b$.
\end{claim}

\begin{proof}
According to \Cref{eq:ar_ap:uniform:2:6}, in the limit case $N_{2}(\eps) \to \infty$ we have
\[
    k \cdot \int_{1 / k}^{\infty} (1 - \hat{D}_{k + 1}(x)) \cdot \d x
    ~ = ~ \ar(\vec{F}_{(\infty)}^*) - 1
    ~ = ~ \Re_{\ar/\ap}(k) - 1.
\]
We know from \Cref{cla:ar_ap:uniform:3} (see \Cref{sec:asymptotic_formula}) that the above improper integral $ = \frac{\Theta(1)}{\sqrt{k}}$ is finite. In addition, the integrand is a non-negative function. Given these, we can easily see \Cref{cla:ar_ap:uniform:lower}.
\end{proof}

Based on the above parameters $b \geq a$, we now construct a desired lower-bound instance.

\begin{example}[Lower-Bound Instance in Asymmetric Regular Setting]
\label{exp:ar_ap_non_iid}
Denote $\delta \eqdef \frac{b - a}{n} > 0$, where the integer $n \geq k$ will be determined later. As \Cref{fig:ar_ap_non_iid} shows, consider such an $(n + n k)$-buyer triangle instance $\vec{F} \eqdef \{\tri(v_{0, l}, q_{0, l})\}_{l \in [n]} \cup \{\tri(v_{j, l}, q_{j, l})\}_{j \in [n], l \in [k]}$:
\begin{itemize}
    \item In the $0$-th group, the involved monopoly prices $v_{0, l} \eqdef b$ for $l \in [n]$ are identical. In each group $j \in [n]$, the involved monopoly prices $v_{j, l} \eqdef b - j \cdot \delta$ for $l \in [k]$ are identical.
    
    \item In the $0$-th group, the involved monopoly quantiles $\{q_{0, l}\}_{l \in [n]}$ are identical, which together give a {\em unit} {\AnonymousPricing} revenue
    \begin{center}
        $\ap(p, \vec{F}) = 1$ under the posted price $p = v_{0, l} = b$.
    \end{center}
    The remaining monopoly quantiles $\{q_{j, l}\}_{j \in [n], l \in [k]}$ are defined recursively. In each group $j \in [n]$, the involved $\{q_{j, l}\}_{l \in [k]}$ are identical, which give a {\em unit} {\AnonymousPricing} revenue
    \begin{center}
        $\ap(p, \vec{F}) = 1$ under the posted price $p = v_{j, l} = b - j \cdot \delta$.
    \end{center}
\end{itemize}
\end{example}

\begin{figure}
    \centering
    \begin{tikzpicture}[thick, smooth, scale = 2]
        \draw[->] (0, 0) -- (7, 0);
        \draw[->] (0, 0) -- (0, 2.1);
        \node[above] at (0, 2.1) {$\ap(p, \vec{F})$};
        \node[right] at (7, 0) {$p$};
        \node[left] at (0, 2) {$1$};
        \draw[loosely dashed] (0, 2) -- (7, 2);
        \node[left] at (0, 1.65) {$1 - \sigma$};
        \draw[loosely dashed] (0, 1.65) -- (7, 1.65);
        \node[anchor = 45] at (0, 0) {$0$};

        \draw[blue, domain = 0: 1.5] plot (\x, {2 * \x - \x^2 / 2.25});
        \node[below] at (1, 0) {$1 / k$};
        \draw[very thick] (1, 0pt) -- (1, 1.5pt);

        \node[below] at (1.5, -1.2pt) {$v_{n, l}$};
        \draw[very thick] (1.5, 0pt) -- (1.5, 1.5pt);
        \draw[blue, dotted] (1.5, 1.85) -- (1.5, 2);
        \draw[blue, fill = red] (1.5, 2) circle (1pt);

        \draw[blue] (1.5, 1.85) -- (2, 2);
        \draw[blue, fill = white] (1.5, 1.85) circle (1pt);
        \node[below] at (2, -1.2pt) {$v_{n - 1, l}$};
        \draw[very thick] (2, 0pt) -- (2, 1.5pt);
        \draw[blue, dotted] (2, 1.7) -- (2, 2);
        \draw[blue, fill = red] (2, 2) circle (1pt);

        \draw[blue] (2, 1.7) -- (2.5, 2);
        \draw[blue, fill = white] (2, 1.7) circle (1pt);
        \node[below] at (2.5, -1.2pt) {$v_{n - 2, l}$};
        \draw[very thick] (2.5, 0pt) -- (2.5, 1.5pt);
        \draw[blue, dotted] (2.5, 1.75) -- (2.5, 2);
        \draw[blue, fill = red] (2.5, 2) circle (1pt);

        \draw[blue] (2.5, 1.75) -- (2.75, 1.875);
        \draw[blue, dotted] (2.75, 1.875) -- (4.25, 1.875);
        \draw[blue] (4.25, 1.875) -- (4.5, 2);
        \draw[blue, fill = white] (2.5, 1.75) circle (1pt);
        \node[below] at (3.5, -1.2pt) {$\cdots\cdots$};

        \draw[blue] (4.5, 1.75) -- (5, 2);
        \node[below] at (4.5, -1.2pt) {$v_{3, l}$};
        \draw[very thick] (4.5, 0pt) -- (4.5, 1.5pt);
        \draw[blue, dotted] (4.5, 2) -- (4.5, 1.75);
        \draw[blue, fill = red] (4.5, 2) circle (1pt);

        \draw[blue] (5, 1.8) -- (5.5, 2);
        \draw[blue, fill = white] (4.5, 1.75) circle (1pt);
        \node[below] at (5, -1.2pt) {$v_{2, l}$};
        \draw[very thick] (5, 0pt) -- (5, 1.5pt);
        \draw[blue, dotted] (5, 2) -- (5, 1.8);
        \draw[blue, fill = red] (5, 2) circle (1pt);

        \draw[blue] (5.5, 1.8) -- (6, 2);
        \draw[blue, fill = white] (5, 1.8) circle (1pt);
        \node[below] at (5.5, -1.2pt) {$v_{1, l}$};
        \draw[very thick] (5.5, 0pt) -- (5.5, 1.5pt);
        \draw[blue, dotted] (5.5, 2) -- (5.5, 1.8);
        \draw[blue, fill = red] (5.5, 2) circle (1pt);

        \draw[blue, fill = white] (5.5, 1.8) circle (1pt);
        \node[below] at (6.1, 0) {$b = v_{0, l}$};
        \draw[very thick] (6, 0pt) -- (6, 1.5pt);
        \draw[blue, dotted] (6, 2) -- (6, 0);
        \draw[blue, fill = red] (6, 2) circle (1pt);

        \draw[blue] (6, 0.4pt) -- (6.9, 0.4pt);
    \end{tikzpicture}
    \caption{Demonstration for the triangle instance $\vec{F}$ given in \Cref{exp:ar_ap_non_iid}, where $\sigma > 0$ can be arbitrarily small when $\delta > 0$ is small enough (i.e.\ when $n \in \NP$ is large enough).}
    \label{fig:ar_ap_non_iid}
\end{figure}
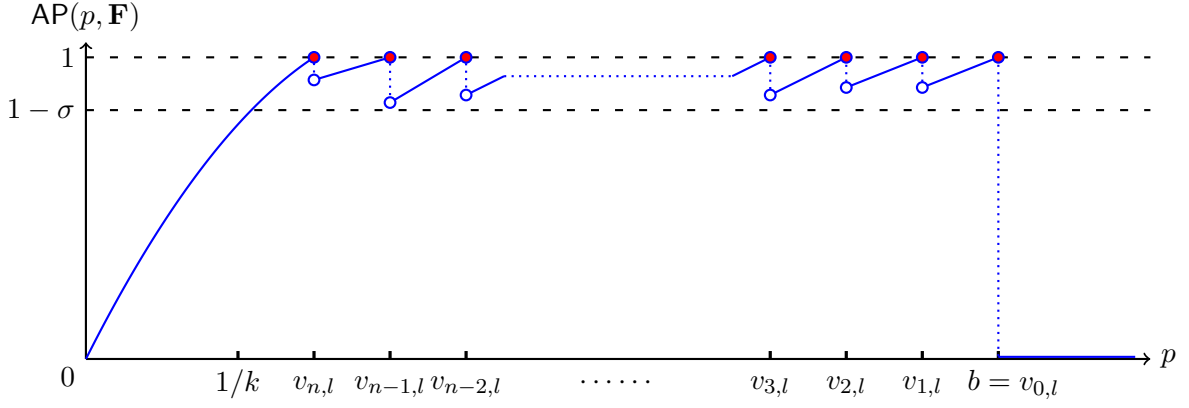

\begin{claim}[Part~4 of \Cref{thm:ar_ap:uniform} in Asymmetric Regular Setting]
\label{cla:ar_ap:uniform:4}
The $(n + n k)$-buyer triangle instance $\vec{F}$ in \Cref{exp:ar_ap_non_iid} is well defined, and satisfies the following:
\begin{enumerate}[font = {\em \bfseries}]
    \item $\ap(p, \vec{F}) \leq 1$ for any posted price $p \in \RP$.
    
    \item There exists a threshold $N_{3}(\eps) \in \NP$ such that for any $n \geq N_{3}(\eps)$,
    \[
        \ar(\vec{F}) ~ \geq ~ \Re_{\ar/\ap}(k) - 2 \cdot \eps.
    \]
\end{enumerate}
\end{claim}

\begin{proof}[Proof of \Cref{cla:ar_ap:uniform:4}]
We first show that $\vec{F}$ is well defined or more precisely, the monopoly quantiles are well defined. Because the monopoly prices/quantiles in an individual group are identical, without ambiguity we denote $v_{0} = v_{0, l} = b$ and $q_{0} = q_{0, l}$ for $l \in [n]$ and $v_{j} = v_{j, l} = b - j \cdot \delta$ and $q_{j} = q_{j, l}$ for $l \in [k]$. Recall \Cref{subsec:prelim:distr} that a triangle distribution $\tri(v_{j}, q_{j})$ has the CDF
\begin{align*}
    F_{j}(x) ~ = ~
    \begin{cases}
        \frac{(1 - q_{j}) \cdot x}{(1 - q_{j}) \cdot x + v_{j} q_{j}},
        & \mbox{when} ~ x \in [0, v_{j}] \\
        1, & \mbox{when} ~ x \in (v_{j}, \infty)
    \end{cases}.
\end{align*}

Under the posted price $p = v_{0} = b$, we have $F_{0}(p) = 1 - q_{0}$ for the $0$-th group and $F_{j}(p) = 1$ for any other group $j \in [n]$. Thus, only the group-$0$ buyers contribute to the {\AnonymousPricing} revenue $\ap(v_{0}, \vec{F})$. This revenue formula $\ap(v_{0}, \vec{F})$ can be regarded as a continuous function in $q_{0} \in [0, 1]$. Further, we observe that
\begin{enumerate}[label = (\roman*)]
    \item If a group-$0$ buyer is willing to pay $p = v_{0} = b$ with probability $q_{0} = 0$, then $\ap(v_{0}, \vec{F}) = 0$.
    
    \item If a group-$0$ buyer is willing to pay $p = v_{0} = b$ with probability $q_{0} = 1$, then $\ap(v_{0}, \vec{F}) = v_{0} \cdot \min\{n, k\} = b \cdot k > \frac{1}{k} \cdot k = 1$ (given that $n \geq k$ and $b > \frac{1}{k}$).
\end{enumerate}
Given these and due to the intermediate value theorem, $\ap(v_{0}, \vec{F}) = 1$ for some $q_{0} \in [0, 1]$. We conclude that the group-$0$ monopoly quantiles are well defined.

For some $m \in [0: n - 1]$, suppose that all the monopoly quantiles $q_{j} \in [0, 1]$ in the groups $j \in [0: m]$ are well defined, below we justify the existence of the group-$(m + 1)$ monopoly quantiles $q_{m + 1} \in [0, 1]$.

By construction, under any posted price $p \in (v_{m + 1}, v_{m}]$, the revenue $\ap(p, \vec{F})$ is contributed only by the buyers in the groups $j \in [0: m]$. In particular, when $p = v_{m}$, by construction we have $\ap(v_{m}, \vec{F}) = 1$.  Within the support $x \in [0, v_{j}]$, a triangle distribution $\tri(v_{j}, q_{j})$ has the virtual value function
\begin{align*}
    \varphi_{j}(x)
    ~ = ~ x - \frac{1 - F_{j}(x)}{f_{j}(x)}
    ~ = ~
    \begin{cases}
        -\frac{v_{j} q_{j}}{1 - q_{j}},
        & \mbox{when} ~ x \in [0, v_{j}) \\
        v_{j}, & \mbox{when} ~ x = v_{j}
    \end{cases}.
\end{align*}
Hence, any allocation under any posted price $p \in (v_{m + 1}, v_{m})$ gives a {\em negative} virtual welfare. Due to the revenue-equivalence theorem \cite{M81}, the revenue formula $\ap(p, \vec{F})$ is a strictly increasing function in $p \in (v_{m + 1}, v_{m}]$.

When $p = v_{m + 1}$, we shall incorporate the contribution from the group-$(m + 1)$ buyers into the revenue $\ap(v_{m + 1}, \vec{F})$ as well. Once again, this revenue formula $\ap(v_{m + 1}, \vec{F})$ can be regarded as a continuous function in $q_{m + 1} \in [0, 1]$. And we have
\begin{enumerate}[label = (\roman*)]
    \item If a group-$(m + 1)$ buyer is willing to pay $p = v_{m + 1}$ with probability $q_{m + 1} = 0$, then we have $\ap(v_{m + 1}, \vec{F}) = \lim_{p \to v_{m + 1}^{+}} \ap(p, \vec{F}) < 1$, where the inequality holds because $\ap(p, \vec{F})$ is a strictly increasing function when $p \in (v_{m + 1}, v_{m}]$.
    
    \item If a group-$(m + 1)$ buyer is willing to pay $p = v_{m + 1}$ with probability $q_{0} = 1$, since there are $k$ such buyers, we have $\ap(v_{m + 1}, \vec{F}) = v_{m + 1} \cdot k > \frac{1}{k} \cdot k = 1$, where the inequality holds because $v_{m + 1} \geq a > \frac{1}{k}$ (by construction).
\end{enumerate}
Once again, we deduce from the intermediate value theorem that $\ap(v_{m + 1}, \vec{F}) = 1$ for some $q_{m + 1} \in [0, 1]$, namely the group-$(m + 1)$ monopoly quantiles are well defined. By induction, the triangle instance $\vec{F}$ is well defined.

From the above arguments, we also conclude Part~1 that $\ap(p, \vec{F}) \leq 1$ for all $p \in \RP$.

We next justify Part~2 that the optimal {\AnonymousReserve} revenue $\ar(\vec{F}) \geq \Re_{\ar/\ap}(k) - 2 \cdot \eps$ when the $n \in \NP$ is large enough. To this end, let us consider the specific reserve $r = v_{n} = a$. Indeed, when $n$ is large enough, the $(k + 1)$-th highest $D_{k + 1}$ resulted from $\vec{F}$ satisfies that
\begin{align}
    \label{eq:ar_ap:uniform:4:1}
    k \cdot \int_{a}^{b} (1 - D_{k + 1}(x)) \cdot \d x
    ~ \geq ~ k \cdot \int_{a}^{b} (1 - \hat{D}_{k + 1}(x)) \cdot \d x - \eps,
\end{align}
Assume \Cref{eq:ar_ap:uniform:4:1} to be true, then Part~2 follows immediately:
\begin{align*}
    \ar(a, \vec{F})
    & ~ = ~ \ap(a, \vec{F}) ~ + ~ k \cdot \int_{a}^{b} (1 - D_{k + 1}(x)) \cdot \d x \\
    & ~ \geq ~ 1 ~ + ~ k \cdot \int_{a}^{b} (1 - D_{k + 1}(x)) \cdot \d x \\
    & ~ \geq ~ 1 ~ + ~ k \cdot \int_{a}^{b} (1 - \hat{D}_{k + 1}(x)) \cdot \d x - \eps \\
    & ~ = ~ \Re_{\ar/\ap}(k) - 2 \cdot \eps,
\end{align*}
where the second step follows by construction, i.e.\ $\ap(a, \vec{F}) = \ap(v_{n}, \vec{F}) = 1$; the third step applies \Cref{eq:ar_ap:uniform:4:1}; and the last step applies \Cref{cla:ar_ap:uniform:lower}.

We are left to prove \Cref{eq:ar_ap:uniform:4:1}. By construction, one can easily see that in the limit case $n \to \infty$, every individual monopoly quantile $q_{j}$ involved in $\vec{F}$ approaches to $0^+$. Namely, the CDF $\lim_{n \to \infty} F_{j}(x) \to 1^{-}$ for any $x \in \RP$. Given this, reusing the arguments for \Cref{eq:ar_ap:uniform:2:3}, it can be seen that for each $i \in [k + 1]$, the following holds for the limit $i$-th highest CDF:
\begin{align*}
    \lim_{n \to \infty} D_{i}(x)
    & ~ = ~ \sum_{t \in [0: i - 1]} \frac{1}{t!} \cdot (\lim_{n \to \infty} D_{1}(x)) \cdot (-\ln (\lim_{n \to \infty} D_{1}(x)))^{t},
\end{align*}
for all $x \in (\frac{1}{k}, \infty)$. Accessing the proof of \Cref{cla:ar_ap:uniform:2}, for the limit instance $\vec{F}_{(\infty)}^*$ therein, we have the counterpart identities $\hat{D}_{i}(x) = \sum_{t \in [0: i - 1]} \frac{1}{t!} \cdot \hat{D}_{1}(x) \cdot (-\ln \hat{D}_{1}(x))^{t}$ for all $x \in (\frac{1}{k}, \infty)$.

By construction (as \Cref{fig:ar_ap_non_iid} suggests), in the limit case $n \to \infty$, we have another identity\footnote{More precisely, by construction we have $\ap(v_{j}, \vec{F}) = 1$ for every $j \in [0: n]$. Concerning the revenue formula $\ap(p, \vec{F}) = p \cdot \sum_{i \in [k]} (1 - D_{i}(p))$, we notice that the $i$-th highest CDF's $\{D_{i}\}_{i \in [k]}$ are increasing functions. Given these, for any $j \in [0: n - 1]$ and any posted price $p \in (v_{j + 1}, v_{j}]$ we have $\ap(p, \vec{F}) \geq \frac{p}{v_{j}} \cdot \ap(v_{j}, \vec{F}) = \frac{p}{v_{j}} \geq \frac{v_{j + 1}}{v_{j}}$. Under our construction that $v_{j} = b - j \cdot \delta$ for all $j \in [0: n]$, where $\delta = \frac{b - a}{n}$, the minimum $\frac{v_{j + 1}}{v_{j}}$ is equal to $\frac{v_{n}}{v_{n - 1}} = \frac{a}{a + (b - a) / n} \geq \frac{1 / k}{1 / k + (b - a) / n} \geq 1 - (b - a) \cdot \frac{k}{n}$. Thus, for any $p \in [a, b]$ we have $\lim_{n \to \infty} \ap(p, \vec{F}) \geq \lim_{n \to \infty} (1 - (b - a) \cdot \frac{k}{n}) = 1$. On the other hand, we have shown that $\ap(p, \vec{F}) \leq 1$ for all $p \in [a, b]$ (see Part~1 of the claim).}
\[
    \lim_{n \to \infty} \ap(x, \vec{F}) ~ = ~ 1,
\]
for all $x \in [a, b]$. Accessing the proof of \Cref{cla:ar_ap:uniform:2}, for the limit instance $\vec{F}_{(\infty)}^*$ therein, we have the counterpart identity $\ap(x, \vec{F}_{(\infty)}^*) = 1$ for all $x \in (\frac{1}{k}, \infty)$. Recall that $[a, b] \subseteq (\frac{1}{k}, \infty)$.

Based on the above identities, we can reapply the arguments for \Cref{cla:ar_ap:uniform:2} and deduce that
\begin{align*}
    \lim_{n \to \infty} D_{k + 1}(x)
    ~ = ~ \hat{D}_{k + 1}(x),
\end{align*}
for all $x \in [a, b]$. Given this, and since both integrals $\int_{a}^{b} (1 - D_{k + 1}(x)) \cdot \d x$ and $\int_{a}^{b} (1 - \hat{D}_{k + 1}(x)) \cdot \d x$ in \Cref{eq:ar_ap:uniform:4:1} are {\em definite} integrals, and both integrands $y = 1 - D_{k + 1}(x)$ and $\hat{y} = 1 - \hat{D}_{k + 1}(x)$ are bounded between $[0, 1]$, \Cref{eq:ar_ap:uniform:4:1} must hold for any sufficiently large $n \in \NP$.

This completes the proof of \Cref{cla:ar_ap:uniform:4}.
\end{proof}

\subsection{Matroid feasibility constraints}
\label{subsec:ar_ap:matroid}

In \Cref{subsec:ar_ap:worst_case,subsec:ar_ap:gap,subsec:ar_ap:lower} we assume that any subset of up to $k \in \NP$ willing-to-pay buyers can win simultaneously, i.e.\ the winners meet a rank-$k$ {\em uniform matroid} constraint. To model some particular markets, many past works on Bayesian mechanism design also consider the general matroid constraints.

In this new scenario, the revenue gap between {\AnonymousReserve} and {\AnonymousPricing} is no longer a constant, or precisely, $\Re_{\ar/\ap} = \Omega(\log k)$. In the rest of this section, we assume basic knowledge about matroid, for which the reader can turn to \cite{O06}.

Regarding a general rank-$k$ matroid constraint, {\AnonymousPricing} runs almost in the same way: a certain buyer $i \in [n]$, upon arriving, gets a copy of the item iff (i)~he together with the past winning buyers form an independent set of the matroid; and (ii)~he is willing to pay the posted price $p \geq 0$. No ambiguity would arise throughout the conduct of {\AnonymousPricing}, due to the greedy structure of matroids.

To implement {\AnonymousReserve}, the seller should use {\VCGAuction} instead of {\sf $(k + 1)$-th Price Auction}: (i)~the seller runs {\VCGAuction} only on the buyers whose bids $\{b_{j}\}_{j \in [n]}$ are at least the reserve $r \geq 0$, by taking the matroid constraint into account; and (ii)~each winner pays the threshold bid for him to keep winning.

Our lower-bound example with the $\Omega(\log k)$ revenue gap is constructed below.

\begin{theorem}[{\ar} vs.\ {\ap} under a Matroid Constraint]
\label{thm:ar_ap:matroid}
When the seller faces $n \geq 1$ independent unit-demand buyers and the winners satisfy a rank-$k$ matroid constraint, the revenue gap $\Re_{\ar/\ap}$ between {\AnonymousReserve} and {\AnonymousPricing} is lower bounded by $\Omega(\log k)$.
\end{theorem}

\begin{proof}[Proof of \Cref{thm:ar_ap:matroid}]
For simplicity, we assume that $n = 2m$ is an even integer and that $k \leq m$; the lower-bound instance for the general case is very similar. The buyers are divided into $m$ pairs, and each pair $i \in [m]$ involves the $(2i - 1)$-th and $2i$-th buyers. We consider a specific rank-$k$ matroid $\mathcal{M}$ in terms of the collection $\mathcal{B}$ of its bases:
\begin{quote}
    Any base $B \in \mathcal{B}$ contains exactly one buyer from each chosen pair, for some choice of $k$ pairs. In total, there are $|\mathcal{B}| = \binom{m}{k} \cdot 2^k$ bases.
\end{quote}
One can easily justify the augmentation property, thus showing $\mathcal{M}$ to be a matroid (or more precisely, a laminar matroid with the laminar family $\{[1: m], [m + 1: 2m], [1: 2m]\}$ and the capacity function $c([1: m]) = c([m + 1: 2m]) = c([1: 2m]) = k$). Further, both buyers of each $i$-th pair have a deterministic bid $\max \{\frac{1}{i},~\frac{1}{k + 1}\}$.

In {\AnonymousPricing}, when the seller posts a price $p > \frac{1}{k}$, exactly $\lfloor \frac{1}{p} \rfloor$ pairs would pay this price, hence a revenue $p \cdot \lfloor \frac{1}{p} \rfloor \leq 1$. When the price $p \leq \frac{1}{k}$, although all the $k$ copies will be sold out, the revenue $p \cdot k$ is still at most $1$. But when the seller instead employs {\VCGAuction} (even without a reserve), either buyer in each of the top-$k$ pairs will get an item by paying $\frac{1}{i}$, thus a revenue of $\sum_{j \in [k]} \frac{1}{i} = \Omega(\log k)$.

This completes the proof of \Cref{thm:ar_ap:matroid}.
\end{proof}

\begin{remark}
We notice that a deterministic bid satisfies the regularity distributional assumption, as well as the stronger {\em monotone-hazard-rate} assumption.\footnote{A distribution $F_{j}$ has {\em monotone hazard rate} if $y = \ln(1 - F_{j}(x))$ is a concave function, e.g., see \cite{JLX19}.} Thus, the $\Omega(\log k)$ lower bound still holds for the revenue gap $\Re_{\ar/\ap}$ in these restricted settings.
\end{remark}
\section{Ex-Ante Relaxation vs.\ Anonymous Pricing}
\label{sec:ear_ap}

In this section, we investigate the {\ExAnteRelaxation} (\ear) vs. {\AnonymousPricing} (\ap) problem, under the regularity assumption that $\vec{F} = \{F_{j}\}_{j \in [n]} \subseteq \reg$. Based on the revenue formulas (see \Cref{subsec:prelim:mech}), the revenue gap between both mechanisms is given by the optimal solution to the following mathematical program. Recall that  $D_{i}$ is the $i$-th highest bid distribution, and $\reg$ is the family of all regular distributions.
\begin{align}
    & \sup
    && \ear(\vec{q}', \vec{F}) ~ = ~ \sum_{j \in [n]} F_{j}^{-1}(1 - q'_{j}) \cdot q'_{j}
    \label{prog:ear_old}\tag{P2} \\
    & \mbox{\,\,s.t.}
    && \ap(p, \vec{F}) ~ = ~ p \cdot \sum_{i \in [k]} (1 - D_{i}(p)) ~ \leq ~ 1,
    && \forall p \in \RP,
    \notag \\
    &&& \sum_{j \in [n]} q'_{j} ~ \leq ~ k,
    \notag \\
    & && \vec{q}' = \{q'_{j}\}_{j \in [n]} \in [0, 1]^{n}, ~ \vec{F} = \{F_{j}\}_{j \in [n]} \subseteq \reg,
    && \forall n \in \NP.
    \notag
\end{align}

We will establish an $O(\log k)$ upper bound for the optimal solution to Program~\eqref{prog:ear_old}, which is formalized as \Cref{thm:ear_ap}. Combine this result with the matching lower bound by \cite[Example~5.4]{HR09}, then the revenue gap gets understood.

\begin{theorem}[{\ear} vs.\ {\ap}]
\label{thm:ear_ap}
Given that the seller has $k \in \NP$ homogeneous items and faces $n \geq k$ independent unit-demand buyers, who have regular value distributions $\vec{F} = \{F_{j}\}_{j \in [n]} \subseteq \reg$, the revenue gap between {\ExAnteRelaxation} and {\AnonymousPricing} is $\Re_{\ear / \ap}(k) = \Theta(\log k)$.
\end{theorem}

We establish \Cref{thm:ear_ap} in three steps. First, we give a reduction from a regular instance to a triangle instance, which preserves the feasibility; then we just need to optimize $n$ pairs of monopoly price and quantile $\{(v_{j}, q_{j})\}_{j \in [n]}$ instead of $n$ regular distributions $\{F_{j}\}_{j \in [n]}$. Second, we relax the constraint $\ap(p, \vec{F}) \leq 1$ to a more tractable constraint, which avoids the correlation among the order statistics $\{D_{i}\}_{i \in [k]}$. Afterwards, we divide all buyers into three careful groups under certain criteria for $\{(v_{j}, q_{j})\}_{j \in [n]}$, and separately bound the contribution from each group to the {\ear} revenue. The total {\ear} revenue turns out to be $O(\log k)$.

\vspace{.1in}
\noindent
{\bf Reduction to triangle instances.}
For the single-item case $k = 1$, \cite{AHNPY19} show that the worst case of Program~\eqref{prog:ear_old} w.l.o.g.\ is achieved by a triangle instance. Indeed, their arguments work in the general case $k \in \NP$ as well. Formally, we have the following lemma (see \Cref{fig:triangle} for a demonstration).

\begin{lemma}[Reduction for {\ear} vs.\ {\ap} {\cite[Lemma~4.1]{AHNPY19}}]
\label{lem:ear_ap:reduction}
Given a feasible solution $(\vec{q}', \vec{F})$ to Program~\eqref{prog:ear_old}, there exists another $n$-buyer feasible instance $(\vec{q}^*, \vec{F}^*)$ such that:
\begin{enumerate}[font = {\em\bfseries}]
    \item The distributions $\vec{F}^* = \{F_{j}^*\}_{j \in [n]} \subseteq \tri$ are triangle distributions, and $\vec{q}^* = \{q_{j}^*\}_{j \in [n]} \in [0, 1]^{n}$ (such that $\sum_{j \in [n]} q_{j}^* \leq k$) are the monopoly quantiles thereof.
    
    \item The {\ExAnteRelaxation} revenue keeps the same, i.e.\ $\ear(\vec{q}^*, \vec{F}^*) = \ear(\vec{q}', \vec{F})$.
    
    \item The distributions $\vec{F}^* = \{F_{j}^*\}_{j \in [n]}$ are stochastically dominated by $\vec{F} = \{F_{j}\}_{j \in [n]}$ and thus, for any price $p \in \RP$, the {\AnonymousPricing} revenue drops, i.e.\ $\ap(p, \vec{F}^*) \leq \ap(p, \vec{F})$.
\end{enumerate}
\end{lemma}

In view of \Cref{lem:ear_ap:reduction}, to establish \Cref{thm:ear_ap} we can focus on Program~\eqref{prog:ear} in place of the previous Program~\eqref{prog:ear_old}. For a triangle distribution $\tri(v_{j}, q_{j})$, where $v_{j} = F_{j}^{-1}(1 - q_{j}) \geq 0$ is the monopoly price, we reuse $F_{j}$ to denote its CDF. Recall \Cref{subsec:prelim:distr} that $F_{j}(x) = \frac{(1 - q_{j}) \cdot x}{(1 - q_{j}) \cdot x + v_{j} q_{j}}$ for all $x \leq v_{j}$ and $F_{j}(x) = 1$ for all $x > v_{j}$.
\begin{align}
    & \sup
    && \ear(\vec{F}) ~ = ~ \sum_{j \in [n]} v_{j} q_{j}
    \label{prog:ear}\tag{P3} \\
    & \mbox{\,\,s.t.}
    && \ap(p, \vec{F}) ~ = ~ p \cdot \sum_{i \in [k]} (1 - D_{i}(p)) ~ \leq ~ 1,
    && \forall p \in \RP,
    \label{cstr:ear_ap} \tag{C2} \\
    &&& \sum_{j \in [n]} q_{j} ~ \leq ~ k,
    \label{cstr:ear_quantile} \tag{C3} \\
    & && \vec{F} = \{\tri(v_{j}, q_{j})\}_{j \in [n]} \subseteq \reg,
    && \forall n \in \NP.
    \notag
\end{align}
For a single triangle distribution $\tri(v_{j}, q_{j})$, the optimal {\AnonymousPricing} revenue from it equals $\ap(\tri(v_{j}, q_{j})) = v_{j} q_{j}$, which $\leq \ap(\vec{F}) \leq 1$ due to constraint~\eqref{cstr:ear_ap}. We thus add one more constraint
\begin{align}
    & v_{j} q_{j} ~ \leq ~ 1,
    && \forall j \in [n].
    \label{cstr:ear_single} \tag{C4}
\end{align}

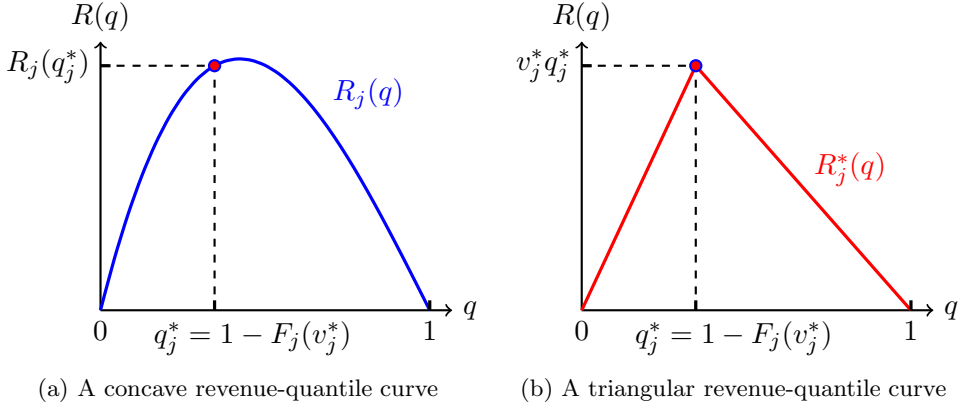
\begin{figure}[t]
    \centering
    \subfloat[A concave revenue-quantile curve
    \label{fig:triangle:1}]{
    \begin{tikzpicture}[thick, smooth, domain = 0: 1.732, scale = 2.5]
    \draw[->] (0, 0) -- (1.85, 0);
    \draw[->] (0, 0) -- (0, 1.425);
    \node[above] at (0, 1.425) {$R(q)$};
    \node[right] at (1.85, 0) {$q$};
    \node[below] at (0, 0) {$0$};
    \draw[very thick, color = blue] plot (\x, {2 * ((\x - 0.732)^3 / 3 - (\x - 0.732)^2 + 2 / 3)});
    \node[anchor = south west, color = blue] at (1.166, 1.011) {$R_{j}(q)$};
    \draw[dashed] (0.6, 0) -- (0.6, 1.297);
    \draw[dashed] (0, 1.297) -- (0.6, 1.297);
    \draw[color = blue, fill = red] (0.6, 1.297) circle(0.835pt);
    \draw[very thick] (1.732, 0) -- (1.732, 1.5pt);
    \node[below] at (1.732, 0) {$1$};
    \draw[very thick] (0.6, 0) -- (0.6, 1.5pt);
    \node[below] at (0.8, 0) {$q_{j}^* = 1 - F_{j}(v_{j}^*)$};
    \draw[very thick] (0, 1.297, 0) -- (1.5pt, 1.297);
    \node[left] at (0, 1.297) {$R_{j}(q_{j}^*)$};
    \end{tikzpicture}
    }
    \subfloat[A triangular revenue-quantile curve
    \label{fig:triangle:2}]{
    \begin{tikzpicture}[thick, smooth, domain = 0: 1.732, scale = 2.5]
    \draw[->] (0, 0) -- (1.85, 0);
    \draw[->] (0, 0) -- (0, 1.425);
    \node[above] at (0, 1.425) {$R(q)$};
    \node[right] at (1.85, 0) {$q$};
    \node[below] at (0, 0) {$0$};
    \draw[very thick, color = red] (0, 0) -- (0.6, 1.297) -- (1.732, 0);
    \node[anchor = south west, color = red] at (1.166, 0.6) {$R_{j}^*(q)$};
    \draw[dashed] (0.6, 0) -- (0.6, 1.297);
    \draw[dashed] (0, 1.297) -- (0.6, 1.297);
    \draw[color = blue, fill = red] (0.6, 1.297) circle(0.835pt);
    \draw[very thick] (1.732, 0) -- (1.732, 1.5pt);
    \node[below] at (1.732, 0) {$1$};
    \draw[very thick] (0.6, 0) -- (0.6, 1.5pt);
    \node[below] at (0.8, 0) {$q_{j}^* = 1 - F_{j}(v_{j}^*)$};
    \draw[very thick] (0, 1.297, 0) -- (1.5pt, 1.297);
    \node[left] at (0, 1.297) {$v_{j}^* q_{j}^*$};
    \end{tikzpicture}
    }
    \caption{Demonstration for the reduction in \cref{lem:ear_ap:reduction}, in terms of the {\em revenue-quantile curves}. For a distribution $F_{j}$, its revenue-quantile curve is given by $R_{j}(q) = q \cdot F_{j}^{-1}(1 - q)$ for $q \in [0, 1]$.
    This distribution $F_{j}$ is regular iff the $R_{j}$ is a concave function (as \Cref{fig:triangle:1} suggests). And the revenue-quantile curve of a triangle distribution is basically a triangle (i.e., a $2$-piecewise linear function, as \Cref{fig:triangle:2} suggests); in particular, the two base angles have the tangent values $v_{j}^*$ and $v_{j}^* q_{j}^* / (1 - q_{j}^*)$, respectively.}
    \label{fig:triangle}
\end{figure}

\noindent
{\bf Relaxing constraint~\eqref{cstr:ear_ap}.}
Given Program~\eqref{prog:ear}, both the objective function $\ear(\vec{F})$ and constraint~\eqref{cstr:ear_quantile} are easy to deal with. However, constraint~\eqref{cstr:ear_ap} is rather complicated, because it involves the correlated top-$k$ bids $\{b_{(i)}\}_{i \in [k]}$ and the corresponding order CDF's $\{D_{i}\}_{i \in [k]}$ (as formulas of the individual CDF's $\{F_{j}\}_{j \in [n]}$) are cumbersome.

The following \Cref{lem:ear_ap:relaxed_constraint} relaxes constraint~\eqref{cstr:ear_ap} to another constraint. The resulting constraint is  much easier to reason about. Namely, it avoids the correlation among the top-$k$ bids $\{b_{(i)}\}_{i \in [k]}$ and admits a clean formula of the individual CDF's $\{F_{j}\}_{j \in [n]}$. Later we will see that after this relaxation, the optimal objective value of Program~\eqref{prog:ear} blows up just by a constant multiplicative factor. Denote $m \eqdef \lfloor \frac{k}{2} \rfloor \geq 2$ for convenience.

\begin{lemma}[Relaxed Constraint]
\label{lem:ear_ap:relaxed_constraint}
The following is a necessary condition for constraint~\eqref{cstr:ear_ap}:
\begin{align*}
    & \sum_{j \in [n]} (1 - F_{j}(p)) ~ \leq ~ \frac{4}{p},
    && \forall p \in \Big[\frac{1}{m}, ~ \frac{1}{2}\Big].
\end{align*}
\end{lemma}

\begin{proof}[Proof of \Cref{lem:ear_ap:relaxed_constraint}]
Let us consider a specific price $p \in [\frac{1}{m}, \frac{1}{2}]$ for constraint~\eqref{cstr:ear_ap}. For any $j \in [n]$, let the independent Bernoulli random variable $X_{j} \in \{0, 1\}$ denote whether the $j$-th buyer is willing to pay the price $p$, with the failure probability $\Pr[X_{j} = 0] = F_{j}(p)$. Then $X \eqdef \sum_{j \in [n]} X_{j}$ denotes how many buyers are willing to pay, and $Y \eqdef \min \{k, X\}$ denotes how many items are sold out in {\AnonymousPricing}.

We have the revenue $\ap(p, \vec{F}) = p \cdot \E[Y]$, and constraint~\eqref{cstr:ear_ap} is identical to $\E[Y] \leq \frac{1}{p}$. For the equation given in \Cref{lem:ear_ap:relaxed_constraint}, the $\lhs = \sum_{j \in [n]} \Pr[X_{j} = 1] = \sum_{j \in [n]} \E[X_{j}] = \E[X]$.

On the opposite of \Cref{lem:ear_ap:relaxed_constraint}, suppose that $\E[X] > \frac{4}{p}$. We have $\E[X] > 8$, given that the price $p \leq \frac{1}{2}$. Since $X$ is the sum of independent Bernoulli random variables, due to Chernoff bound, $\Pr[X < (1 - \delta) \cdot \E[X]] < \frac{e^{-\delta \cdot \E[X]}}{(1 - \delta)^{(1 - \delta) \cdot \E[X]}}$ for any $\delta \in (0, 1)$. In particular,
\begin{align}
    \Pr\Big[X < \frac{1}{2} \cdot \E[X]\Big]
    ~ < ~ \Big(\frac{2}{e}\Big)^{\frac{1}{2} \cdot \E[X]}
    ~ < ~ \Big(\frac{2}{e}\Big)^{4}
    ~ \approx ~ 0.2931
    ~ < ~ \frac{1}{2},
    \label{eq:ear_ap:relaxed_constraint:1}
\end{align}
where the first step follows by setting $\delta = \frac{1}{2}$; and the second step follows since $\E[X] > 8$.

And because $Y = \min\{k, X\}$, we further deduce that
\begin{align}
    \Pr\Big[Y \geq \min\Big\{k, ~ \frac{1}{2} \cdot \E[X]\Big\}\Big]
    & ~ = ~ 1 - \Pr\Big[Y < \min\Big\{k, ~ \frac{1}{2} \cdot \E[X]\Big\}\Big]
    \notag \\
    & ~ = ~ 1 - \Pr\Big[X < \min\Big\{k, ~ \frac{1}{2} \cdot \E[X]\Big\}\Big]
    \notag \\
    & ~ \geq ~ 1 - \Pr\Big[X < \frac{1}{2} \cdot \E[X]\Big]
    \notag \\
    & ~ > ~ \frac{1}{2},
    \label{eq:ear_ap:relaxed_constraint:2}
\end{align}
where the second step follows since $Y < \min\{k, \frac{1}{2} \cdot \E[X]\}$ holds only if $Y < k$, and thus only if $Y = X$; and the last step follows from Equation~\eqref{eq:ear_ap:relaxed_constraint:1}.

Based on the above arguments, we conclude a contradiction $\E[Y] > \frac{1}{p}$ as follows:
\begin{align*}
    \E[Y]
    & ~ \geq ~ \Pr\Big[Y \geq \min\Big\{k, ~ \frac{1}{2} \cdot \E[X]\Big\}\Big] ~ \cdot ~ \min\Big\{k, ~ \frac{1}{2} \cdot \E[X]\Big\} \\
    & ~ > ~ \frac{1}{2} \cdot \min\Big\{k, ~ \frac{1}{2} \cdot \E[X]\Big\} \\
    & ~ \geq ~ \frac{1}{2} \cdot \min\Big\{k, ~ \frac{2}{p}\Big\} \\
    & ~ \geq ~ \frac{1}{p},
\end{align*}
where the second step applies Equation~\eqref{eq:ear_ap:relaxed_constraint:2}; the third step applies our assumption $\E[X] > \frac{4}{p}$; and the last step follows as $\frac{2}{p} \leq 2 m \leq k$, given that $p \in [\frac{1}{m}, \frac{1}{2}]$ and $m = \lfloor \frac{k}{2} \rfloor$.

By refuting the assumption, we get $\E[X] \leq \frac{4}{p}$ for any price $p \in [\frac{1}{m}, \frac{1}{2}]$. This completes the proof of \Cref{lem:ear_ap:relaxed_constraint}.
\end{proof}

Given a triangle instance $\{\tri(v_{j}, q_{j})\}_{j \in [n]}$, by plugging the CDF formulas $\{F_{j}\}_{j \in [n]}$, we can reformulate \Cref{lem:ear_ap:relaxed_constraint} as follows:
\begin{align}
    & \sum_{j \in [n]: v_{j} \geq p} \frac{v_{j} q_{j}}{(1 - q_{j}) \cdot p + v_{j} q_{j}} ~ \leq ~ \frac{4}{p},
    && \forall p \in \Big[\frac{1}{m}, ~ \frac{1}{2}\Big].
    \label{cstr:ear_ap_relaxed} \tag{C2$'$}
\end{align}

\noindent
{\bf Grouping the buyers.}
To upper bound the objective function $\ear(\vec{F}) = \sum_{j \in [n]} v_{j} q_{j}$, let us partition all the buyers into three groups $[n] = A \sqcup B \sqcup C$, where
\begin{align*}
    A & ~ \eqdef ~ \Big\{j \in [n]: v_{j} \geq \frac{1}{m} ~~ \mbox{and} ~~ \frac{v_{j} q_{j}}{1 - q_{j}} \geq \frac{1}{m}\Big\}, \\
    B & ~ \eqdef ~ \Big\{j \in [n]: v_{j} \geq \frac{1}{m} ~~ \mbox{and} ~~ \frac{v_{j} q_{j}}{1 - q_{j}} < \frac{1}{m}\Big\}, \\
    C & ~ \eqdef ~ \Big\{j \in [n]: v_{j} < \frac{1}{m}\Big\}.
\end{align*}

Regarding the groups $A$, $B$ and $C$ given above, their individual contributions to the benchmark $\ear(\vec{F})$ actually admit the following bounds:
\begin{align*}
    & \sum_{j \in A} v_{j} q_{j} ~ = ~ O(\log k),
    && \sum_{j \in B} v_{j} q_{j} ~ \leq ~ 8,
    && \sum_{j \in C} v_{j} q_{j} ~ \leq ~ 3.
\end{align*}
Suppose these bounds to be true, then combining them together immediately gives \Cref{thm:ear_ap}. Below we explain the intuitions of our grouping criteria (\cref{rem:grouping}), give an interesting observation for the instances that are constituted by ``small'' distributions (\cref{rem:small_distributions}), and then verify the above three bounds in the reverse order.

\begin{remark}[Grouping Criteria]
\label{rem:grouping}
Recall the objective function of Program~\eqref{prog:ear}, i.e., $\ear(\vec{F}) = \sum_{j \in [n]} v_{j} q_{j}$, and constraint~\eqref{cstr:ear_quantile}, i.e., $\sum_{j \in [n]} q_{j} \leq k$. Here the monopoly revenues $\{v_{j} q_{j}\}_{j \in [n]}$ are the individual contributions by the triangle distributions $\{\tri(v_{j}, q_{j})\}_{j \in [n]}$, and (in the sense of the {\em Knapsack Problem}) the monopoly quantiles $\{q_{j}\}_{j \in [n]}$ can be regarded as the individual capacities. Thereby, the monopoly prices $\{v_{j}\}_{j \in [n]}$ somehow are the bang-per-buck ratios (i.e., the contribution to the {\ear} benchmark per unit of the capacity).

Of course we prefer those distributions with higher bang-per-buck ratios $\{v_{j}\}_{j \in [n]}$, but also need to take the capacities $\{q_{j}\}_{j \in [n]}$ into account. In particular:
\begin{itemize}
    \item The group-$C$ distributions have lower bang-per-buck ratios $v_{j} \leq 1 / m$. So conceivably, the total contribution $\sum_{j \in C} v_{j} q_{j}$ by this group to the {\ear} benchmark shall be small, and we will prove an upper bound of $3$.
    
    \item The group-$B$ distributions have high enough bang-per-buck ratios $v_{j} \geq 1 / m$ but small capacities, namely $v_{j} q_{j} / (1 - q_{j}) < 1 / m$. It turns out that the total contribution $\sum_{j \in B} v_{j} q_{j}$ by this group is also small, and we will prove an upper bound of $8$.
    
    \item The group-$A$ distributions have high enough bang-per-buck ratios as well as big enough capacities. Thus, this group should contribute the most to the {\ear} benchmark, for which we will show $\sum_{j \in A} v_{j} q_{j} = O(\log k)$.
\end{itemize}
Indeed, our grouping criteria borrow ideas from the ``budget-feasible mechanism design'' literature \cite{S10,CGL11,GJLZ19}, where the primary goal is to design approximately optimal mechanisms for the Knapsack Problem under the incentive concerns.
\end{remark}

\begin{remark}[``Small'' Distributions]
\label{rem:small_distributions}
As argued in \cref{sec:our_results}, regarding a {\em continuum of ``small'' buyers} (i.e., any single buyer has an infinitesimal contribution to the {\ear} benchmark, but there are infinitely many buyers $n \to \infty$), the {\ear} vs.\ {\ap} revenue gap would be (at most) a universal constant for whatever $k \geq 1$. This is because every ``small'' buyer belongs to either group $B$ or group $C$, and thus the {\ear} benchmark is at most $\sum_{j \in B \cup C} v_{j} q_{j} \leq 8 + 3 = 11$.
\end{remark}

\vspace{.1in}
\noindent
{\bf Revenue from group $C$.}
Since such a buyer $j \in C$ has a monopoly price $v_{j} < \frac{1}{m}$, we have
\begin{align*}
    \sum_{j \in C} v_{j} q_{j}
    & ~ \leq ~ \frac{1}{m} \cdot \sum_{j \in C} q_{j} \\
    & ~ \leq ~ \frac{1}{m} \cdot \sum_{j \in [n]} q_{j} \\
    & ~ \leq ~ \frac{1}{m} \cdot k \\
    & ~ \leq ~ 3,
\end{align*}
where the second step follows since $C \subseteq [n]$; the third step follows from constraint~\eqref{cstr:ear_quantile}; and the last step holds for $m = \lfloor \frac{k}{2} \rfloor$ and $k \geq 4$. (We will deal with the cases $k \in \{1, 2, 3\}$ separately, at the end of this section.)

\vspace{.1in}
\noindent
{\bf Revenue from group $B$.}
Setting $p = \frac{1}{m}$ for constraint~\eqref{cstr:ear_ap_relaxed}, we deduce that
\begin{align*}
    4m ~ = ~ \rhs \mbox{ of } \eqref{cstr:ear_ap_relaxed}
    ~ \geq ~ \lhs \mbox{ of } \eqref{cstr:ear_ap_relaxed}
    & ~ = ~ \sum_{j \in [n]: v_{j} \geq \frac{1}{m}} \frac{v_{j} q_{j}}{(1 - q_{j}) \cdot \frac{1}{m} + v_{j} q_{j}} \\
    & ~ \geq ~ \sum_{j \in B} \frac{v_{j} q_{j}}{(1 - q_{j}) \cdot \frac{1}{m} + v_{j} q_{j}} \\
    & ~ \geq ~ \sum_{j \in B} \frac{v_{j} q_{j}}{(1 - q_{j}) \cdot \frac{1}{m} + (1 - q_{j}) \cdot \frac{1}{m}} \\
    & ~ \geq ~ \frac{m}{2} \cdot \sum_{j \in B} v_{j} q_{j},
\end{align*}
where the second line follows since $\{j \in [n]: v_{j} \geq \frac{1}{m}\} \supseteq B$ (see the definition of $B$); the third line follows since $\frac{v_{j} q_{j}}{1 - q_{j}} < \frac{1}{m}$ for any $j \in B$; and the last line drops the $(1 - q_{j})$ terms and then rearranges the formula.

Rearranging the above equation immediately gives $\sum_{j \in B} v_{j} q_{j} \leq 8$, as desired.

\vspace{.1in}
\noindent
{\bf Revenue from group $A$.}
To verify the upper bound about this group, we shall generalize the definition of $A$, and get a chain of subgroups $A = A_{m} \supseteq A_{m - 1} \supseteq \cdots \supseteq A_{2}$:
\begin{align*}
    A_{t} & ~ \eqdef ~ \Big\{j \in [n]: v_{j} \geq \frac{1}{t} ~~ \mbox{and} ~~ \frac{v_{j} q_{j}}{1 - q_{j}} \geq \frac{1}{t}\Big\},
    && \forall t \in [2: m].
\end{align*}

Given an index $t \in [2: m]$, by setting $p = \frac{1}{t} \in [\frac{1}{m}, \frac{1}{2}]$ for constraint~\eqref{cstr:ear_ap_relaxed}, we deduce that
\begin{align*}
    4t ~ = ~ \rhs \mbox{ of } \eqref{cstr:ear_ap_relaxed}
    ~ \geq ~ \lhs \mbox{ of } \eqref{cstr:ear_ap_relaxed}
    & ~ = ~ \sum_{j \in [n]: v_{j} \geq \frac{1}{t}} \frac{v_{j} q_{j}}{(1 - q_{j}) \cdot \frac{1}{t} + v_{j} q_{j}} \\
    & ~ \geq ~ \sum_{j \in A_{t}} \frac{v_{j} q_{j}}{(1 - q_{j}) \cdot \frac{1}{t} + v_{j} q_{j}} \\
    & ~ \geq ~ \sum_{j \in A_{t}} \frac{v_{j} q_{j}}{v_{j} q_{j} + v_{j} q_{j}} \\
    & ~ = ~ \frac{1}{2} \cdot |A_{t}|,
\end{align*}
where the second step follows because $\{j \in [n]: v_{j} \geq \frac{1}{m}\} \supseteq A_{t}$ (see the definition of $A_{t}$); and the third step follows because $(1 - q_{j}) \cdot \frac{1}{t} \leq v_{j} q_{j}$ for each $j \in A_{t}$.

Based on the above equation, we easily bound the cardinality $|A_{t}| \leq 8 t$ for each $t \in [2: m]$. Combining the above arguments together gives
\begin{align}
    \sum_{j \in A} v_{j} q_{j}
    ~ = ~ \sum_{j \in A_{m}} v_{j} q_{j}
    & ~ = ~ \sum_{j \in A_{2}} v_{j} q_{j} + \sum_{t \in [3: m]} \sum_{j \in A_{t} \setminus A_{t - 1}} v_{j} q_{j}
    \notag \\
    & ~ \leq ~ \sum_{j \in A_{2}} v_{j} q_{j} + \sum_{t \in [3: m]} \sum_{j \in A_{t} \setminus A_{t - 1}} \frac{1}{t - 1} \cdot 1
    \notag \\
    & ~ = ~ \sum_{j \in A_{2}} v_{j} q_{j} + \sum_{t \in [3: m]} \frac{|A_{t}| - |A_{t - 1}|}{t - 1}
    \notag \\
    & ~ = ~ \Big(\sum_{j \in A_{2}} v_{j} q_{j} - \frac{|A_{2}|}{2}\Big) + \frac{|A_{m}|}{m - 1} + \sum_{t \in [3: m]} |A_{t}| \cdot \Big(\frac{1}{t - 1} - \frac{1}{t}\Big)
    \notag \\
    & ~ \leq ~ \Big(\sum_{j \in A_{2}} v_{j} q_{j} - \frac{|A_{2}|}{2}\Big) + \frac{8m}{m - 1} + \sum_{t \in [3: m]} 8t \cdot \Big(\frac{1}{t - 1} - \frac{1}{t}\Big)
    \notag \\
    & ~ \leq ~ \Big(\sum_{j \in A_{2}} v_{j} q_{j} - \frac{|A_{2}|}{2}\Big) + 16 + \sum_{t \in [3: m]} 8t \cdot \Big(\frac{1}{t - 1} - \frac{1}{t}\Big)
    \notag \\
    & ~ = ~ \Big(\sum_{j \in A_{2}} v_{j} q_{j} - \frac{|A_{2}|}{2}\Big) + 8 + \sum_{t \in [m - 1]} \frac{8}{t},
    \label{eq:group_A}
\end{align}
where the second line follows because the monopoly price $v_{j} \in (\frac{1}{t - 1}, \frac{1}{t}]$ for each $j \in A_{t} \setminus A_{t - 1}$ (see the definitions of $A_{t}$ and $A_{t - 1}$), and the monopoly quantiles $q_{j} \in [0, 1]$ are bounded; the fifth line applies the bounds $|A_{t}| \leq 8t$ for each $t \in [2: m]$; the sixth line holds for $m = \lfloor \frac{k}{2} \rfloor$ and $k \geq 4$; and the last line is by elementary calculation.

Because $v_{j} q_{j} \leq 1$ for all $j \in A_{2}$ (see constraint~\eqref{cstr:ear_single}) and $|A_{2}| \leq 16$, we can bound the first term in Equation~\eqref{eq:group_A}: $\sum_{j \in A_{2}} v_{j} q_{j} - \frac{|A_{2}|}{2} \leq |A_{2}| - \frac{|A_{2}|}{2} \leq 8$. Plug this into Equation~\eqref{eq:group_A}:
\begin{align*}
    \sum_{j \in A} v_{j} q_{j}
    ~ \leq ~ 16 + \sum_{t \in [m - 1]} \frac{8}{t}
    ~ = ~ O(\log k),
\end{align*}
where the last step holds for $m = \lfloor \frac{k}{2} \rfloor$.

\vspace{.1in}
\noindent
{\bf Upper bound when $k \in \{1, 2, 3\}$.}
Clearly, the optimal value $\Re_{\ear/\ap}(k)$ of Program~\eqref{prog:ear_old}, which involves $k \in \NP$ items in both mechanisms, is at most the revenue gap between the $k$-item {\ExAnteRelaxation} and the $1$-item {\AnonymousPricing}. The later revenue gap is given by the next mathematical program.
\begin{align}
    & \sup
    && \sum_{j \in [n]} F_{j}^{-1}(1 - q'_{j}) \cdot q'_{j}
    \label{prog:ear_fake} \tag{P4} \\
    & \mbox{\,\,s.t.}
    && p \cdot (1 - D_{1}(p)) ~ \leq ~ 1,
    && \forall p \in \RP,
    \notag \\
    &&& \sum_{j \in [n]} q'_{j} \leq k,
    \notag \\
    & && \vec{q}' = \{q'_{j}\}_{j \in [n]} \in [0, 1]^{n}, ~ \vec{F} = \{F_{j}\}_{j \in [n]} \subseteq \reg,
    && \forall n \in \NP.
    \notag
\end{align}

The only difference between Program~\eqref{prog:ear_fake} and the one in \cite[Section~4]{AHNPY19} is the constraint $\sum_{j \in [n]} q'_{j} \leq k$ (rather than $\leq 1$). We can resolve Program~\eqref{prog:ear_fake} by following the exactly same steps as in \cite[Section~4]{AHNPY19}. By doing so, we will get
\begin{align*}
    \Re_{\ear/\ap}(k)
    ~ \leq ~ \mbox{optimal value of } \eqref{prog:ear_fake}
    ~ = ~ 1 + \V(\Q^{-1}(k)),
\end{align*}
where the functions $\V(p) \eqdef p \cdot \ln(\frac{p^{2}}{p^{2} - 1})$ and $\Q(p) \eqdef \ln(\frac{p^{2}}{p^{2} - 1}) - \frac{1}{2} \cdot \sum_{t = 1}^{\infty} t^{-2} \cdot p^{-2t}$. Then we can derive \Cref{thm:ear_ap} in the case $k \in \{1, 2, 3\}$ via numeric calculation, as the next table shows.

\begin{longtable}{|c|c|c|c|}
    \hline
    $k$ & $1$ & $2$ & $3$ \\
    \hline
    $1 + \V(\Q^{-1}(k))$ & $\approx 2.7184$ & $\approx 3.7897$ & $\approx 4.8111$ \\
    \hline
\end{longtable}

\paragraph{Acknowledgements.}
We would like to thank Xi Chen, Eric Neyman, Tim Roughgarden, and Rocco Servedio for helpful comments on an earlier version of this work.

\bibliographystyle{alpha}
\bibliography{main}

\newcommand{\etalchar}[1]{$^{#1}$}
\begin{thebibliography}{JLQ{\etalchar{+}}19b}

\bibitem[ACK18]{ACK18}
Yossi Azar, Ashish Chiplunkar, and Haim Kaplan.
\newblock Prophet secretary: Surpassing the 1-1/e barrier.
\newblock In {\'{E}}va Tardos, Edith Elkind, and Rakesh Vohra, editors, {\em
  Proceedings of the 2018 {ACM} Conference on Economics and Computation,
  Ithaca, NY, USA, June 18-22, 2018}, pages 303--318. {ACM}, 2018.

\bibitem[AGN14]{AGN14}
Nima Anari, Gagan Goel, and Afshin Nikzad.
\newblock Mechanism design for crowdsourcing: An optimal 1-1/e competitive
  budget-feasible mechanism for large markets.
\newblock In {\em 55th {IEEE} Annual Symposium on Foundations of Computer
  Science, {FOCS} 2014, Philadelphia, PA, USA, October 18-21, 2014}, pages
  266--275. {IEEE} Computer Society, 2014.

\bibitem[AHN{\etalchar{+}}19]{AHNPY19}
Saeed Alaei, Jason~D. Hartline, Rad Niazadeh, Emmanouil Pountourakis, and Yang
  Yuan.
\newblock Optimal auctions vs. anonymous pricing.
\newblock {\em Games and Economic Behavior}, 118:494--510, 2019.

\bibitem[Ala14]{A14}
Saeed Alaei.
\newblock Bayesian combinatorial auctions: Expanding single buyer mechanisms to
  many buyers.
\newblock {\em {SIAM} J. Comput.}, 43(2):930--972, 2014.

\bibitem[AW18]{AWT18}
Marek Adamczyk and Michal Wlodarczyk.
\newblock Random order contention resolution schemes.
\newblock In Mikkel Thorup, editor, {\em 59th {IEEE} Annual Symposium on
  Foundations of Computer Science, {FOCS} 2018, Paris, France, October 7-9,
  2018}, pages 790--801. {IEEE} Computer Society, 2018.

\bibitem[BGL{\etalchar{+}}18]{BGLPS18}
Hedyeh Beyhaghi, Negin Golrezaei, Renato~Paes Leme, Martin Pal, and
  Balasubramanian Sivan.
\newblock Improved approximations for free-order prophets and second-price
  auctions.
\newblock {\em CoRR}, abs/1807.03435, 2018.

\bibitem[BHW02]{BHW02}
Ziv Bar{-}Yossef, Kirsten Hildrum, and Felix Wu.
\newblock Incentive-compatible online auctions for digital goods.
\newblock In {\em Proceedings of the Thirteenth Annual {ACM-SIAM} Symposium on
  Discrete Algorithms, January 6-8, 2002, San Francisco, CA, {USA.}}, pages
  964--970, 2002.

\bibitem[BK{\etalchar{+}}96]{BK96}
Jeremy Bulow, Paul Klemperer, et~al.
\newblock Auctions versus negotiations.
\newblock {\em American Economic Review}, 86(1):180--194, 1996.

\bibitem[CD17]{CD17}
Yang Cai and Constantinos Daskalakis.
\newblock Learning multi-item auctions with (or without) samples.
\newblock In {\em 58th {IEEE} Annual Symposium on Foundations of Computer
  Science, {FOCS} 2017, Berkeley, CA, USA, October 15-17, 2017}, pages
  516--527, 2017.

\bibitem[CFH{\etalchar{+}}17]{CFHOV17}
Jos{\'{e}}~R. Correa, Patricio Foncea, Ruben Hoeksma, Tim Oosterwijk, and Tjark
  Vredeveld.
\newblock Posted price mechanisms for a random stream of customers.
\newblock In {\em Proceedings of the 2017 {ACM} Conference on Economics and
  Computation, {EC} '17, Cambridge, MA, USA, June 26-30, 2017}, pages 169--186,
  2017.

\bibitem[CFH{\etalchar{+}}18]{CFHOV18}
Jos{\'{e}}~R. Correa, Patricio Foncea, Ruben Hoeksma, Tim Oosterwijk, and Tjark
  Vredeveld.
\newblock Recent developments in prophet inequalities.
\newblock {\em SIGecom Exch.}, 17(1):61--70, 2018.

\bibitem[CFPV19]{CFPV19}
Jos{\'{e}}~R. Correa, Patricio Foncea, Dana Pizarro, and Victor Verdugo.
\newblock From pricing to prophets, and back!
\newblock {\em Oper. Res. Lett.}, 47(1):25--29, 2019.

\bibitem[CGL11]{CGL11}
Ning Chen, Nick Gravin, and Pinyan Lu.
\newblock On the approximability of budget feasible mechanisms.
\newblock In Dana Randall, editor, {\em Proceedings of the Twenty-Second Annual
  {ACM-SIAM} Symposium on Discrete Algorithms, {SODA} 2011, San Francisco,
  California, USA, January 23-25, 2011}, pages 685--699. {SIAM}, 2011.

\bibitem[CGL14]{CGL14}
Ning Chen, Nick Gravin, and Pinyan Lu.
\newblock Optimal competitive auctions.
\newblock In David~B. Shmoys, editor, {\em Symposium on Theory of Computing,
  {STOC} 2014, New York, NY, USA, May 31 - June 03, 2014}, pages 253--262.
  {ACM}, 2014.

\bibitem[CGL15]{CGL15}
Ning Chen, Nikolai Gravin, and Pinyan Lu.
\newblock Competitive analysis via benchmark decomposition.
\newblock In Tim Roughgarden, Michal Feldman, and Michael Schwarz, editors,
  {\em Proceedings of the Sixteenth {ACM} Conference on Economics and
  Computation, {EC} '15, Portland, OR, USA, June 15-19, 2015}, pages 363--376.
  {ACM}, 2015.

\bibitem[CGM15]{CGM15}
Nicol{\`{o}} Cesa{-}Bianchi, Claudio Gentile, and Yishay Mansour.
\newblock Regret minimization for reserve prices in second-price auctions.
\newblock {\em {IEEE} Trans. Information Theory}, 61(1):549--564, 2015.

\bibitem[CHMS10]{CHMS10}
Shuchi Chawla, Jason~D. Hartline, David~L. Malec, and Balasubramanian Sivan.
\newblock Multi-parameter mechanism design and sequential posted pricing.
\newblock In {\em Proceedings of the 42nd {ACM} Symposium on Theory of
  Computing, {STOC} 2010, Cambridge, Massachusetts, USA, 5-8 June 2010}, pages
  311--320, 2010.

\bibitem[CSZ19]{CSZ19}
Jos{\'{e}}~R. Correa, Raimundo Saona, and Bruno Ziliotto.
\newblock Prophet secretary through blind strategies.
\newblock In {\em Proceedings of the Thirtieth Annual {ACM-SIAM} Symposium on
  Discrete Algorithms, {SODA} 2019, San Diego, California, USA, January 6-9,
  2019}, pages 1946--1961, 2019.

\bibitem[DFK16]{DFK16}
Paul D{\"{u}}tting, Felix~A. Fischer, and Max Klimm.
\newblock Revenue gaps for static and dynamic posted pricing of homogeneous
  goods.
\newblock {\em CoRR}, abs/1607.07105, 2016.

\bibitem[EHKS18]{EHKS18}
Soheil Ehsani, MohammadTaghi Hajiaghayi, Thomas Kesselheim, and Sahil Singla.
\newblock Prophet secretary for combinatorial auctions and matroids.
\newblock In Artur Czumaj, editor, {\em Proceedings of the Twenty-Ninth Annual
  {ACM-SIAM} Symposium on Discrete Algorithms, {SODA} 2018, New Orleans, LA,
  USA, January 7-10, 2018}, pages 700--714. {SIAM}, 2018.

\bibitem[FILS15]{FILS15}
Hu~Fu, Nicole Immorlica, Brendan Lucier, and Philipp Strack.
\newblock Randomization beats second price as a prior-independent auction.
\newblock In {\em Proceedings of the Sixteenth {ACM} Conference on Economics
  and Computation, {EC} '15, Portland, OR, USA, June 15-19, 2015}, page 323,
  2015.

\bibitem[GHK{\etalchar{+}}05]{GHKKKM05}
Venkatesan Guruswami, Jason~D. Hartline, Anna~R. Karlin, David Kempe, Claire
  Kenyon, and Frank McSherry.
\newblock On profit-maximizing envy-free pricing.
\newblock In {\em Proceedings of the Sixteenth Annual {ACM-SIAM} Symposium on
  Discrete Algorithms, {SODA} 2005, Vancouver, British Columbia, Canada,
  January 23-25, 2005}, pages 1164--1173. {SIAM}, 2005.

\bibitem[GHW01]{GHW01}
Andrew~V. Goldberg, Jason~D. Hartline, and Andrew Wright.
\newblock Competitive auctions and digital goods.
\newblock In S.~Rao Kosaraju, editor, {\em Proceedings of the Twelfth Annual
  Symposium on Discrete Algorithms, January 7-9, 2001, Washington, DC, {USA}},
  pages 735--744. {ACM/SIAM}, 2001.

\bibitem[GJLZ20]{GJLZ19}
Nick Gravin, Yaonan Jin, Pinyan Lu, and Chenhao Zhang.
\newblock Optimal budget-feasible mechanisms for additive valuations.
\newblock {\em ACM Transactions on Economics and Computation (TEAC)},
  8(4):1--15, 2020.

\bibitem[GZ18]{GZ18}
Yiannis Giannakopoulos and Keyu Zhu.
\newblock Optimal pricing for {MHR} distributions.
\newblock In {\em Web and Internet Economics - 14th International Conference,
  {WINE} 2018, Oxford, UK, December 15-17, 2018, Proceedings}, pages 154--167,
  2018.

\bibitem[Har13]{H13}
Jason~D Hartline.
\newblock Mechanism design and approximation.
\newblock {\em Book draft. October}, 122, 2013.

\bibitem[HKS07]{HKS07}
Mohammad~Taghi Hajiaghayi, Robert~D. Kleinberg, and Tuomas Sandholm.
\newblock Automated online mechanism design and prophet inequalities.
\newblock In {\em Proceedings of the Twenty-Second {AAAI} Conference on
  Artificial Intelligence, July 22-26, 2007, Vancouver, British Columbia,
  Canada}, pages 58--65, 2007.

\bibitem[HR09]{HR09}
Jason~D. Hartline and Tim Roughgarden.
\newblock Simple versus optimal mechanisms.
\newblock In {\em Proceedings 10th {ACM} Conference on Electronic Commerce
  (EC-2009), Stanford, California, USA, July 6--10, 2009}, pages 225--234,
  2009.

\bibitem[JJLZ21]{JJLZ21}
Yaonan Jin, Shunhua Jiang, Pinyan Lu, and Hengjie Zhang.
\newblock Tight revenue gaps among multi-unit mechanisms.
\newblock In P{\'{e}}ter Bir{\'{o}}, Shuchi Chawla, and Federico Echenique,
  editors, {\em {EC} '21: The 22nd {ACM} Conference on Economics and
  Computation, Budapest, Hungary, July 18-23, 2021}, pages 654--673. {ACM},
  2021.

\bibitem[JKM13]{JKM13}
Oliver Johnson, Ioannis Kontoyiannis, and Mokshay~M. Madiman.
\newblock Log-concavity, ultra-log-concavity, and a maximum entropy property of
  discrete compound poisson measures.
\newblock {\em Discrete Applied Mathematics}, 161(9):1232--1250, 2013.

\bibitem[JLQ19a]{JLQ19}
Yaonan Jin, Weian Li, and Qi~Qi.
\newblock On the approximability of simple mechanisms for {MHR} distributions.
\newblock In {\em Web and Internet Economics - 15th International Conference,
  {WINE} 2019, New York, NY, USA, December 10-12, 2019, Proceedings}, pages
  228--240, 2019.

\bibitem[JLQ{\etalchar{+}}19b]{JLQTX19a}
Yaonan Jin, Pinyan Lu, Qi~Qi, Zhihao~Gavin Tang, and Tao Xiao.
\newblock Tight approximation ratio of anonymous pricing.
\newblock In {\em Proceedings of the 51st Annual {ACM} {SIGACT} Symposium on
  Theory of Computing, {STOC} 2019, Phoenix, AZ, USA, June 23-26, 2019.}, pages
  674--685, 2019.

\bibitem[JLQ{\etalchar{+}}19c]{JLQTX19b}
Yaonan Jin, Pinyan Lu, Qi~Qi, Zhihao~Gavin Tang, and Tao Xiao.
\newblock Tight revenue gaps among simple and optimal mechanisms.
\newblock {\em SIGecom Exch.}, 17(2):54--61, 2019.

\bibitem[JLTX20]{JLTX19}
Yaonan Jin, Pinyan Lu, Zhihao~Gavin Tang, and Tao Xiao.
\newblock Tight revenue gaps among simple mechanisms.
\newblock {\em SIAM Journal on Computing}, 49(5):927--958, 2020.

\bibitem[JLX19]{JLX19}
Yaonan Jin, Pinyan Lu, and Tao Xiao.
\newblock Learning reserve prices in second-price auctions.
\newblock {\em CoRR}, abs/1912.10069, 2019.

\bibitem[Luc17]{L17}
Brendan Lucier.
\newblock An economic view of prophet inequalities.
\newblock {\em SIGecom Exchanges}, 16(1):24--47, 2017.

\bibitem[MM16]{MM16}
Mehryar Mohri and Andres~Mu{\~{n}}oz Medina.
\newblock Learning algorithms for second-price auctions with reserve.
\newblock {\em Journal of Machine Learning Research}, 17:74:1--74:25, 2016.

\bibitem[MR16]{MR16}
Jamie Morgenstern and Tim Roughgarden.
\newblock Learning simple auctions.
\newblock In Vitaly Feldman, Alexander Rakhlin, and Ohad Shamir, editors, {\em
  Proceedings of the 29th Conference on Learning Theory, {COLT} 2016, New York,
  USA, June 23-26, 2016}, volume~49 of {\em {JMLR} Workshop and Conference
  Proceedings}, pages 1298--1318. JMLR.org, 2016.

\bibitem[MS20]{MS19}
Will Ma and Balasubramanian Sivan.
\newblock Separation between second price auctions with personalized reserves
  and the revenue optimal auction.
\newblock {\em Oper. Res. Lett.}, 48(2):176--179, 2020.

\bibitem[MSVV07]{MSVV07}
Aranyak Mehta, Amin Saberi, Umesh~V. Vazirani, and Vijay~V. Vazirani.
\newblock Adwords and generalized online matching.
\newblock {\em J. {ACM}}, 54(5):22, 2007.

\bibitem[Mye81]{M81}
Roger~B. Myerson.
\newblock Optimal auction design.
\newblock {\em Math. Oper. Res.}, 6(1):58--73, 1981.

\bibitem[OLBC10]{OLBC10}
Frank~WJ Olver, Daniel~W Lozier, Ronald~F Boisvert, and Charles~W Clark.
\newblock {\em NIST handbook of mathematical functions hardback and CD-ROM}.
\newblock Cambridge university press, 2010.

\bibitem[Oxl06]{O06}
James~G Oxley.
\newblock {\em Matroid theory}, volume~3.
\newblock Oxford University Press, USA, 2006.

\bibitem[Rob55]{R55}
Herbert Robbins.
\newblock A remark on stirling's formula.
\newblock {\em The American mathematical monthly}, 62(1):26--29, 1955.

\bibitem[Sin10]{S10}
Yaron Singer.
\newblock Budget feasible mechanisms.
\newblock In {\em 51th Annual {IEEE} Symposium on Foundations of Computer
  Science, {FOCS} 2010, October 23-26, 2010, Las Vegas, Nevada, {USA}}, pages
  765--774. {IEEE} Computer Society, 2010.

\bibitem[Yan11]{Y11}
Qiqi Yan.
\newblock Mechanism design via correlation gap.
\newblock In {\em Proceedings of the Twenty-Second Annual {ACM-SIAM} Symposium
  on Discrete Algorithms, {SODA} 2011, San Francisco, California, USA, January
  23-25, 2011}, pages 710--719, 2011.

\end{thebibliography}

\appendix
\renewcommand{\appendixname}{Appendix~\Alph{section}}

\section{\texorpdfstring{Asymptotic Formulas for the Revenue Gap $\Re_{\ar/\ap}(k)$}{}}
\label{sec:asymptotic_formula}

\begin{claim}[Part~3 of \Cref{thm:ar_ap:uniform}]
\label{cla:ar_ap:uniform:3}
For each $k \in \NP$, the supremum revenue gap
\[
    \Re_{\ar / \ap}(k) = 1 + k \cdot \int_{0}^{\infty} \frac{T_k(x) \cdot (1 - T_{k + 1}(x))}{(k - \sum_{i \in [k]} T_{i}(x))^{2}} \cdot \d x,
\]
where the functions $T_i(x) \eqdef e^{-x} \cdot \sum_{t \in [0: i - 1]} \frac{1}{t!} \cdot x^t$ for all $i \in [k + 1]$,
is bounded between
\[
    1 + \frac{0.1}{\sqrt{k}}
    ~ \leq ~ \Re_{\ar / \ap}(k)
    ~ \leq ~ 1 + \frac{2}{\sqrt{k}}.
\]
The value of $\Re_{\ar / \ap}(k)$ when $k\leq 24$ is listed in Table~\ref{tab:re_ar_ap_k_25}.
\end{claim}

\begin{table}[t]
    \centering
    \begin{tabular}{|c|c|c|c|c|c|c|c|c|}
    \hline
    $k$ & $1$ & $2$ & $3$ & $4$ & $5$ & $6$ & $7$ & $8$  \\
    \hline
    $\Re_{\ar / \ap}(k)$ & $\pi^2/6$ & $1.4445$ & $1.3575$ & $1.3065$ & $1.2721$ & $1.2470$&$1.2276$ &$ 1.2121$ \\
    \hline
    $c_k$ & $0.6449$ & $0.6287$ & $0.6192$ & $0.6130$ & $0.6085$& $0.6050$& $0.6023$& $0.6000$\\
    \hline
    \hline
    $k$ & $9$ & $10$ & $11$ &$12$ & $13$ & $14$ & $15$& $16$  \\
    \hline
    $\Re_{\ar / \ap}(k)$ & $1.1994$&$1.1886$&$1.1794$&$1.1714$&$1.1644$&$1.1581$&$1.1525$&$1.1475$\\
    \hline
    $c_k$ & $0.5982$&$0.5965$&$0.5951$&$0.5939$&$0.5928$&$0.5918$&$0.5909$&$0.5901$\\
    \hline
    \hline
    $k$ & $17$ & $18$ & $19$ &$20$ & $21$ & $22$ & $23$& $24$  \\
    \hline
    $\Re_{\ar / \ap}(k)$&$1.1429$&$1.1387$&$1.1349$&$1.1313$&$1.1281$&$1.1250$&$1.1221$&$1.1195$\\
    \hline
    $c_k$&$0.5894$&$0.5887$&$0.5881$&$0.5875$&$0.5878$&$0.5865$&$0.5860$&$0.5855$\\
    \hline
    \end{tabular}
    \caption{List of $\Re_{\ar / \ap}(k)$ for $k \leq 24$, where $c_k$ means that $\Re_{\ar / \ap}(k) = 1 + c_k/\sqrt{k}$.}
    \label{tab:re_ar_ap_k_25}
\end{table}

\begin{proof}[Proof of \Cref{cla:ar_ap:uniform:3} (Lower Bound)]
We define $lb(k)\eqdef 1 + k \cdot \int_{0}^{\infty} T_k(x) \cdot (1 - T_{k + 1}(x))\cdot k^{-2} \cdot \d x$. Apparently $lb(k) \leq \Re_{\ar / \ap}(k)$  for all $k\geq 1$. Then we have
\begin{eqnarray*}
    lb(k)
    & = & 1 + \frac{1}{k} \int_{0}^{\infty} T_k(x) \cdot (1 - T_{k + 1}(x)) \cdot \d x\\
    & = & 1 +  \frac{1}{k} \int_{0}^{\infty} \sum_{i=0}^{k-1}e^{-x}x^i/i! \cdot (1 - \sum_{i=0}^{k}e^{-x}x^i/i!) \cdot \d x\\
    & = & 1 + \Big(\frac{1}{k}\int_{0}^{\infty} \sum_{i=0}^{k-1}e^{-x}x^i/i! \cdot \d x \Big)- \Big( \frac{1}{k}
    \int_{0}^{\infty} \sum_{i=0}^{k-1}\sum_{j=0}^{k}e^{-2x}x^{i+j}/i!/j! \d x \Big)\\
    & = & 2-\frac{1}{k} \int_{0}^{\infty} \sum_{i=0}^{k-1}\sum_{j=0}^{k}e^{-2x}x^{i+j}/i!/j! \d x \\
    & = & 2-\frac{1}{2k} \int_{0}^{\infty} \sum_{i=0}^{k-1}\sum_{j=0}^{k}e^{-x}x^{i+j}/i!/j!/2^{i+j} \d x \\
    & = & 2 -\frac{1}{2k} \sum_{i=0}^{k-1}\sum_{j=0}^{k} \binom{i+j}{i}/2^{i+j}\\
    & = & 2 -\frac{1}{2k} \left (\sum_{i=0}^{2k-1}\sum_{j=0}^i\binom{i}{j}/2^i - \sum_{\substack{i+j\leq 2k-1\\i\geq k~\text{or}~j>k}}\binom{i+j}{i}/2^{i+j}\right)  \\
    & = & 1 + \frac{2}{2k}\cdot \sum_{i+j\leq k} \binom{i+j+k}{j}/2^{i+j+k}\\
    & = & 1+\frac{1}{k}\cdot \sum_{m=k}^{2k} g(m)/2^m,
\end{eqnarray*}
where the second step is by definition of $T_k(x)$, the third step is by Fact~\ref{fac:integral_gamma} that $
\int_0^{\infty}e^{-x}x^n \d x = n!$, the fifth step is by substitution, the sixth step is by \Cref{fac:integral_gamma}, and in the last step we define $g(m)\eqdef \sum_{i=0}^{m-k} \binom{m}{i}$.

For all $m\in \{\lceil 2k-\sqrt{k}/2\rceil, \lceil 2k-\sqrt{k}/2\rceil +1 , \cdots,2k\}$,  
\begin{eqnarray*}
    g(m) 
    & = & \sum_{i=0}^{m-k} \binom{m}{i} \\
    & = & \sum_{i=0}^{\lceil m/2\rceil} \binom{m}{i} - \sum_{i=m-k+1}^{\lceil m/2\rceil} \binom{m}{i}\\
    & \geq & 2^m/2 - (\lceil m/2\rceil - (m-k))\cdot \binom{m}{\lceil m/2\rceil}\\
    & \geq & \frac{2^m}{2} - \frac{\sqrt{k}}{4}\cdot \frac{2\cdot 2^m}{\sqrt{\pi m}}\\
    & \geq & \frac{2^m}{2} - \frac{2^m}{2\sqrt{\pi}}\\
    & \geq & 2^m/5,
\end{eqnarray*}
where the fourth step is by \Cref{fac:binomial_bound}.

Therefore,
\begin{eqnarray*}
    \Re_{\ar / \ap}(k)
    & \geq & lb(k) \\
    & = & 1+\frac{1}{k}\cdot \sum_{m=k}^{2k} g(m)/2^m\\
    & \geq & 1+\frac{1}{k} \sum_{m=\lceil 2k-\sqrt{k}/2\rceil}^{2k} g(m)/2^m\\
    & \geq & 1/(10\sqrt{k}).
\end{eqnarray*}

This accomplishes the lower-bound part of \Cref{cla:ar_ap:uniform:3}.
\end{proof}

\begin{proof}[Proof of \Cref{cla:ar_ap:uniform:3} (Upper Bound)]
We define 
\begin{align*}
a(x) \eqdef \sum_{i=0}^{k-1} \frac{x^i}{i!}; ~~~~~~~~
b(x) \eqdef \sum_{i=k+1}^{\infty} \frac{x^i}{i!}.
\end{align*}

Then the integral part of $\Re_{\ar / \ap}(k)$ can be written as
\begin{eqnarray*}
    h(x)
    & \eqdef & ~ T_k(x) \cdot (1 - T_{k + 1}(x)) \cdot (k - \sum_{i \in [k]} T_{i}(x))^{-2}\\
    & = & a(x) b(x) \cdot \big(\sum_{i \in [k]} ( 1 - T_{i}(x) ) \big)^{-2}\\
    & = & a(x) b(x) \cdot \big(\sum_{i \in [k]} \sum_{j=i}^{\infty}\frac{x^j}{j!}\big)^{-2}\\
    & = & a(x) b(x) \cdot \big(\sum_{i \in [k]} i\cdot \frac{x^i}{i!} + k\cdot \sum_{i=k+1}^{\infty}\frac{x^i}{i!} \big)^{-2}\\
    & = & a(x) b(x) \cdot (x\cdot a(x) + k \cdot b(x))^{-2}.
\end{eqnarray*}

Then 
\begin{eqnarray*}
    \Re_{\ar / \ap}(k)
    & = & 1 + k\cdot \int_{0}^{\infty} h(x) \d x \\
    & = & 1 + k\cdot \Big( \underbrace{\int_{0}^{k-\sqrt{6k}} h(x) \d x}_{h_1} + \underbrace{\int_{k-\sqrt{6k}}^{k} h(x) \d x}_{h_2} + \underbrace{\int_{k}^{\infty} h(x) \d x}_{h_3}\Big).
\end{eqnarray*}

We are going to upper bound $h_1$, $h_2$, $h_3$ separately.

\noindent {\bf Case 1: Bound $h_1$}

\begin{align}\label{eq:h1_equal_int_G}
    h_1 ~ = & ~ \int_{0}^{k-\sqrt{6k}} h(x) \d x \notag\\
    ~ \leq & ~ \int_{0}^{k-\sqrt{6k}} \frac{b(x)}{x^2 \cdot e^x} \d x \notag\\
    ~ = & ~ \int_{0}^{k-\sqrt{6k}} \sum_{i=k+1}^{\infty}\frac{x^{i-2}}{i!}e^{-x} \d x  \notag\\ 
    ~ = & ~ \int_{\sqrt{6k}}^{k}\underbrace{ \sum_{i=k+1}^{\infty}\frac{ (k-x)^{i-2} }{ i! } e^{-(k-x)} }_{G(x)} \d x,
\end{align}
where the second step is by Part~\ref{fac:small_x_a_k_b_e_x} of \Cref{fac:a_b_facts}, the third step is by definition of $b(x)$. 

We define $G(x)\eqdef \sum_{i=k+1}^{\infty}\frac{(k-x)^{i-2}}{i!}e^{-(k-x)}$. So 
\begin{align*}
G(0)
= & ~ \sum_{i=k+1}^{\infty}\frac{k^{i-2}}{i!}e^{-k}\\
= & ~ \frac{1}{k^2}\sum_{i=k+1}^{\infty}\frac{k^i}{i!}e^{-k}\\
= & ~ \frac{1}{k^2}(1-\Gamma(1+k,k)/\Gamma(1+k))\\
\leq & ~ \frac{1}{2k^2},
\end{align*}
where the last step is by Lemma~\ref{lem:incomplete_gamma_function} that $\Gamma(1+k,k)/k!>1/2$.

For the ease of the proof, we define $d(i,x)=\frac{(k-x)^{i-2}}{i!}e^{-(k-x)}$. So that $G(x) = \sum_{i=k+1}^{\infty} d(i,x)$.

we have $\forall i\geq k+1$, 
\begin{align*}
    \ln (d(i,x)/d(i,0))
    ~ = & ~ \ln (\frac{ (k-x)^{i-2}e^{-(k-x)}/i! }{ k^{i-2}e^{-k}/i !} )\\
    ~ = & ~ \ln( (1-x/k)^{i-2}e^x )\\
    ~ = & ~ x+(i-2)\ln(1-x/k)\\
    ~ = & ~ x - (i-2) \sum_{j=1}^{\infty} \frac{1}{j}(\frac{x}{k})^j\\
    ~ = & ~ x - i\sum_{j=1}^{\infty} \frac{1}{j}(\frac{x}{k})^j + 2\sum_{j=1}^{\infty} \frac{1}{j}(\frac{x}{k})^j\\
    ~ \leq & ~ - i\sum_{j=2}^{\infty} \frac{1}{j}(\frac{x}{k})^j + 2\sum_{j=1}^{\infty} \frac{1}{j}(\frac{x}{k})^j\\
    ~ = & ~ - \frac{i}{2}(\frac{x}{k})^2 - \sum_{j=1}^{\infty} \left( \frac{i}{j+2}\frac{x^2}{k^2} - \frac{2}{j}\right)(\frac{x}{k})^j\\
    ~ \leq & ~ - \frac{i}{2}(\frac{x}{k})^2\\
    ~ \leq & ~ - \frac{1}{2}\frac{x^2}{k},
\end{align*}
where the sixth step follows from $i\geq k$, and the eighth step follows from $x\geq \sqrt{6k}$. So $d(i,x)\leq e^{-\frac{x^2}{2k}}d(i,0)$.

Therefore,
\begin{align*}
    G(x)
    ~ = & ~ \sum_{i=k+1}^{\infty} d(i,x)\\
    ~ \leq & ~  \sum_{i=k+1}^{\infty} e^{-\frac{x^2}{2k}}d(i,0)\\
    ~ \leq & ~ e^{-\frac{x^2}{2k}} G(0).
\end{align*}

Thus, we have 
\begin{align*}
    h_1 
    ~ \leq & ~ \int_{\sqrt{6k}}^{k} G(x) \d x\\
    ~ \leq & ~ \int_{\sqrt{6k}}^{k} e^{-\frac{x^2}{2k}} G(0) \d x\\
    ~ \leq & ~ \sqrt{k}G(0)\int_{\sqrt{6}}^{\sqrt{k}} e^{-\frac{x^2}{2}} \d x\\
    ~ \leq & ~ \sqrt{k}G(0)\int_{\sqrt{6}}^{\infty} e^{-\frac{x^2}{2}} \d x\\
    ~ \leq & ~0.018\sqrt{k}G(0)\\
    ~ \leq & ~ 0.009 k^{-1.5},
\end{align*}
where the first step is by Eq.\eqref{eq:h1_equal_int_G},  the fifth step is by $\int_{\sqrt{6}}^{\infty} e^{-\frac{x^2}{2}} \d x \leq 0.018$, the last step is by $G(0)\leq k^{-2}/2$.

{\bf Case 2, Bound $h_2$}

\begin{eqnarray*}
    h_2
    & = & ~ \int_{k - \sqrt{6k}}^{k} \frac{a(x)b(x)}{(x\cdot a(x)+k\cdot b(x))^2} \d x\\
    & \leq & ~ \int_{k - \sqrt{6k}}^{k} \frac{1}{4kx}\d x \\
    & = & ~ \ln( k/(k-\sqrt{6k}) )/(4k)\\
    & \leq & ~ ( k/(k-\sqrt{6k}) -1 ) /(4k) \\
    & = & ~ (\sqrt{6}/4) \cdot k^{-0.5} /(k-\sqrt{6k})\\
    & \leq & ~ (2\sqrt{6}/4)\cdot k^{-1.5},
\end{eqnarray*}
where the first step follows from $(a+b)^2\geq 4ab$, the last step is by $k-\sqrt{6k}\geq 0.5k$ for $k\geq 24$.

{\bf Case 3, Bound $h_3$.}

We have
\begin{eqnarray*}
    h_3
    & = & ~ \int_{k}^{\infty}\frac{a(x)b(x)}{(xa(x) + kb(x))^2}\d x\\
    & \leq & ~ k^{-2} \cdot \int_{k}^{\infty} a(x)e^{-x}\d x\\
    & \leq & ~ a(k) \cdot k^{-2} \int_{k}^{\infty} \frac{x^{k-1}e^{-x}}{k^{k-1}}\d x\\
    & \leq & ~ \frac{e^k k!}{2k^2 \cdot k^k} \int_{k}^{\infty} x^{k-1}e^{-x}/(k-1)!\d x\\
    & = & ~ \frac{e^k k!}{2k^2 \cdot k^k} \big( \Gamma(k,k)/(k-1)! \big)\\
    & \leq & ~ \frac{e^k k!}{2k^2 \cdot k^k} \cdot (1/2)\\
    & \leq & ~ \frac{e^k}{4k^2 \cdot k^k} \cdot (e k^{k+1/2} e^{-k})\\
    & = & ~ e/4\cdot k^{-1.5},
\end{eqnarray*}
where the second step is by Part~\ref{fac:large_x_a_k_b_e_x} of \Cref{fac:a_b_facts} that for all $x\geq k + c_2 \sqrt{k}$, $b(x)/(xa(x)+kb(x))^2 \leq e^{-x}/k^2$, the third step is by Part~\ref{fac:a_k_a_x} of \Cref{fac:a_b_facts}, the fourth step is by Part~\ref{fac:a_x_is_small} of \Cref{fac:a_b_facts} and that $\forall x,~a(x)+b(x)\leq e^x$,
the fifth step is by definition of incomplete gamma function, the sixth step is by \Cref{lem:incomplete_gamma_function}. the last step is by \Cref{lem:striling}.

Therefore, we can upper bound 
\begin{eqnarray*}
    \Re_{\ar / \ap}(k) 
    & = & 1 + k\cdot \int_{0}^{\infty} h(x) \d x \\
    & = & 1 + k\cdot(h_1 + h_2 + h_3) \\
    & \leq & 1 + (0.009 + 2\sqrt{6}/4 + e/4)/\sqrt{k} \\
    & \leq & ~ 1 + 2/\sqrt{k}.
\end{eqnarray*}

This accomplishes the upper-bound part of \Cref{cla:ar_ap:uniform:3}.
\end{proof}

\section{Mathematical Tools}

\begin{lemma}[Incomplete gamma function {\cite[Chapter~8]{OLBC10}}]
\label{lem:incomplete_gamma_function}
Define the incomplete gamma function $\Gamma(n,x) \eqdef \int_{x}^{\infty} t^{n-1}e^{-t}\d t$. Then for all positive integer $n$, we have
\begin{enumerate}[font = {\em\bfseries}]
    \item $\frac{\Gamma(n,n)}{(n-1)!} < \frac{1}{2} < \frac{\Gamma(n,n-1)}{(n-1)!}$.
    
    \item $\frac{\Gamma(n,x)}{(n-1)!} = e^{-x}\sum_{i=0}^{n-1}\frac{x^i}{i!}$.
\end{enumerate}
\end{lemma}

\begin{lemma}[Stirling's approximation \cite{R55}]
\label{lem:striling}
For all positive $n$, the following holds:
\begin{align*}
    \sqrt{2\pi} n^{n+1/2}e^{-n} \leq n! \leq en^{n+ 1/2} e^{-n}.
\end{align*}
\end{lemma}

It's easy to see the following facts:
\begin{fact}
\label{fac:binomial_bound}
For all positive integer $n$, we have 
\begin{align*}
    \frac{1}{2}\frac{4^n}{\sqrt{\pi n}} \leq \binom{2n}{n} \leq \frac{4^n}{\sqrt{\pi n}}
\end{align*}
\end{fact}

\begin{fact}
\label{fac:integral_gamma}
For all positive integer $n$, we have 
\begin{align*}
\int_0^{\infty}e^{-x}x^n \d x=\Gamma(n+1) = n!.
\end{align*}
\end{fact}

\begin{fact}
\label{fac:a_b_facts}
$a(x)$ and $b(x)$ satisfies the following facts:
\begin{enumerate}[font = {\em\bfseries}]
    \item \label{fac:b_x_is_small} 
        For all $x\in [0,k]$, $b(x)\cdot e^{-x} \leq 1/2$,
    \item \label{fac:a_x_is_small} 
        For all $x \in [k,\infty)$, $a(x) \cdot e^{-x} \leq 1/2$,
    \item \label{fac:small_x_a_k_b_e_x} 
        For all $x \in [0, k - 1]$, $$\frac{a(x)}{(xa(x) + kb(x))^2} \leq \frac{1}{x^2e^x},$$
    \item \label{fac:large_x_a_k_b_e_x} 
        For all $x \in [k,\infty)$, $$\frac{b(x)}{(xa(x) + kb(x))^2} \leq \frac{1}{k^2e^x},$$
    \item \label{fac:a_k_a_x}
        For all $x \in [k,\infty)$, $a(k) \cdot x^{k-1} / k^{k-1} \geq a(x).$
\end{enumerate}
\end{fact}
\begin{proof}
{\bf Part 1.}
For all $x\in[0,k]$,
\begin{align*}
    b(x)\cdot e^{-x}
    = & ~ 1 - \Gamma(k+1,x)/k!\\
    \leq & ~ 1-\Gamma(k+1,k)/k!\\
    \leq &~ 1/2,
\end{align*}
where the first step is by Part 2 of \Cref{lem:incomplete_gamma_function}, the second step is because $\Gamma(k+1,x)$ is a decreasing function on $x\in[0,\infty)$, the third step is by Part 1 of \Cref{lem:incomplete_gamma_function}.

{\bf Part 2.}
By the same reason, for all $x\in[k,\infty)$, we have that 
\begin{align*}
    a(x) \cdot e^{-x} 
    = & ~ \Gamma(k,x)/(k-1)!\\
    \leq & ~ \Gamma(k,k)/(k-1)!\\
    \leq & ~ 1/2.
\end{align*}

{\bf Part 3.}
For all $x\in [0,k-1]$, we have 
\begin{align*}
& ~ (x \cdot a(x) + k \cdot b(x))^2 - x^2 a(x) e^x \\
= & ~ x^2 a(x)^2 + 2xk \cdot a(x)b(x) + k^2 b(x)^2 - x^2 a(x) \big( a(x) + b(x) + x^k/k! \big)\\
\geq & ~ 2xk \cdot a(x)b(x) - x^2 a(x) \big( b(x) + x^k/k! \big)\\
\geq & ~ (2k - (k-1))b(x) - x^{k+1}/k!\\
\geq & ~ (k+1) \cdot x^{k+1}/(k+1)! - x^{k+1}/k!\\
= & ~ 0,
\end{align*}
where the fourth step is by $b(x) \geq x^{k+1}/(k+1)!$. Therefore, $\frac{a(x)}{(xa(x) + kb(x))^2} \leq \frac{1}{x^2e^x}$ follows directly.

{\bf Part 4.}
For all $x\in [k,\infty)$, we have 
\begin{align*}
& ~ (x \cdot a(x) + k \cdot b(x))^2 - k^2 b(x) e^x\\
= & ~ x^2 a(x)^2 + 2xk \cdot a(x)b(x) + k^2 b(x)^2 - k^2 b(x) \big( a(x) + b(x) + x^k/k! \big)\\
\geq & ~ 2xk \cdot a(x)b(x)  - k^2 b(x) \big( a(x) + x^k/k! \big)\\
\geq & ~ (2x - k) \cdot a(x) - x^k/(k-1)!\\
\geq & ~ k \cdot x^k/k! - x^k/(k-1)!\\
= & ~ 0,
\end{align*}
where the fourth step is by $a(x)\geq x^k/k!$. Therefore, $\frac{b(x)}{(xa(x) + kb(x))^2} \leq \frac{1}{k^2e^x}$ follows directly.

{\bf Part 5.}
For all $x\geq k$, we have 
\begin{align*}
    a(k) \cdot x^{k-1} / k^{k-1} 
    \geq & ~ \sum_{i=0}^{k-1} \frac{k^i}{i!}\frac{x^{k-1}}{k^{k-1}}\\
    \geq & ~ \sum_{i=0}^{k-1} \frac{k^i}{i!}\frac{x^{i}}{k^{i}}\\
    \geq & ~ \sum_{i=0}^{k-1} \frac{x^i}{i!}\\
    = & ~ a(x).
\end{align*}
\end{proof}

\end{document}